\documentclass[11pt]{article}

\usepackage[margin=1in]{geometry}
\usepackage{setspace}
\usepackage{amsmath,amssymb,amsthm}
\usepackage{graphicx}
\usepackage{natbib}
\usepackage{hyperref}
\usepackage{authblk}
\theoremstyle{definition}
\newtheorem{definition}{Definition S}
\newtheorem{proposition}{Proposition}
\newtheorem{lemma}{Lemma}

\usepackage{tocloft}

\usepackage{etoc}
\usepackage{tikz}
\usepackage{pgfplots}
\usetikzlibrary{arrows.meta,positioning,calc,fit}

\usepackage{booktabs}
\usepackage{xcolor}
\usepgfplotslibrary{groupplots}

\def\wt{\widetilde{\alpha}}
\def\half{\small{1\over2}}

\title{\bf Predicting fixed-sample test decisions enables anytime-valid inference}

\author[1,2]{Chris Holmes\footnote{Corresponding author. Email: cholmes@eit.org; cholmes@stats.ox.ac.uk}}
\author[3]{Stephen G. Walker}
\affil[1]{AI \& Robotics Institute, Ellison Institute of Technology, Oxford, UK}
\affil[2]{Department of Statistics, University of Oxford, UK}
\affil[3]{University of Kent, UK}

\date{}

\begin{document}

\maketitle

\begin{abstract}

Statistical hypothesis tests typically use prespecified sample sizes, yet data often arrive sequentially. Interim analyses invalidate classical error guarantees, while existing sequential methods require rigid testing preschedules or incur substantial losses in statistical power. We describe a simple procedure that transforms any fixed-sample hypothesis test into an anytime-valid test while ensuring Type-I error control and near-optimal power with substantial savings in samples when the null hypothesis is false. At each step, the procedure predicts the probability that a classical test would reject the null hypothesis at its fixed-sample size, treating future observations as missing data under the null hypothesis. Thresholding this probability yields an anytime-valid stopping rule. In areas such as clinical trials, stopping early and safely can ensure that subjects receive the best treatments and accelerate the development of effective therapies.

\end{abstract}

\section*{Introduction}

Statistical hypothesis testing \citep{Neyman33} underpins decision-making across science and e-commerce, from laboratory experiments to clinical trials, to large-scale adaptive experimentation and online A/B testing.\footnote{A/B testing, also known as split testing, is a randomized experiment comparing digital assets with control A and variant B} In its classical form, a hypothesis test is defined for a fixed sample size and guarantees control of Type I error (false positives) only when the decision is made at the planned endpoint, often referred to  as the maximal sample size. In practice, however, data are rarely generated in a single block. Experiments can be monitored as they progress, often with the desire of stopping early when evidence of the end decision appears compelling \citep{Robbins70, Robbins71}. Such interim analyses fundamentally break the error guarantees of classical tests, making repeated ``peeking” at the data statistically invalid, despite the clear motivation.

A substantial literature has sought to reconcile sequential decision-making with statistical validity. Group sequential designs control error by prespecifying a limited number of interim analyses, requiring rigid testing schedules and careful calibration using numerical methods, see e.g. \citep{Jennison89, Jennison99}. More recent anytime-valid methods allow arbitrary stopping with arbitrary large samples
using  likelihood-ratios and decisions based on notions of betting-odds. However, such approaches suffer substantial losses in power \citep{Shafer11, DeHeideGrunwald2021, hendriksen2021optional, Vovk21, Shafer21, Ramdas2023, Grunwald24}. As a consequence, it is currently believed that validity under arbitrary stopping necessarily requires  accepting a large efficiency penalty.

Earlier, and largely forgotten in the literature, an anytime valid approach was introduced by \cite{Lan1982}. Under the name stochastic curtailment, the approach relies on martingales and the Doob martingale inequality to provide a suitably bounded Type I error for early stopping. See also \cite{Jennsion1990} and \cite{Koning2026}. A disadvantage here, by defining the predictive probability of rejecting the null hypothesis as a martingale is (i) the anytime valid tests which can be implemented are limited due to the martingale constraint, and (ii) such anytime valid tests based on the martingale and martingale inequality are inadmissable, in that there is an alternative test with matching Type I error and superior power.

In this article we provide a comprehensive analysis of stochastic curtailment which we do through a general predictive framework. 
We also show how with this predictive framework we can relax the martingale constraint in a simple way while maintaining the ability to control Type I errors.
The predictive perspective allows us to examine more complicated anytime valid tests, such as nonparametric tests and tests when the null hypothesis is not fully specified, for example, there are nuisance parameters. We can also extend the framework to present anytime valid bootstrap tests: effectively, whenever a bootstrap test is feasible, we can define an anytime valid bootstrap test.

The key idea is to view sequential testing as a prediction problem. At any interim stage, uncertainty about the test outcome at the maximal sample size arises solely from future, unobserved observations. Assuming the null hypothesis is true, we compute the probability that a fixed-sample test would reject the null hypothesis at its planned endpoint. This forward-looking quantity represents a prediction of the final test decision and can be updated continually as new data arrive. This is most often best done by first predicting the test statistic at the maximal sample size conditional on the current observed data.

By thresholding the predicted rejection probability, we obtain a simple sequential decision rule: the null hypothesis is rejected as soon as the predictive probability that the fixed-sample test would ultimately reject the null hypothesis exceeds a user set threshold. However, if this threshold 
is determined from a martingale sequence of test statistics and the martingale inequality, then the resulting test is  inadmissable. The superior test is one in which the threshold is evaluated using Monte Carlo methods, since in this case there is no inaccuracy arising from the martingale inequality.
Ultimately, therefore, there is no need for a martingale constraint on the test statistics and this allows for a large class of test to be adapted to anytime valid tests, even ones using bootstrap methods. 

Under the null hypothesis, we show that the predicted rejection probability evolves in a way that guarantees control of the overall Type I error under arbitrary stopping. Unlike existing anytime-valid approaches, the resulting tests preserve the original fixed-sample test itself rather than replacing it, and this applies broadly, including to nonparametric tests. Our approach applies whenever the null distribution of the fixed-sample test statistic is known or can be simulated, a condition satisfied by most classical parametric and nonparametric tests. Figure~\ref{fig:main1} provides a schematic overview of the procedure.

\begin{figure}[t]
  \centering
  \includegraphics[width=0.8\linewidth]{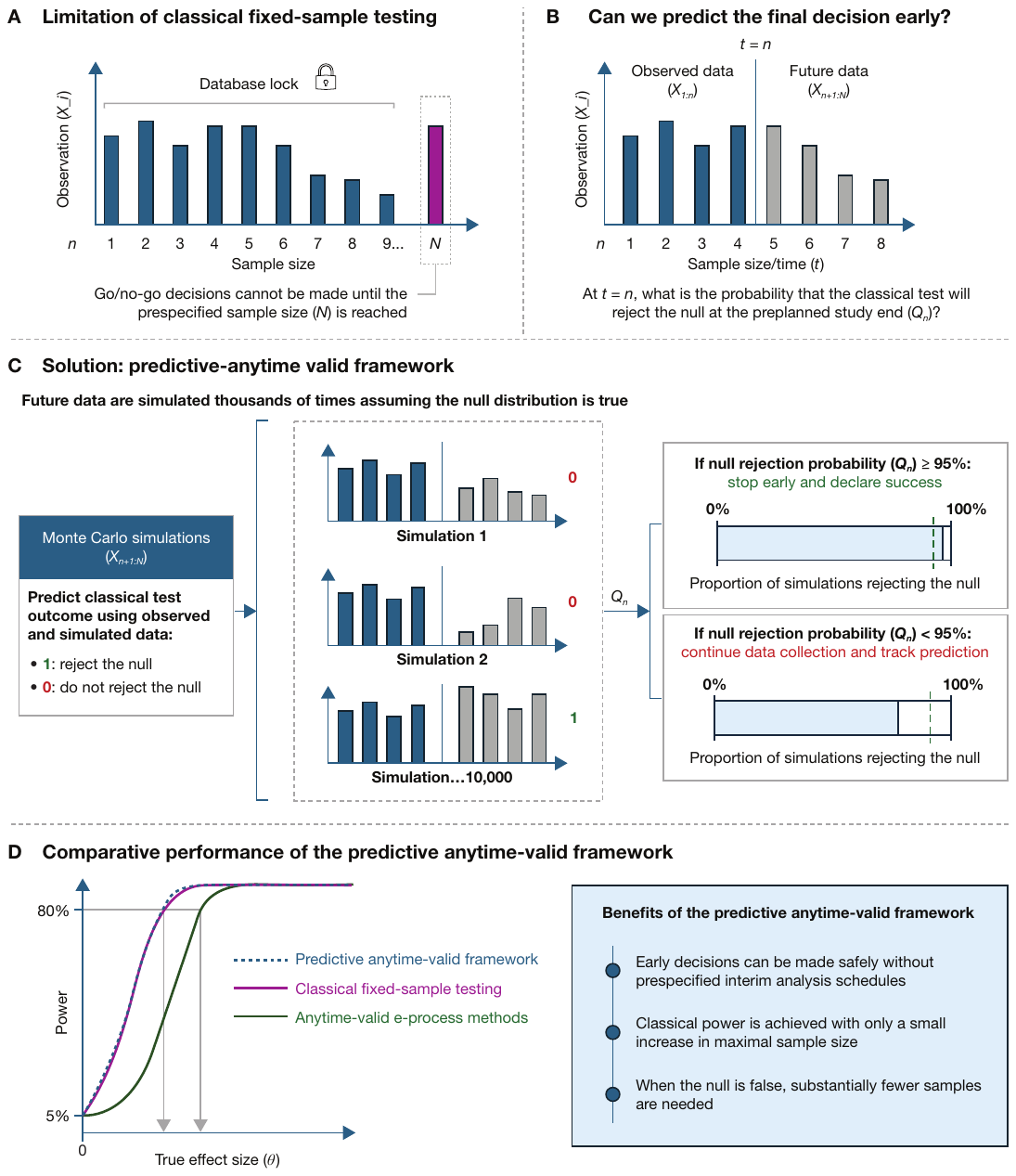}
  \caption{\textbf{Overview of the predictive procedure.} Panel A: the classical fixed-sample test precludes early analysis. Panel B: the uncertainty in the outcome of the fixed-sample test at $n < N$ is driven by the unobserved data $X_{n+1:N}$. Panel C: we can characterise the uncertainty in the test decision at $N$, assuming the null hypothesis to be true, by simulating the missing data under the null hypothesis. Repeated testing of completed datasets estimates the probability that the corresponding fixed-sample test will reject at $N$. Panel D: The predictive test has good power and saves samples when the null hypothesis is false.}
  \label{fig:main1}
\end{figure}

A closely related line of work arises in Bayesian adaptive trials, where decisions are guided by the posterior probability that a treatment will be effective at the planned end of the experiment \citep{berry2010bayesian, Spieg94}. These approaches adopt a forward-looking perspective, updating beliefs about the final outcome as data accrue. However, Bayesian adaptive designs require specification of priors and typically do not provide explicit control of frequentist Type I error \citep{DeHeideGrunwald2021}, making their operating characteristics sensitive to modeling assumptions that can be difficult to justify or communicate in many scientific and regulatory settings. The present work shares the predictive aspect of these Bayesian approaches, but differs fundamentally in that it provides objective, prior-free, predictions of the final fixed-sample test decision with rigorous control of Type I error.

Technical definitions, proofs, and additional examples to those contained in this manuscript are provided in the Supplementary Materials (SM). The background and notation to the anytime valid problem is described in Section 1 of the SM. Section 2 of the SM presents the mathematical foundations for the predictive approach to hypothesis testing.  In what follows, we focus on stopping for efficacy when the null hypothesis is false; stopping for futility,  when the null hypothesis is true, is mentioned in the Discussion and in the Supplementary Materials. We also describe how to construct repeated confidence intervals (RCI) derived from the anytime valid tests and we also consider nonparametric tests based on bootstraps.

\section*{Results}

\subsection*{Predicting the fixed-sample decision enables anytime validity}

At an interim sample size $n < N$, where $N$ is the fixed sample size, the future $N-n$ data points contributing to the final test are unobserved. Our approach provides the ability of early stopping at sample size $n$ by computing the predictive probability that the fixed-sample test would reject the null hypothesis at sample size $N$, conditional on the observed data and assuming the null hypothesis is true. We denote this quantity by $Q_n$. Intuitively, $Q_n$ measures how likely it is that, if the experiment were completed as planned under the null, the final test would reject the null hypothesis. The formal definition of \(Q_n\), together with generic Monte Carlo computation schemes, is given in Section 3 of the SM.

\begin{figure}[t]
  \centering
  \includegraphics[width=0.75\linewidth]{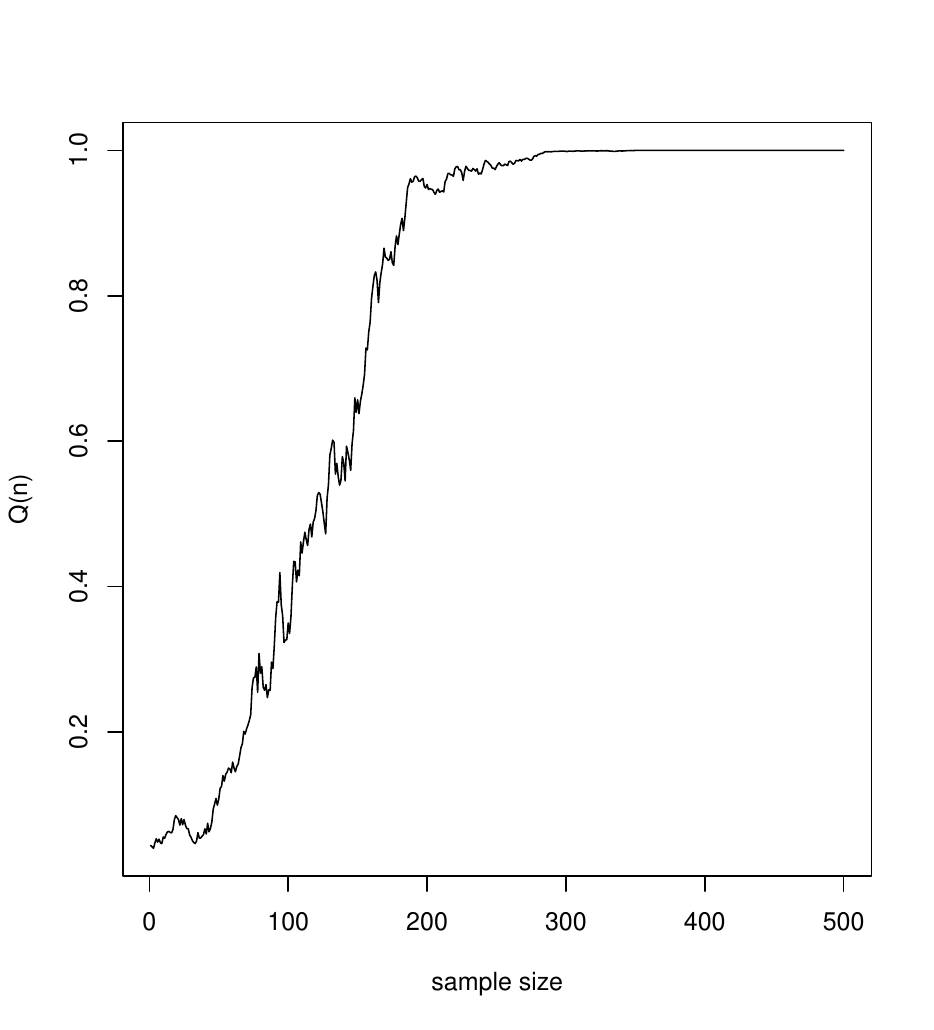}
  \caption{\textbf{Predicting the fixed-sample decision enables anytime-valid testing.} Shown is the evolution of the predicted rejection probability $(Q_n)$, defined as the probability that the fixed-sample test would reject at sample size $N$, conditional on data observed up to stage $n$ and assuming the null hypothesis is true. The experiment
  is testing $H_0:\theta=0$ versus $H_1:\theta>0$ for a normal mean $\theta$ with variance assumed known at 1. The value of $N$ is $500$ with Type I error of 0.05 for the sequential test.  The alternative value of $\theta=0.2$ is used.}
  \label{fig:main2}
\end{figure}

Figure~\ref{fig:main2} illustrates the evolution of $(Q_n)$ for a representative experiment where the null is false. Early in the experiment, the predicted probability of rejection is low, which reflects substantial uncertainty about the final outcome. As data accumulate, $Q_n$ increases and may cross a fixed threshold, $\gamma$, at which point the null hypothesis can be rejected immediately. This decision rule does not require a pre-specified protocol for interim analyses or adjustments based on how often the data are examined. 

The stopping threshold admits a direct probabilistic interpretation, where 
 early stopping occurs only when continued sampling is unlikely to alter the final fixed-sample decision. Section 4 of the SM details the properties of $(Q_n)$ under the null hypothesis which, under certain conditions, can be a martingale sequence. In this case, the martingale inequality can lead to the setting of a Type I error, using the Doob/Ville  inequalities, \cite{Ville39}, \cite{Doob53}. However, the martingale constraint does not lend itself to all type of test and indeed it is quite possible to control the Type I error, more accurately than the use of the martingale inequality, using Monte Carlo methods. 
 

Section 5 of the SM discuss the controlling of the Type I error which follow from rejecting the null hypothesis at sample size $n$ if ever $Q_n\geq \gamma$. The $\gamma$ is calculated based on the choice of Type I error at the maximal sample size. Denote this by $\wt$. For an overall anytime valid test to have Type I error $\alpha$ we need $\wt<\alpha$ and the reasonable default choice is $0.95\alpha$. With the choice of $\alpha$ and $\wt$ the $\gamma$ threshold can be calculated using Monte carlo methods. 
This choice of $\wt$ with $\alpha=0.05$ leads to around a 20\% saving in samples used when the null hypothesis is false for the typical examples we consider (see Section 6 of the SM). 
To allow for sequential monitoring without inflating false positives, the procedure uses a slightly more conservative rejection region than the classical test, i.e. $\wt<\alpha$. 
The practical cost of this conservatism in terms of power is minimal, to the order of a 0.5\% loss in power for typical settings. Any such small loss in sensitivity can be recovered by a minor increase in the maximum sample size, around 2\%  for the typical examples we consider. Importantly, these additional observations are only required if the experiment reaches the planned end of the original test without stopping. Formal definitions and calculations are given in Section 6 of the SM.


\subsection*{A canonical example: one-sided testing of a normal mean}

To make the construction concrete, consider a one-sided test of a normal mean with known variance, designed to reject the null hypothesis at sample size $N$ when the sample mean exceeds a fixed threshold; this example is considered in \citep{Spieg94}. At an interim stage $n < N$, the final test statistic can be written as a sum of the observed data and future, unobserved observations. Assuming the null hypothesis is true, the distribution of the future data is known, allowing the predicted rejection probability $Q_n$ to be computed analytically.

\begin{figure}[t]
  \centering
  \includegraphics[width=0.75\linewidth]{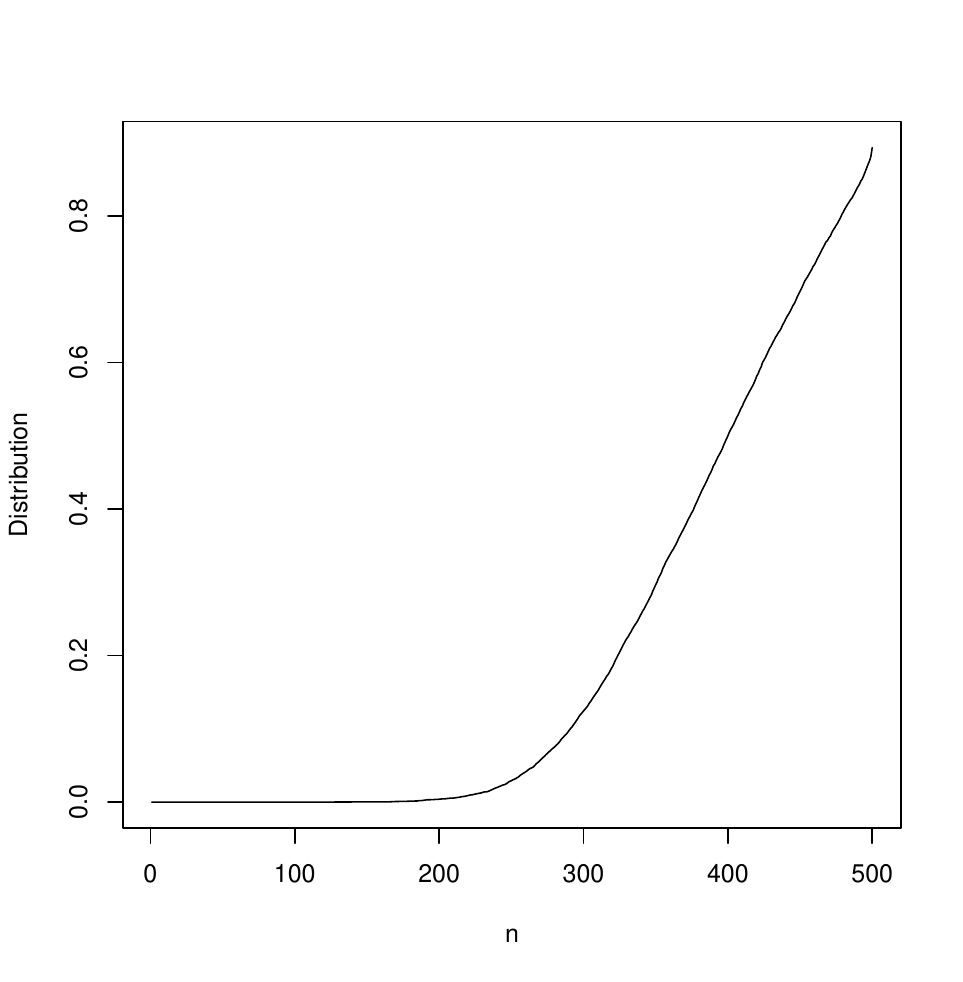}
  \caption{\textbf{
  Distribution of early stopping times.} Shown is the cumulative distribution of the early stopping times of the experiment of the same type as in Figure \ref{figsm1}. 
  The distribution records the sample size at which the $Q_n$ sequence crosses $0.95$ when the alternative hypothesis is true with $\theta^*=0.13$. 10,000 simulations were used to estimate the probability distribution.} 
  \label{fig:main3}
\end{figure}

As data accumulate, when the null hypothesis is false, the predicted probability that the fixed-sample test would reject the null hypothesis increases smoothly. When this probability exceeds the chosen threshold, the null hypothesis is rejected immediately. Figure~\ref{fig:main2} illustrates $Q_n$ for this example and Figure~\ref{fig:main3} shows the distribution of stopping-times under repeated experiments. The early stopping samples were simulated 10,000 times and the mean and median values are 384 and 388, respectively.
This illustrates how a familiar fixed-sample test can be transformed into an anytime-valid procedure, with substantial savings on the number of samples used when the null hypothesis is false, without altering its fundamental structure.
Closed-form expressions for \(Q_n\) in the one-sided Gaussian mean test, together with derivations, are provided in Section 7 of the SM.

\subsection*{Near-optimal power with minimal sample inflation}

Existing anytime-valid approaches, based on e-processes and likelihood-ratios, replace the fixed-sample test with a sequential statistic designed to ensure validity under optional stopping, replacing probabilities with the notion of betting-odds. This leads to a substantial loss in power, while affording sequential testing. 
Unlike e-processes, or Bayesian adaptive designs that monitor posterior probabilities, the predicted rejection probability $Q_n$ is defined entirely with respect to a classical fixed-sample test and the null hypothesis, ensuring explicit frequentist error guarantees and retaining near-optimal power. 

To match the power of the classical test, existing anytime-valid testing approaches often require a near doubling of the sample size, to approximately $2N$.  This has substantial impact on their utility in real-world settings, such as in clinical trials where sampling is costly and recruitment can be difficult. We find that this is not the case when sequential decisions are anchored to the original fixed-sample test. Across a range of standard testing problems, the predictive anytime-valid tests require only a small increase in maximal sample size, on the order of $2\%$, to ensure the power is at least as large as that of the corresponding fixed-sample tests (see Section 6 of the SM). A more detailed comparison with e-processes and betting-based constructions, including a discussion on power is given in Section 8 of the SM.

\begin{figure}[t]
  \centering
  \includegraphics[width=0.75\linewidth]{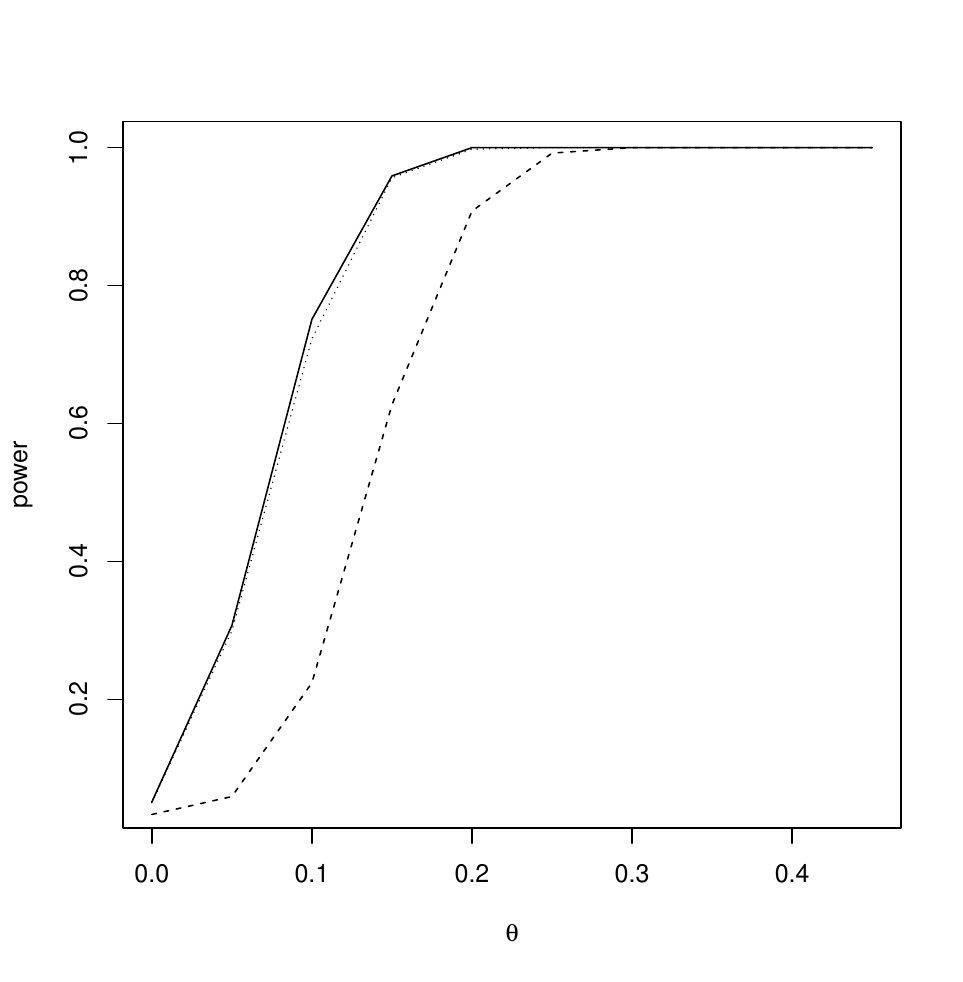}
  \caption{\textbf{Near-optimal power with minimal sample inflation.} 
  Power as a function of effect size is shown for a classical fixed-sample test (dotted line), the predictive anytime-valid test (bold line), and a representative anytime-valid likelihood-ratio–based method (dashed line); the latter two with a 2\% sample size increase. The predictive test closely tracks the power of the fixed-sample test across effect sizes, requiring only a small increase in maximal sample size, from $N=500$ for the original test to $N'=509$ for the sequential test, to recover classical power 
  while enabling early stopping and reducing expected sample usage under alternatives. The details of the experiment are the same as with Figures~\ref{fig:main2} and \ref{fig:main3}. The power functions are computed for varying $\theta$ values via simulation with Monte Carlo sample sizes of 10,000. }
  \label{fig:main4}
\end{figure}

Figure~\ref{fig:main4} compares power as a function of effect size for three procedures: the classical fixed-sample test, the predictive anytime-valid test, and a representative e-process anytime-valid likelihood-ratio based method. The details of the experiment are the same as for Figures~\ref{fig:main2} and \ref{fig:main3} save for the maximal sample size for the anytime-valid test and the predictive anytime-valid test being $N'=510$. The power functions are computed for varying $\theta$ values via simulation with Monte Carlo sample sizes of 10,000. The power curve of the predictive test closely tracks that of the fixed-sample test across effect sizes, while retaining the ability to stop early. In contrast, existing anytime-valid methods exhibit substantial power loss. Specifically, for $\theta^*=0.13$, the power values for the fixed-sample test, the predictive anytime-valid test, the e-process anytime-valid test are 0.90, 0.91, and 0.44, respectively.

Under alternatives where the null hypothesis is false, early stopping typically yields a substantive reduction in the expected number of samples required to reach a decision, often leading to around a $25\%$ saving in the mean number of required samples, see Section 6 of the SM. 
Thus, the predictive approach achieves flexibility in stopping without sacrificing efficiency. Notably, the same construction applies to nonparametric tests, such as distribution-free two-sample procedures, demonstrating that anytime validity need not rely on likelihood-ratio statistics that require a probability model under the alternative hypothesis. In practice, the predicted rejection probability can often be computed analytically when the null distribution of the test statistic is known, and efficiently approximated by Monte Carlo simulation otherwise. Because these simulations involve only forward sampling under the null hypothesis, they are straightforward to parallelize and scale to large experiments. 
A range of tests of varying types, including two sample tests, tests with composite null hypothesis and nonparametric tests are described in Section 9 of the SM.

\subsection*{Real-data illustration: early stopping in a clinical trial}

To illustrate predictive anytime-valid testing on real trial data, we analyze outcomes from the International Stroke Trial (IST), a large clinical trial investigating a number of aspects of care for individuals following a stroke. One of these is the effect of aspirin administered during the first 14 days after ischemic stroke; see \cite{Sander11}.  We focus on a comparison of death and censoring times between patients assigned to aspirin and those  not assigned to aspirin, using a standard two-sample nonparametric test, the two-sample log rank test, as the fixed-sample analysis. 
The group receiving aspirin had 9,071 individuals, and the other group had 10,326 individuals (i.e. 19,397 trial participants in total). The number of deaths recorded from the aspirin group is 1,897 and for the other group is 2,462. All remaining observations are censored with given censoring times. Full details of the experiment and the settings are provided in Section 10 of the SM. 

In the classical analysis, the null hypothesis of equal outcome distributions is tested after all patient data are observed, yielding a statistically significant difference in distribution of event times between groups. We treat this log rank fixed-sample test as the target decision and apply the predictive anytime-valid framework to the same data, analyzed sequentially in the order in which patient outcomes become available. 

At each interim stage $t$, for a discrete set of $t$ values, we compute the predicted rejection probability \(Q_t\), defined as the probability that the fixed-sample test would reject the null hypothesis at its maximal sample size, conditional on the data observed to time $t$ and assuming the null hypothesis is true. Early stopping is permitted only when the probability that the fixed-sample test would reject the null hypothesis falls above the prespecified threshold $\gamma$. 


In this illustration we observe substantial early stopping. Figure~\ref{fig:main5} shows the evolution of \(Q_t\) over the course of 15 months. Under the observed data, the predicted rejection probability increases as patient outcomes accumulate and crosses the stopping threshold, of $\gamma = 0.95$, at 7 months, quite well before the final sample size is reached. Importantly, the sequential procedure reaches the same qualitative conclusion as the classical fixed-sample test, but does so much earlier and without inflating the Type I error.

\begin{center}
\begin{figure}[!htbp]
\begin{center}
\includegraphics[width=0.75\linewidth]{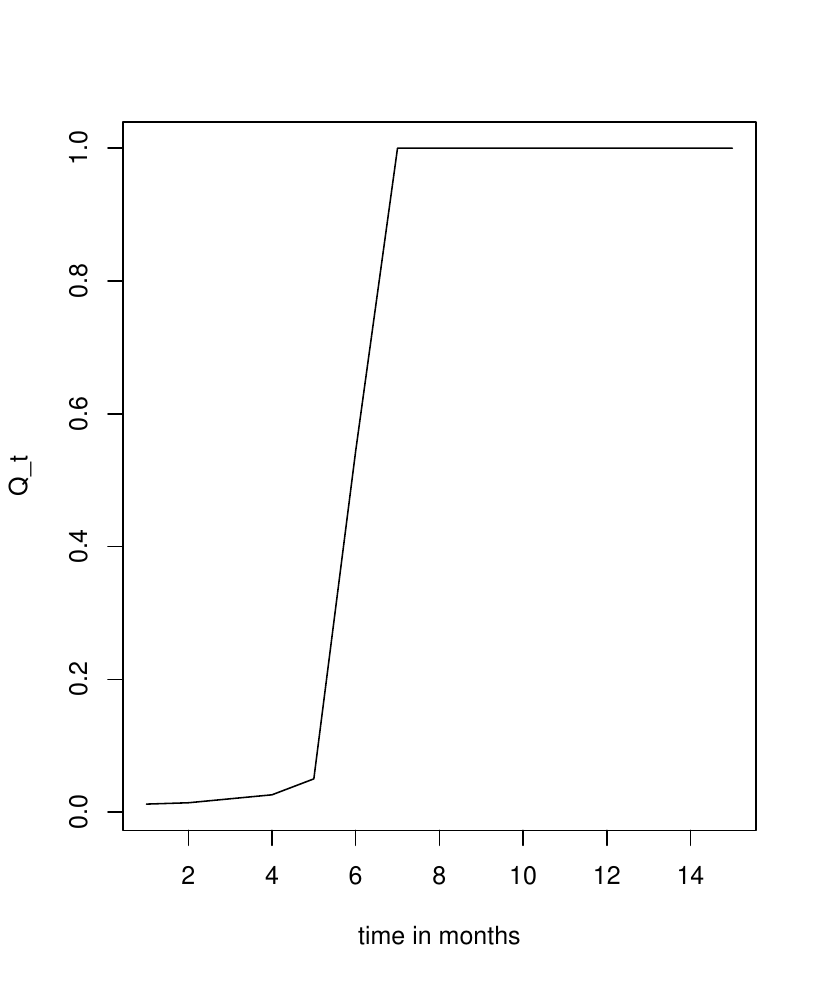}
\caption{\textbf{Predictive anytime-valid analysis of a clinical trial.}
    Predicted rejection probability \(Q_t\) for the International Stroke Trial (IST),
    comparing death including censoring times for patients assigned to aspirin for the first 14 days of the trial versus  not assigned to aspirin. Here the $t$ are represented by a discrete set of time points which are set over 15 months following the start of the trial.
    Full details of the experiment are presented in Section 10 of the Supplementary Materials.
    At each interim stage, \(Q_t\) represents the probability that the original
    fixed-sample two-sample test would reject the null hypothesis at the planned
    sample size, conditional on the data observed to time $t$ and assuming the null
    hypothesis is true. The predicted rejection probability crosses the stopping
    threshold before the final sample size is reached, allowing early stopping
    while preserving the calibration and conclusion of the fixed-sample analysis. }
\label{fig:main5}
\end{center}
\end{figure}
\end{center} 

This example demonstrates that predictive anytime-valid testing can be applied directly to real clinical trial data using familiar test statistics. The framework enables safe interim monitoring and early stopping for efficacy while preserving the integrity, interpretation, and calibration of the original fixed-sample analysis.
Additional details on the International Stroke Trial analysis, including data preprocessing, test specification, null completion, and sequential monitoring, are provided in Section 10 of the Supplementary Materials. 


Describing the remainder of the SM, in Section 12 we show how to derive repeated confidence intervals, in the spirit of \cite{Jennison89}, from the tests appearing in Section 9. Section 13 also follows \cite{Jennison89} by considering tests which rely on paths of statistics which behave as Brownian motion.
Finally, Section 14 presents anytime valid bootstrap tests which can be shown to exist whenever a test has a fixed sample bootstrap version.   

\section*{Discussion}

Sequential data generation and streaming data are defining features of modern applications, yet classical hypothesis testing frameworks, developed over 90 years ago, remain anchored to single test decisions at a maximal sample size. Existing sequential methods come with compromises. Our results show that the mismatch between theory and the practical needs of sequential testing is not fundamental. By re-framing sequential testing as a prediction problem, targeting the outcome of a fixed-sample test, while avoiding martingales and the martingale inequality, we demonstrate that anytime-valid error control, statistical efficiency, and broad applicability can be achieved simultaneously.

Our procedure turns any fixed-sample size test with a hypothesis into an anytime-valid sequential version of the test. Introducing a threshold for the predictive probabilities we can secure a Type I error at a chosen $\alpha$, and match the fixed-sample size power with substantial sample savings, by allowing for a modest increase in the maximal sample size. Alternatively, the maximal sample size of the sequential test can match that of the fixed sample size test and lose a small amount of power. The procedure is simple, requiring one user set parameter, $\wt<\alpha$.    

The central conceptual shift is to focus on the \emph{final decision} implied by a fixed-sample test rather than on accumulating evidence through intermediate statistics. At any interim stage, uncertainty about the test outcome arises solely from future, unobserved data. Treating this uncertainty explicitly, by predicting whether the fixed-sample test would reject under the null hypothesis, yields a sequential decision rule that remains valid under arbitrary stopping. Extensions to stopping for futility, when evidence accumulates that the null hypothesis is true, follow a similar predictive procedure and are considered in Section 11 of the SM. 

Prevailing wisdom has suggested that robustness to optional stopping necessarily comes at the cost of substantial efficiency losses, as reflected in existing anytime-valid and safe testing constructions \citep{Vovk21,Grunwald24}. Our results show that this belief is not intrinsic to anytime-valid inference, but rather reflects the limitations of existing constructions that replace fixed-sample tests with alternative sequential statistics based on likelihood-ratios.

The predictive perspective also clarifies the relationship between frequentist and Bayesian approaches to adaptive experimentation. Bayesian adaptive trials often monitor the posterior probability that a treatment will be effective at the planned end of the study \citep{berry2010bayesian,Spieg94}. While conceptually appealing, such procedures depend on specifying a prior and do not, in general, yield explicit guarantees on frequentist Type I error. Predictive anytime-valid testing can be viewed as providing an analogous forward-looking assessment of the final decision, but one that is anchored to a classical hypothesis test and yields objective probabilities with rigorously controlled error rates. In this sense, the framework here offers a bridge between the intuitive appeal of Bayesian prediction and the operating-characteristic guarantees required in many scientific and regulatory contexts.

The generality of the framework is particularly important. Because the construction operates on the distribution of the test statistic under the null hypothesis, it applies to a wide range of classical tests, including nonparametric procedures for which likelihood ratios are unavailable or ill-defined. This substantially expands the scope of anytime-valid inference.

These results have immediate practical implications in domains where data are monitored continuously and decisions must be made under uncertainty. In clinical trials, the ability to stop early for efficacy while maintaining rigorous error control can reduce patient exposure to inferior treatments and accelerate evaluation of effective therapies. More broadly, in adaptive experimentation and other forms of streaming data analysis, predictive anytime-valid testing offers a principled alternative to informal monitoring practices that are widespread but statistically unsound. 

Taken together, our findings show that classical fixed-sample inference and modern sequential decision-making are not competing paradigms, but can be unified within a single coherent framework.


\section*{Acknowledgments}
ChatGPT was used by the authors as part of the manuscript review. C.H. gratefully acknowledges support and interactions with Novo Nordisk and GSK through the University of Oxford that contributed to the background motivation for this research. Novo Nordisk and GSK had no role in the study design, analysis, interpretation, or manuscript preparation.

\clearpage
\section*{}

\begin{center}
\vspace*{0.2cm}
{\LARGE\bfseries Supplementary Materials}
\vspace{0.2cm}
\end{center}

\begingroup
\setlength{\cftsecnumwidth}{2.5em} 
\setlength{\cftsubsecnumwidth}{3em} 
\etocsetnexttocdepth{2}
\localtableofcontents
\endgroup












\setcounter{section}{0}
\setcounter{subsection}{0}
\renewcommand{\thesection}{S\arabic{section}}
\renewcommand{\thesubsection}{S\arabic{section}.\arabic{subsection}}

\setcounter{figure}{0}
\renewcommand{\thefigure}{S\arabic{figure}}

\setcounter{table}{0}
\renewcommand{\thetable}{S\arabic{table}}

\renewcommand{\theHsection}{S\arabic{section}}
\renewcommand{\theHsubsection}{S\arabic{section}.\arabic{subsection}}
\renewcommand{\theHfigure}{S\arabic{figure}}
\renewcommand{\theHtable}{S\arabic{table}}



\def\wt{\widetilde{\alpha}}
\def\half{\small{1\over2}}






\section{Background and notation}
\label{sec:S0}

Anytime valid testing is the idea of performing a  test as often and as frequently as desired without incurring an explosion in errors.   
Further, the power of the anytime valid testing procedure up to a maximal sample size $N$ should not be significantly reduced  from the corresponding fixed sample test of size $N$.
When comparing anytime valid tests with a fixed sample size test at sample size  $N$, the power for the anytime valid test will be reduced. The compensation for this is the possible early stopping of the experiment. 
In short, there is for anytime valid tests a balancing act to be performed between the degree of early stopping and the overall power of the test.
It goes without saying that the more early stopping is encouraged,  so the power is reduced. However, not necessarily by large amounts. 
The only way to maintain the same power as the fixed sample size test at $N$ is to not allow for early stopping. That is, any test performed at sample size $n<N$ will not reject the null hypothesis.

As far as we can ascertain, all current anytime valid testing procedures rely on sequences of statistics which are martingales.
The Type I errors are controlled by use of the martingale inequality.  This inequality gives the probability for a positive martingale sequence to remain bounded by a particular value. This probability can be used to set the Type I error.

The need for a martingale sequence of statistics is a serious constraint on the type of test which can be adapted to an anytime valid test. The main types are those based on simple likelihood ratio tests and tests for which the distribution of the data under the null hypothesis is fully specified. That is, there are no nuisance parameters. This would appear to exclude the well known Student -$t$ test which does not assume a known variance when testing for a normal mean.  Nonparametric tests are also problematic when requiring martingales. Another type of test which is problematic involve two sample tests when the common parameter value is unspecified. 

Anytime valid testing dates back to the 1950s, see \cite{Robbins70} and \cite{Robbins71}, which relied on martingales.  An alternative  martingale approach was introduced in \cite{Lan1982} and further developed by \cite{Koning2026}. Our major contribution is to change the perspective of the martingale construction and to view it as one of prediction. The essence of the \cite{Lan1982} approach is to predict the outcome of the test at sample size $N$ at any sample size $n<N$. The martingale approach requires a very strict version of the prediction. On the other hand, we  focus on the prediction of the test statistic at sample size $N$,  in a natural way, which could include the estimation of nuisance parameters. The key is that we do not need to find Type I errors using the martingale inequality; instead, we  can find suitable critical regions using Monte Carlo methods.

Group sequential methods also do not require martingales for the sequence of statistics. However, they are only able to find critical regions for a limited number of interim tests. For example, the Tables in \cite{Jennison89} give 10 critical values for 10 interim analyses. And even then they are only based on tests involving normal models. The issue with the group sequential methods is that the sequence of test statistics need their own critical value and this makes Monte Carlo methods for finding suitable values very difficult. With our approach, we have a standardized sequence of statistics and so only need a single critical value which can be found using Monte Carlo methods.

\subsection{Notation}
We use the following notation: 
\begin{itemize}
    \item $T_N$ --  the test statistic for the fixed maximal sample test with sample size $N$. 
    \item $\mathbb{P}_0$ the distribution of the $X_{1:N}$ defined by the null hypothesis when it is known in full. We write
    $\mathbb{P}_0(T_N\in\mathcal{C})$, for example, where $\mathcal{C}$ would represent a critical region.
    \item $X_{i:j}$ -- the sequence of data $\{X_i, X_{i+1}, \ldots, X_{j}\}$
    \item $\mathcal{F}_n=\sigma(X_1,\dots,X_n)$ -- the filtration generated
by the observed data up to the sample size $n$.
    \item $T_N^{(n)}$ -- the predicted value of $T_N$ of the test statistic at sample size $N$ given partial information from $X_{1:n}$ at $0\leq n \leq N$.
    \item $P$ is the distribution of the $X_{n+1:N}$ given $X_{1:n}$ when the null hypothesis not not fully define the distribution of the future sample. We write
    $P(T_N^{(n)}\in\mathcal{C})$, for example.
    \item $\alpha$ -- the Type I error of a test.
    \item $\mathcal{C}_{\alpha}$ -- the critical region of a fixed sample test with Type I error $\alpha$. We will often simplify to use $\mathcal{C}$ unless unclear. 

    \item $Q_n$ -- the predictive probability that the fixed sample test will reject the null hypothesis at $N$ given partial information  $\mathcal{F}_n$. 
    \item $\gamma$ -- the user set threshold probability for early stopping if $Q_n \ge \gamma$.
    \item $\tilde{\alpha} <\alpha$ -- a value used in the prediction probability of the sequential test to ensure non-inflation of Type I error; i.e. $Q_n=P(T_N^{(n)}\in \mathcal{C}_{\widetilde{\alpha}})$. The $\gamma$ and $\wt$ are set to ensure
    $P(\max_n Q_n\geq\gamma)=\alpha$.
\end{itemize}

\section{Problem setup}
\label{sec:S1}

We observe data sequentially, say $X_1,X_2,\dots$ and consider a classical hypothesis
test defined at a pre-specified maximal sample size $N$.
The fixed-sample test is specified by:
(i) a test statistic $T_N = T(X_{1:N})$ and
(ii) a critical rejection region $\mathcal{C}_{\alpha}$ such that
\[
\mathbb{P}_0(T_N \in \mathcal{C}_{\alpha}) \le \alpha ,
\]
where $\mathbb{P}_0$ denotes the probability of $T_N$ under the null hypothesis. Of course it is preferable for reasons of power for $\mathcal{C}_\alpha$ to provide a probability of $\alpha$.

The goal of predictive anytime-valid testing is to allow continuous monitoring, as data accrue,
and early stopping of the experiment through a rejection of the null hypothesis, while preserving the calibration, interpretation, and
scientific meaning of the original fixed sample size test.

Throughout, $\mathcal{F}_n=\sigma(X_1,\dots,X_n)$ denotes the filtration generated
by the observed data up to the sample size $n$. The goal is to predict the event $\{T_N\in \mathcal{C}\}$, for any $\mathcal{C}$, conditional on $\mathcal{F}_n$ and for each $n$. This can only be achieved by the construction of a conditional distribution of the missing part of the sample, i.e. $X_{n+1:N}$ conditional on what has been observed, i.e. $X_{1:n}$. This is in many cases going to be defined by the null hypothesis. Hence, for example, if $\mathbb{P}_0$ defines an independent and identically distributed (i.i.d.) sequence on $X$ then the $X_{n+1:N}$ will be taken to be distributed as i.i.d. from $\mathbb{P}_0$.

If the current observed sample size is $n<N$, the corresponding uncertainty 
about any decision is created by the missing $X_{n+1:N}$. The $X_{1:n}$ has been observed, and contains no uncertainty.  Quantification of the uncertainty, using probability, follows via a probability model for the missing data given what has been seen; i.e.
\begin{equation}\label{predictive}
p(X_{n+1:N}\mid X_{1:n}).
\end{equation}
Although we use a $\mid$ in the usual conditional probability sense, it is to be understood that this merely reflects for us that the $X_{1:n}$ have been observed and is given and used to construct a joint probability model for $X_{n+1:N}$. There is no implication of a joint model for $X_{1:N}$. Indeed, we are not modeling what has been seen. 

We argue that (\ref{predictive}) is crucial to the modeling of decision uncertainty as viewed in terms of complete datasets. It is also the foundation of the conventional Bayesian approach which views (\ref{predictive}) more in terms of data structures than density functions (see \cite{deFinetti37, Doob49, Fong2023}). However, while the Bayesian approach would be regarded as the standard tool for a joint predictive model, the appropriate predictive when considering the assessment of uncertainty under a null hypothesis is the distribution defined by the hypothesis. So, if the null hypothesis defines
$\mathbb{P}_0(X_1,\ldots,X_N)$ then the predictive (\ref{predictive}) for us will be given by
$\mathbb{P}_0(X_{n+1:N}\mid X_{1:n})$ and this predictive is available for some standard tests explicitly.
In this case, the standard approach is to define
$$Q_n=\mathbb{P}_0(T_N\in\mathcal{C}\mid\mathcal{F}_n).$$
Standard math yields this sequence of $(Q_n)$ as a martingale and hence the errors can be bounded using the martingale inequality.
We use a more general definition of the prediction $Q_n$ by defining the predicted statistic at sample size $N$  given samples up to size $n$ by $T_N^{(n)}$.
Our more general definition of $Q_n$ is given by $$Q_n=P\bigg(T_{N}^{(n)}\in\mathcal{C}\bigg).$$ 
 Under this definition there is no notion that the sequence is a martingale. However, we are going to be able to determine the critical value $\gamma$ for which $P(\max_{n\leq N} Q_n\geq\gamma)=\alpha$ using Monte Carlo methods. Note that, under the martingale definition, it is that $P(\max_{n\leq N}  Q_n\geq\gamma)\leq Q_0/\gamma$ and so $\gamma$ is set to be $Q_0/\alpha$.
 
 {\bf Inadmissibility of martingale inequality based approaches:} We can immediately see that if $(Q_n)$ is a martingale and the martingale inequality is used to obtain the Type I error, we end with
 $P(\max_{n\leq N} Q_n\geq\gamma)\leq\alpha$, whereas using Monte Carlo methods we can find the $\gamma$ for which $P(\max_{n\leq N} Q_n\geq\gamma)=\alpha$. The $\gamma$ will be smaller in the latter case and hence the power will be larger. This demonstrates the anytime valid test using the martingale inequality is inadmissable.


\section{Predictive rejection probability}
\label{sec:S2}

This section builds on Section \ref{sec:S1} and introduces the predictive rejection probability used throughout the
paper, and explains the role of the predictive critical region.

\subsection{Definition}
\label{sec:S2.1}

To assess at an interim stage $n<N$  whether the fixed-sample test defined in Section \ref{sec:S1} would ultimately reject the null hypothesis,  
we treat the remaining observations $X_{n+1:N}$ as missing data and generate completions
$X'_{n+1},\dots,X'_N$ from the null distribution $\mathbb{P}_0$ (or $P$, a null-consistent resampling
scheme for composite or nonparametric settings).
Define the predicted completed-sample statistic
\[
T_N^{(n)} = T(X_{1:n},X'_{n+1:N}).
\]
Hence, for each $n$, $T_N^{(n)}$ is a random variable. For completeness we define $T_N^{(0)}$ to be the random variable with the same distribution as $T_N$ under the null hypothesis and $T_N^{(N)}=T_N$ as the final test statistic derived from the completed experiment. 

\begin{center}
\begin{figure}[!htbp]
\begin{center}
\includegraphics[width=14cm,height=6cm]{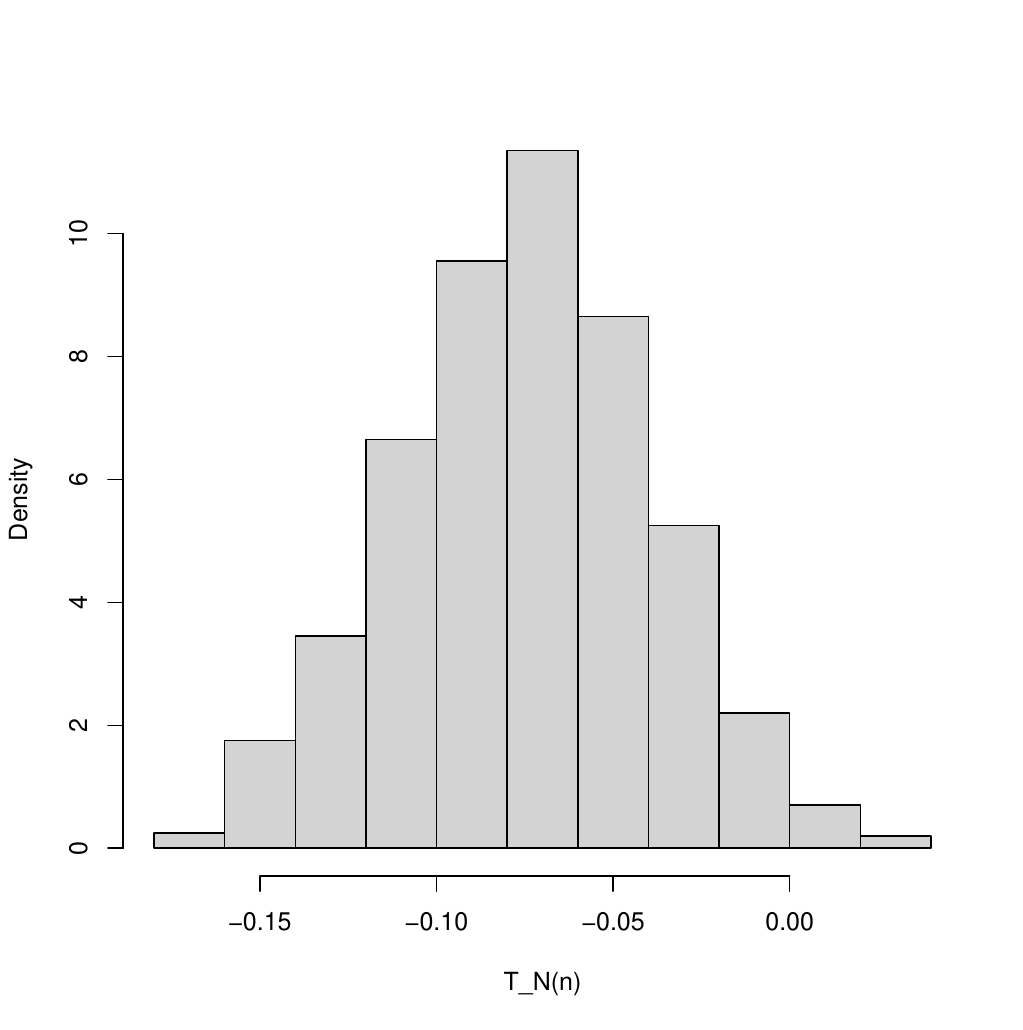}
\caption{Simulations of $X'_{n+1:N}$ leading to $T_N^{(n)}$ samples, with $n=200$ and $N=500$.}  
\label{figsm10}
\end{center}
\end{figure}
\end{center}

A histogram of $T_N^{(n)}$ is presented in Fig.\ref{figsm10} based on a null hypothesis for a normal mean being 0 with known variance 1.
So 
$$T_N^{(n)}=N^{-1}\left(n\,T_n+\sum_{i=n+1}^N X_i'\right),$$
where $T_n=n^{-1}\sum_{i=1}^n X_i$ and the $X'_{n+1}$ are i.i.d. normal with mean 0 (the null hypothesis value for a normal mean) and variance 1 (the assumed known variance). 
Here $N=500$ and $n=200$. 
Repeated sampling of the $X'_{n+1:N}$ yield the uncertainty in $T_N^{(n)}$ and from this we can derive $Q_n=P(T_{N}^{(n)}\in\mathcal{C})$. 

\begin{definition}
The \emph{predictive rejection probability} given a critical region $\mathcal{C}$ is defined as
\[
Q_n = P\!\left(T_N^{(n)}\in\mathcal{C}\right)
\]
with $Q_0=P(T_N\in\mathcal{C})$ and $Q_N={\bf 1}(T_N\in\mathcal{C})$.

\end{definition}

\noindent
Intuitively, $Q_n$ is the predictive probability, under the  null hypothesis, that
the test using critical region $\mathcal{C}$ would reject at sample size $N$, conditional on
the data observed so far and represented by the predicted test statistic. 

\vspace{0.1in}
\noindent
{\sc Normal mean example with known variance}. To illustrate a sequence of $(Q_n)$, consider the test of a normal mean with known unit variance, i.e. $H_0:\theta=0$ versus $H_1:\theta>0$. The standard test statistic for a sample of size $N$ is $T_N=\sqrt{N}\bar{X}_N$ and the Type I error $\alpha$ defines the critical region through $\mathcal{C}_\alpha=(z_{\alpha},\infty)$, where $z_\alpha=\Phi^{-1}(1-\alpha)$. Now
$$T_N^{(n)}=\sqrt{n/N}\,T_n+\sqrt{1/N}\,\sum_{i=n+1}^N X_i'$$
where the $X'_{n+1:N}$ are i.i.d. normal with zero mean and unit variance and $T_n=\sqrt{n}\,\bar{X}_n$. Hence,
$$Q_n=1-\Phi\left(\frac{\sqrt{N}\,z_\alpha-\sqrt{n}T_n}{\sqrt{N-n}}\right),\quad n=0,1,\ldots,N-1.$$
Note that $Q_0=\alpha$ and $Q_N$ will be 1 or 0 depending on whether $T_N\in \mathcal{C}$ or not. In this case, the $(Q_n)$ forms a martingale sequence.

\vspace{0.1in}
\noindent
{\sc Normal mean example with unknown variance}.
This normal case is an example of when the $X_{n+1:N}$ are defined by the null hypothesis. On the other hand, suppose we are in the situation of testing for a normal mean to be 0 when the variance is unknown. So now the standardized statistic is
$T_N=\sqrt{N}\bar{X}_N/S_N$
where $S_N$ is the sample variance. It is well known that $T_N$ is standard Student $t$ distributed with $N-1$ degrees of freedom. 
We define the predicted test statistic at $N$ given samples $X_{1:n}$ to be
$$T_N^{(n)}=\sqrt{N}\bar{X}_N/S_n,$$
where we take $X'_{n+1:N}$ to be i.i.d. normal with zero mean and variance $S_n^2$. Hence, for some independent standard normal random variable, we have
$$T_N^{(n)}=\sqrt{n/N}\,T_n+\sqrt{1-n/N}\,Z,$$
where $T_n=\sqrt{n}\bar{X}_n/S_n$,
giving
$$Q_n=1-\Phi\left(\frac{c\sqrt{N}-\sqrt{n}T_n}{\sqrt{N-n}}\right)$$
for some critical value choice $c$. Note that since the $(T_n)$ do not depend on the unknown $\sigma$ neither do the $(Q_n)$ and so the distribution of $(Q_2,\ldots,Q_N)$ can be established via Monte Carlo simulation using any choice of $\sigma$. Note in this case the sequence $(Q_n)$ is not a martingale.

\subsection{Monte Carlo computation of the predictive rejection probability}

In many settings the conditional distribution of the predicted sample statistic
$T_N^{(n)}$ under the null hypothesis is not available in closed form.
In such cases, the predictive rejection probability $(Q_n)$ can be computed by Monte
Carlo simulation.
Conditional on the observed data $\mathcal{F}_n$, we can generate independent replicates
of the unobserved future data
$X_{n+1}(b),\dots,X_N(b) \sim P$, $b=1,\dots,B$,
according to the null hypothesis which defines $P$, and which may include parameter estimation as in the Student $t$ example.
Each replicate is combined with the observed data to form a completed-sample
statistic $T_N^{(n)}(b)$.
The predictive rejection probability is then approximated by
\[
\widehat Q_n
= \frac{1}{B}\sum_{b=1}^B
\mathbf{1}\!\left\{T_N^{(n)}(b) \in C\right\}.
\]
This estimator is unbiased for $Q_n$ conditional on $\mathcal{F}_n$, and its accuracy
can be controlled by the number of replicates $B$.
In practice, modest values of $B$ (e.g.\ a few thousand) are sufficient, depending on the value of the stopping criteria. 






\subsection{Sequential stopping rule}

Given a threshold critical value $\gamma \in (0,1)$, 
the stopping time for the experiment $\tau$ is defined by
\[
\tau = \min\{n \le N : Q_n \ge \gamma\},
\]
where $N$ is the maximum sample size.
If no such $n$ occurs before $N$ observations are collected, the experiment automatically stops at
$N$. 

\subsection{Predictive critical region}

On moving from a fixed sample size test at $N$ with Type I error $\alpha$ to an anytime valid test we need to define $Q_n$ using an $\wt<\alpha$. For if we define $Q_n=P(T_N^{(n)}\in \mathcal{C}_\alpha)$ then, for any $0<\gamma<1$, it is that $P(Q_N\geq\gamma)=P(T_N\in \mathcal{C}_\alpha)=\alpha$.
Hence, the overall Type I error will exceed $\alpha$ making it necessary to define $Q_n$ via $P(T_N^{(n)}\in\mathcal{C}_{\wt})$ for a $\wt<\alpha$.  
The region $C_{\tilde{\alpha}}$ acts as a predictive proxy for the original
rejection region.
It is used only within the predictive calculation of $(Q_n)$ and does not define
a new hypothesis test. As shown in Section \ref{sec:S5} of the SM , setting $\tilde{\alpha} <\alpha $ and hence $C_{\tilde{\alpha}}$ in this way guarantees that the Type I error of the sequential procedure is  $\alpha$, matching that of the original test. Indeed, using Monte Carlo simulation we can find the threshold $\gamma$ so that the Type I error is precisely $\alpha$.
Fig.~\ref{fig:S-critical-regions} illustrates this relationship for a canonical
one-sided normal test, where $C_{\tilde{\alpha}}$ corresponds to a slightly tightened
upper-tail cutoff. The Fig.~\ref{fig:S-critical-regions} assumes that $\wt$ has been constructed as $\alpha\gamma$ for some $\gamma<1$.
It is important to stress that while the sequential test uses critical region, $C_{\tilde{\alpha}}$, the procedure has Type I error control of $\alpha$ (see Section \ref{sec:S5} of the SM).

\begin{figure}[t]
\centering
\begin{tikzpicture}

\def\za{1.645}
\def\zag{1.670}

\begin{groupplot}[
  group style={group size=1 by 2, vertical sep=1.0cm},
  width=0.85\linewidth,
  height=5cm,
  domain=-3:3,
  samples=250,
  axis lines=left,
  ylabel={density},
  ymin=0,
  ymax=0.45,
  no markers
]

\nextgroupplot[
  xlabel={$z$},
  xmin=-3, xmax=3,
  xtick={-3,-2,-1,0,1,2,3},
  xticklabels={-3,-2,-1,0,1,2,3}
]

\addplot {1/sqrt(2*pi)*exp(-x^2/2)};

\addplot[dashed] coordinates {(\za,0) (\za,0.45)};
\addplot[dashed] coordinates {(\zag,0) (\zag,0.45)};

\addplot [draw=none, fill=gray, fill opacity=0.25, domain=\za:3]
{1/sqrt(2*pi)*exp(-x^2/2)} \closedcycle;

\addplot [draw=none, fill=gray, fill opacity=0.45, domain=\zag:3]
{1/sqrt(2*pi)*exp(-x^2/2)} \closedcycle;

\node[anchor=south west] at (axis cs:\za,0.20) {$z_{1-\alpha}\approx \za$};
\node[anchor=south west] at (axis cs:\zag,0.14) {$z_{1-\alpha\gamma}\approx \zag$};

\nextgroupplot[
  xlabel={$z$},
  xmin=1, xmax=3,
  xtick={1,1.5,2,2.5,3},
  xticklabels={1,1.5,2,2.5,3}
]

\addplot {1/sqrt(2*pi)*exp(-x^2/2)};

\addplot[dashed] coordinates {(\za,0) (\za,0.45)};
\addplot[dashed] coordinates {(\zag,0) (\zag,0.45)};

\addplot [draw=none, fill=gray, fill opacity=0.25, domain=\za:3]
{1/sqrt(2*pi)*exp(-x^2/2)} \closedcycle;

\addplot [draw=none, fill=gray, fill opacity=0.45, domain=\zag:3]
{1/sqrt(2*pi)*exp(-x^2/2)} \closedcycle;

\node[anchor=north west] at (axis cs:1.05,0.42)
{$C_{\alpha\gamma}\subset C_\alpha$};

\end{groupplot}
\end{tikzpicture}

\caption{\textbf{Illustration of the predictive critical region (one-sided).}
Under a standard normal null, the original fixed-sample rejection region is
$C_\alpha=\{Z\ge z_{1-\alpha}\}$.
The predictive calculation uses the slightly more conservative region
$C_{\alpha\gamma}=\{Z\ge z_{1-\alpha\gamma}\}$, so that
$C_{\alpha\gamma}\subset C_\alpha$.
The upper plot shows the two regions for $\alpha=0.05$ and $\gamma=0.95$, and the lower panel zooms in on the region $z\in[1,3]$. }
\label{fig:S-critical-regions}
\end{figure}

\section{Properties of $(Q_n)$}
\label{sec:S3}

When we have $P\equiv \mathbb{P}_0$, and the best illustration of this is the comparison of the $z$ and Student $t$ tests, then we have
$Q_n=P(T_N\in\mathcal{C}\mid \mathcal{F}_n)$ and as is well known such sequences are martingales. For the sake of completeness, we present the proof.

\begin{proposition}
Assume 
$X_{n+1:N}$ arise from the null model $\mathbb{P}_0$ conditional on $X_{1:n}$.
Then $(Q_n)_{n=0}^N$ is a martingale with respect to $(\mathcal{F}_n)$.
\end{proposition}

\begin{proof}
With an observed sample size of $n$ we will provide a complete dataset by taking $X'_{n+1:N}$ to be from the null density $p_0$ conditional on $X_{1:n}$.
We then define 
$$Q_n=P\big(T(X_{1:n},X'_{n+1:N})\in \mathcal{C}\mid X_{1:n}\big).$$
We also define $Q_0=\mathbb{P}_0(T^{(0)}_N\in \mathcal{C})$ and $Q_N={\bf 1}(T_N\in \mathcal{C})$.
Expanding the definition of $Q_n$, we have
$$Q_n=\int\cdots\int {\bf 1}(T(X_{1:n},x'_{n+1:N})\in \mathcal{C})\,p_0(x'_{n+1:N}\mid X_{1:n})\,dx'_{n+1:N}.$$
Now
$E(Q_{n+1}\mid X_{1:n})$ is given by
$$\begin{array}{ll}
& \int\cdots\int 1(T(X_{1:n},x_{n+1},x'_{n+2:N})\in C)\,p_0(x'_{n+2:N}\mid x_{n+1},X_{1:n})\,p_0(x_{n+1}\mid X_{1:n}) \,dx_{n+1}\,dx'_{n+2:N} \\ \\
= &  \int\cdots\int 1(T(X_{1:n},x_{n+1},x'_{n+2:N})\in C)\,p_0(x_{n+1},x'_{n+2:N}\mid X_{1:n})\,dx_{n+1}\,dx'_{n+2:N}
\end{array}
$$
which is easily seen to be $Q_n$. 
\end{proof}

\noindent
A more explicit version of the proof can be seen by investigating the illustration involving the normal mean in Section \ref{sec:S2.1}. In this case we have seen that
$Q_n=1-\Phi((\sqrt{N}z_\alpha-T_n)/\sqrt{N-n})$
which we start by writing as 
$P(Z\geq (\sqrt{N}z_\alpha-T_n)/\sqrt{N-n})$ where $Z$ is a standard normal random variable. So
$$Q_{n+1}=P\bigg(Z\geq (\sqrt{N}z_\alpha-T_n-Z')/\sqrt{N-n-1}\bigg),$$
where $Z'$ is also a standard normal random variable independent of $Z$, 
and hence
$$Q_n=P\bigg(Z+Z'/\sqrt{N-n-1}\geq (\sqrt{N}z_\alpha -T_n)/\sqrt{N-n-1}\bigg).$$
Now $Z+Z'/\sqrt{N-n-1}=Z''\sqrt{(N-n)/(N-n-1)}$
for some standard normal variable $Z''$.
So $$Q_n=P(Z''\sqrt{(N-n)/(N-n-1)}\geq (\sqrt{N}z_\alpha-T_n)/\sqrt{N-n-1})=1-\Phi((\sqrt{N}z_\alpha-T_n)/\sqrt{N-n})$$
as required.

While we can define a martingale process for the $(Q_n)$ for the Student $t$ test, we are unable to compute the $Q_n$ since we are unable to compute the statistic $T_N$ with only knowledge of the $X_{1:n}$. Hence, we are required to define a predicted statistic $T_N^{(n)}$ which we can compute given $X_{1:n}$. The feature now is that the $(Q_n)$ will not form a martingale, but this does not prevent us from controlling the Type I error for anytime valid tests.

There are a number of tests which do not yield practical martingales and these include

\begin{description}

\item 1. Tests which specify equality of e.g. means without specifying the common value; i.e. 
$$H_0:\theta_X=\theta_Y.$$

\item 2. Tests which include nuisance parameters; i.e.
$$H_0:\theta=\theta_0$$
with $X\sim f(\cdot\mid\theta,\phi)$ with $\phi$ a nuisance parameter.

\end{description}

\noindent
In case 1. the statistic is 
$$T_N=|\widehat{\theta}(X_{1:N})-\widehat{\theta}(Y_{1:N})|.$$
and the predictive statistic based on a sample of size $n$ is
$$T_N^{(n)}=|\widehat{\theta}(X_{1:n},X'_{n+1:N})-\widehat{\theta}(Y_{1:n},Y'_{n+1:N})|$$
where the $X'_{n+1:N}$ and $Y'_{n+1:N}$ are independent from the model with common parameter
$$\theta=\widehat\theta(X_{1:n},Y_{1:n}).$$
Then $Q_n=P(T_N^{(n)}\geq c)$ and in this case the $(Q_n)$ is not a martingale sequence.

\vspace{0.2in}
\noindent
In case 2. the statistic is
$$T_N=T\left(\widehat\theta(X_{1:N}),\widehat{\phi}_N\right).$$
and the predictive statistic is
$$T_N^{(n)}=T\left(\widehat\theta(X_{1:n},X'_{n+1:N}),\widehat\phi_n\right)$$
where the $X'_{n+1:N}$ are from the model $f(\cdot\mid\theta_0,\widehat\phi_n)$.
As before, $Q_n=P(T_N^{(n)}\geq c)$ for some critical value $c$ with $(Q_n)$ not being a martingale sequence.

\section{Predictive anytime-valid Type I error control}
\label{sec:S4}

When  $(Q_n)$ is a martingale as a special case, from the Doob maximal inequality (\cite{Doob53}) it follows that for any $0<\gamma\leq 1$ it is that
$$P\left(\max_{1\leq n\leq N}Q_n\geq \gamma\right)\leq Q_0/\gamma.$$
Note that $Q_0$ is determined by the choice of $\mathcal{C}$. In order to guarantee a Type I error bounded by $\alpha$ for the anytime valid test we need to set $Q_0/\gamma=\alpha$. Hence, for the anytime valid test we determine $\mathcal{C}$ so that $P(T_N\in \mathcal{C})=\wt$ and therefore we have the relation $\alpha\gamma=\wt.$ Throughout we will use both $\wt$ and $\alpha\gamma$ to denote this quantity. The $\gamma$ can take any value in $(0,1)$ and thus it is possible to have some control over Type II errors while guaranteeing that the overall Type I error is $\alpha$. Note that for $\gamma=1$ we recover the fixed maximal sample size test.


However, we do not require the martingale inequality to determine a suitable $\gamma$ and Type I error. 
Since we have the distribution of the $(Q_n)$ via simulation we can use Monte Carlo methods to find the $\gamma$ for which
$$P\left(\max_{n\leq N} Q_n\geq\gamma\right)=\alpha,\quad\mbox{with}\quad Q_n=P\left(T_N^{(n)}\in \mathcal{C}_{\wt}\right).$$
For completeness we present the algorithm.

For some large $B$ compute $(Q_n^{(b)})$ for $b=1,\ldots,B$. For each $b$ a data set $X_{1:N}^{(b)}$ is generated. So, for example, for the Student $t$ test, we would have
$$Q_n^{(b)}=1-\Phi\left(\frac{c_{\wt}\sqrt{N}-\sqrt{n}T_n^{(b)}}{\sqrt{N-n}}\right)$$
and
$Q_N^{(b)}=1(T_N^{(b)}\geq c_{\wt})$,
where the $(T_{n}^{(b)})$ is the sequence of Student $t$ statistics from the data set $X_{1:N}^{(b)}$.  
We then find the $\gamma>0$ for which
$$B^{-1}\sum_{b=1}^B 1\bigg(\max_{n\leq N} Q_n^{(b)}\geq \gamma\bigg)=\alpha,$$ 
This will actually gives a more precise connection between $\gamma$ and $\alpha$ than can be provided by the martingale inequality.


In summary, when $(Q_n)$ is not a martingale,  we choose the critical value $\gamma$ for $(Q_n)_{n\leq N}$
by Monte Carlo methods; i.e. we find the $\gamma$ for which
$$P\bigg(\max_{n\leq N}Q_n \leq\gamma\bigg)=1-\alpha.$$ 
This can be found via repeated simulation of $(Q_{n\leq N})$ and such a $0<\gamma<1$ will exist when $\wt<\alpha$, 

\section{Power, sample-size, and stopping thresholds}
\label{sec:S5}

\subsection{Reduction in power when tightening the critical region.} The use of the  predictive critical region $C_{\tilde{\alpha}} \subset C_\alpha$ may lead to a small reduction in power. We can characterise a worse case loss in power by considering the sequential test having not rejected the null before $N$, and compare the power of the fixed-sample test versus that of the sequential test at $N$. 

For illustration we consider a one-sided test whose fixed-sample design has Type~I
error $\alpha$ and power $100\% (1 - \beta)$ at a chosen
design alternative hypothesis used to calculate $N$. Under the standard normal (or large-sample $z$) approximation,
the design alternative satisfies
\[
\sqrt{N}\,\theta^* \;=\; z_{1-\alpha} + z_{1-\beta},
\]
where $z_p:=\Phi^{-1}(p)$ and $\theta^*$ is the standardised effect size.
If we keep the same sample size $N$ but tighten the Type~I error in the predictive test to
$\tilde{\alpha}<\alpha$, then the critical value increases from $z_{1-\alpha}$ to
$z_{1-\tilde{\alpha}}$, and the resulting power at the same design alternative is
approximately
\[
1-\tilde{\beta}
\;\approx\;
\Phi\!\bigl(z_{1-\beta}-(z_{1-\tilde{\alpha}}-z_{1-\alpha})\bigr),
\qquad\text{equivalently}\qquad
\tilde{\beta}
\;\approx\;
\Phi\!\bigl((z_{1-\tilde{\alpha}}-z_{1-\alpha})-z_{1-\beta}\bigr).
\]
Thus, tightening $\alpha$ reduces power by an amount governed primarily by the
increase in the critical value $\Delta z := z_{1-\tilde{\alpha}}-z_{1-\alpha}$.

\vspace{0.2in}
\noindent
{\sc Example.}
For $N=500$, $\alpha=0.05$, and $\beta=0.1$ (i.e.\ 90\% power), tightening to
$\tilde{\alpha}=0.0475$, with $\gamma=0.95$, increases the critical value by
$\Delta z = z_{0.9525}-z_{0.95}\approx 0.025$.
The approximation above gives
\[
1-\tilde{\beta}
\;\approx\;
\Phi\!\bigl(z_{0.9}-\Delta z\bigr)
=
\Phi(1.2816-0.025)
\approx
0.895,
\]
so the power decreases from 0.90 to approximately 0.895 (a loss of about
0.5 percentage points) when $N$ is held fixed.

\subsection{Maintaining power via a small increase in sample size}

As shown above, the modest change to the critical region used by the predictive test may lead to a loss in power. In order to guarantee no loss in power relative to the original fixed-sample test, the maximal sample size for the predictive test can be increased slightly to $N' = N + m$. The value of $m$ can be set using standard power analysis (shown below). In Table \ref{tab:power} we illustrate this for a typical example with a threshold stopping decision of $\gamma=0.95$. Correcting to guarantee matching of power leads to around a 2\% increase in maximal sample size, i.e. $N'=1.02N$, and up to a 30\% reduction in samples required under the alternative hypothesis. Note, clearly, that the additional samples, $(N+1, N+2, \ldots, N+m)$, are only ever needed if the procedure has not already stopped by $N$. 

\begin{table}[t]
\caption{\textbf{Additional samples required to preserve 90\% power when tightening the one-sided level with $\alpha=0.05$ and $\gamma=0.95$.} 
For a fixed standardized effect size $\delta$, the required sample size rescales as
$N' = N\left((z_{1-\tilde{\alpha}}+z_{0.90})/(z_{1-\alpha}+z_{0.90})\right)^2$.
Values shown use $z_{\alpha} = z_{0.95}\approx 1.6449$, $z_{\tilde{\alpha}}=z_{0.9525}\approx 1.6696$, and
$z_{0.90}\approx 1.2816$, with $N'$ rounded up to the nearest integer.
For small sample sizes (e.g.\ $N=10$ or $20$), the apparent percentage increase
is dominated by integer rounding of $N'$; the underlying multiplicative inflation
factor is approximately $1.7\%$ across all $N$, and the observed percentage increase
converges rapidly to this value as $N$ grows.
The final four columns report the mean, median, and interquartile range
($q_{0.25}$ to $q_{0.75}$) of the stopping time $\tau=\inf\{n:Q_n\ge\gamma\}$ under the
design alternative calibrated to yield 90\% fixed-sample power at $(N',\tilde{\alpha})$,
estimated using 10{,}000 Monte Carlo runs per $N'$.
Across these $N'$, the interquartile range scales approximately linearly with $N'$
and is about one quarter of $N'$ (i.e. $q_{0.75}(\tau)-q_{0.25}(\tau)\approx 0.25N'$),
indicating that typical stopping times concentrate well before the maximum sample size.}
\centering
\vspace{0.2in}
\label{tab:alpha-tightening}
\begin{tabular}{c c c c c c c c}
\hline 
$N$ & $N'$ & $m$ & \% increase $(m/N)$ &
$\mathbb{E}(\tau)$ & $\mathrm{median}(\tau)$ & $q_{0.25}(\tau)$ & $q_{0.75}(\tau)$ \\
\hline
10   & 11   & 1  & 10.0\% & 9.2  & 9   & 8   & 11  \\
20   & 21   & 1  & 5.0\%  & 17.1 & 17  & 15  & 20  \\
50   & 51   & 1  & 2.0\%  & 40.9 & 41  & 35  & 48  \\
100  & 102  & 2  & 2.0\%  & 81.1 & 82  & 70  & 95  \\
500  & 509  & 9  & 1.8\%  & 400.0 & 403 & 342 & 467 \\
1000 & 1017 & 17 & 1.7\%  & 799.2 & 804 & 681 & 934 \\
\hline
\end{tabular}
\label{tab:power}
\end{table}

Using $C_{\alpha\gamma}$ introduces a small amount of conservatism relative to the
original fixed-sample test that may lead to a loss in power.
To guarantee the original power, this conservatism can be offset by allowing a modest
increase in the maximum sample size from $N$ to $N'=N+m$. Table~\ref{tab:power}  summarizes the required inflation for $\gamma = 0.95$ for a range of sample sizes in representative settings.

Returning to the normal mean illustration, the power for the fixed sample test, with true mean $\theta^*$, is
given by 
$$1-\Phi\bigg(\Phi^{-1}(1-\alpha)-\sqrt{N}\theta^*\bigg).$$
For the sequential test, the power has a lower bound of 
$$1-\Phi\bigg(\Phi^{-1}(1-\alpha\gamma)-\sqrt{N'}\theta^*\bigg).$$
We can find the $N'$ for which these two match; i.e.
$$\Phi^{-1}(1-\alpha)-\sqrt{N}\theta^*=\Phi^{-1}(1-\alpha\gamma)-N'\theta^*,$$
leading to
$$\sqrt{N'}=\sqrt{N}+\frac{\Phi^{-1}(1-\alpha\gamma)-\Phi^{-1}(1-\alpha)}{\theta^*}.$$
Define $z_{1-u}:=\Phi^{-1}(1-u)$ so that the increment can be written as
$\{z_{1-\alpha\gamma}-z_{1-\alpha}\}/\theta^*$, leads to, 
\[
N' \;=\; N\left(1+\frac{z_{1-\alpha\gamma}-z_{1-\alpha}}{\theta^*\sqrt{N}}\right)^2.
\]
If the fixed-sample test at $N$ has power $100\% \times \beta$ at the standardised effect size
$\theta^*$, then $\sqrt{N}\,\theta^* = z_{1-\alpha}+z_{\beta}$ (with
$z_{\beta}:=\Phi^{-1}(\beta)$), and substitution gives
\[
N' \;=\; N\left(1+\frac{z_{1-\alpha\gamma}-z_{1-\alpha}}{z_{1-\alpha}+z_{\beta}}\right)^2
\;=\;
N\left(\frac{z_{1-\alpha\gamma}+z_{\beta}}{z_{1-\alpha}+z_{\beta}}\right)^2.
\]



\subsection{The effect of setting $\gamma$ on the distribution of stopping times} 

In this section we view the choice of $\wt$ being set by a first choice of $\gamma$ while maintaining the relation
$P(\max_{n\leq N}Q_n\geq\gamma)=\alpha$ and $Q_n=P(T_N^{(n)}\in\mathcal{C}_{\wt})$.
The parameter $\gamma$ governs a tradeoff between conservatism and efficiency.
Smaller values of $\gamma$ lead to earlier stopping on average under the alternative
but require a larger $N'$ to guarantee power, while larger values of $\gamma$ reduce
the required $N'$ at the cost of later stopping.

This tradeoff is illustrated empirically in Fig.~\ref{fig:gamma_tradeoff_N500} (design alternative) and 
Fig.~\ref{fig:stopping_tradeoff_null} (null hypothesis).
Under the null hypothesis, most trajectories run to $N'$, as expected for Type~I
error control of $\alpha=0.05$, whereas under the design alternative the procedure typically stops
substantially earlier.
These figures provide practical guidance for selecting $\gamma$ in applications.

The setting of $\gamma$ should, in part, be driven by the confidence of the experimentalist that the alternative hypothesis holds. Most experiments will be motivated by a certain level of belief that the alternative is true. The price to pay for the benefits of early stopping is the increased maximal sample size expected to be used if the null is true. Setting $\gamma = 0.95$, or $\gamma=0.90$ seems a reasonable default, leading to a 20-23\% average saving in samples under the alternative, with a modest increase in samples used, of around 2-4\%, if the null is true.  

Let us return to the illustration involving the normal mean and let the true value of $\theta$ be $\theta^*\ne 0$.
Now
$Q_n\geq\gamma$ if and only if
$$T_n\geq \sqrt{N}z_{\alpha\gamma}-\sqrt{N-n}z_{\gamma}.$$
Under the alternative hypothesis, $T_n=\sqrt{n}Z_n+n\theta^*$
and so the hypothesis is rejected if ever
$$Z_n\geq \sqrt{N}z_{\alpha\gamma}-\sqrt{N-n}z_\gamma-n\theta^*,$$
where the $Z_n$ are a sequence of dependent standard normal variables.

If $\gamma$ is close to 1 then $z_{\alpha\gamma}$ is close to $z_\alpha$ and so when $n=N$ the test will be similar to the fixed maximal sample test and hence will have a power close to that for this test. On the other hand, for this choice of $\gamma$ ,$z_\gamma$ is a large negative value and so rejecting the null hypothesis early will suffer for small $n$. 
If $\gamma$ is close to 0 then $z_\gamma$ is large and positive and so the probability of early stopping increases. However, at the same time, $z_{\alpha\gamma}$ is now large, reducing the power of the test.


\begin{figure}[t]
\centering
\begin{tikzpicture}
\begin{axis}[
    width=0.85\linewidth,
    height=6cm,
    xlabel={Maximum sample size $N'$},
    ylabel={Stopping time $\tau$},
    xmin=500, xmax=590,
    xtick={500,510,520,530,540,550,560,570,580,590},
    ymin=260, ymax=480,
    legend style={at={(0.97,0.97)},anchor=north east},
    axis lines=left,
    tick align=outside
]

\addplot [
    draw=none,
    fill=blue!20,
]
coordinates {
    (502,375) (506,355) (509,342) (514,328) (518,326) (527,315) (537,305)
    (548,295) (559,289) (571,284) (584,279)
    (584,458) (571,457) (559,454) (548,454) (537,459) (527,459) (518,466)
    (514,464) (509,468) (506,472) (502,477)
} \closedcycle;
\addlegendentry{IQR (25th--75th)}

\addplot[
    color=blue,
    mark=o,
    thick
]
coordinates {
    (502,420.0) (506,408.3) (509,400.4) (514,392.7) (518,391.8)
    (527,384.9) (537,380.6) (548,375.7) (559,373.4) (571,372.6) (584,371.4)
};
\addlegendentry{Mean}

\addplot[
    color=orange,
    mark=square*,
    thick
]
coordinates {
    (502,427) (506,413) (509,403) (514,394) (518,390)
    (527,382) (537,374) (548,368) (559,364) (571,361) (584,358)
};
\addlegendentry{Median}

\end{axis}
\end{tikzpicture}
\caption{\textbf{Stopping-time tradeoff as a function of $\gamma$ and hence maximum sample size $N'$ guaranteeing power} for the one-sided normal-mean test with known variance, original sample size $N=500$
and $\alpha=0.05$.  
Curves show the mean and median stopping times under the alternative, and the shaded
region indicates the interquartile range (25th--75th percentiles), based on 10{,}000
Monte Carlo simulations per $\gamma$.
The values of $\gamma$ corresponding to the plotted points are left-to-right $\{0.99,\,0.97,\,0.95,\,0.93,\,0.90,\,0.85,\,0.80,\,0.75,\,0.70,\,0.65,\,0.60\}$, For each $\gamma$, the maximum sample size $N'$ is chosen to guarantee 90\% power, the same as the original design alternative.
Lower values of $\gamma$ lead to earlier stopping on average at the cost of a larger allowable maximum sample size.}
\label{fig:gamma_tradeoff_N500}
\end{figure}

\begin{figure}[t]
\centering
\includegraphics[width=0.75\linewidth]{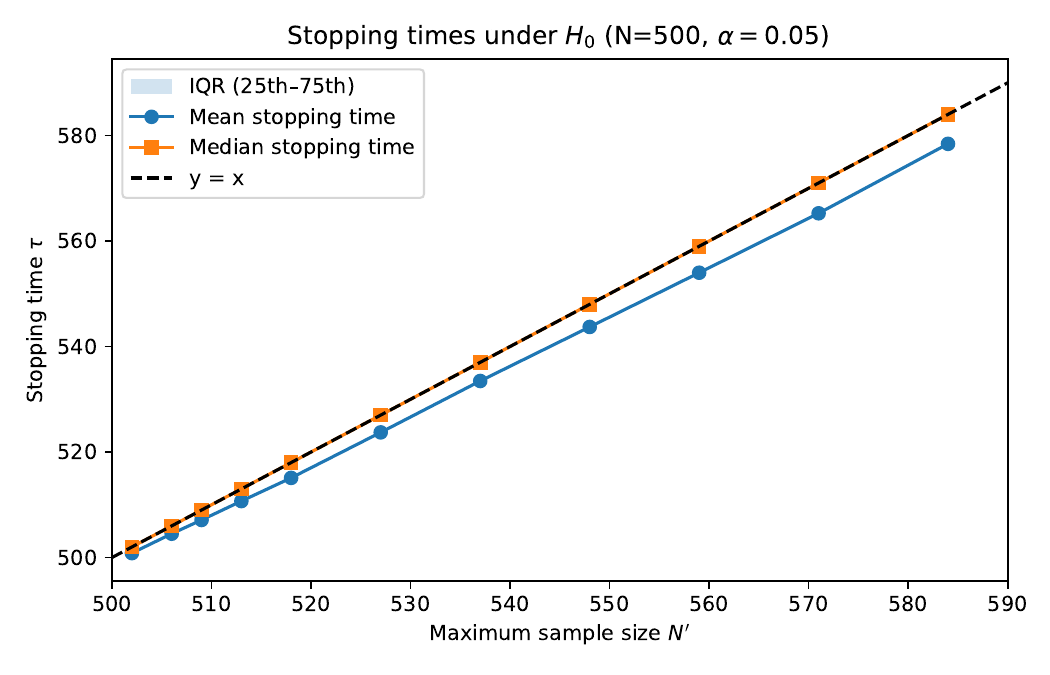}
\caption{\textbf{Stopping times under the null hypothesis $H_0$.}
The $x$-axis is the maximum allowed sample size $N'$, where $N'$ is chosen for each $\gamma$, as in Figure  \ref{fig:gamma_tradeoff_N500}, to preserve 90\% fixed-sample power at the original design
alternative, and the y-axis is the stopping time $\tau$.
Curves show the mean and median stopping times and the shaded region indicates the
interquartile range (25th--75th percentiles), based on 10{,}000 Monte Carlo
simulations per $\gamma$.
The dashed line $y=x$ indicates stopping at $N'$; under $H_0$, most trajectories lie
on or near this line, with only a small fraction stopping earlier due to false
rejection, consistent with the Type~I error control.}
\label{fig:stopping_tradeoff_null}
\end{figure}

\section{Worked example: one-sided Gaussian mean test}
\label{sec:S6}

As a first illustration to demonstrate the anytime valid test including the ability to control both Type I and Type II errors, we look at a normal mean hypothesis test, where under the null hypothesis the data are independent and identically distributed from the standard normal distribution, and under the alternative the mean is greater than 0. The test based on likelihood ratios would need a prior in order to deal with the composite alternative and therefore the test effectively becomes a Bayesian one using Bayes factors (see \cite{Kass95}).

The same setting is discussed by \cite{Spieg94} for a trial involving cancer therapies, one being the old conventional therapy and the other being a new treatment of high energy neutron therapy. The treatment difference is represented by a single parameter $\delta$ from which the observations arise as normal variables with mean $\delta$ and known variance. 

To proceed, the appropriate classical test at maximal sample size $N$ is to reject $H_0$ if $\sqrt{N}\,\bar{X}_N>z_{\wt}$ for some $\wt$, where $X_{1:N}$ are the observed sample and $\bar{X}_N$ is the sample mean.
Then define
$$Q_n=P\left(\sum_{i=1}^n X_i+\sum_{i=n+1}^N X'_i>z_{\wt}\sqrt{N} \,\Big| \,X_{1:n} \right),$$
where the $X_{1:n}$ are observed and the $X'_{n+1:N}$ are sampled from $\mathbb{P}_0$ - the standard normal distribution.
Under the null hypothesis,  $\sum_{i=n+1}^N X'_i$ is normal distributed with mean 0 and variance $N-n$, so, for $n=0,\ldots,N-1$,
$$Q_n=1-\Phi\left(\frac{z_{\wt}\sqrt{N}-\sum_{i=1}^n X_i}{\sqrt{N-n}}\right).$$
Observe that $Q_0=\wt$ and $Q_N=1(\bar{X}_N>z_{\wt}/\sqrt{N})$. 
Hence, since $(Q_n)$ is a martingale under the null hypothesis,
$$P\Big(\max\{Q_1,\ldots,Q_N\}>\gamma\Big)\leq\wt/\gamma$$
for any $0<\gamma\leq 1.$ We write $\wt=\alpha\gamma$ and throughout take $\alpha=0.05$, a benchmark value.

To highlight the difference between using the martingale inequality and Monte Carlo methods for determining $\gamma$; for the martingale inequality if we choose $\wt=0.95\,\alpha$ then the value for $\gamma$ is 0.95. On the other hand, if we determine the $\gamma$ using Monte Carlo methods, then it is 0.88.



Here we investigate the Type I and Type II errors for the sequential test using a maximal sample size of $N$ and for ease of analysis and exposition  we use the martingale inequality. Hence we use $\wt=\alpha\gamma$ and the rejection threshold for the $(Q_n)$ is therefore $\gamma$.
First we provide a look at the power curve as $\gamma$ ranges from 0.1 to 0.9, for fized $N$, using $N=500$, $\alpha=0.05$, and $\theta^*=0.13$. We used Monte Carlo simulation with a sample of size 10,000. As can be seen in Fig.~\ref{figsm6} the power curve increases with $\gamma$ though does not change very much. The value at $\gamma=0.1$ and $\gamma=0.9$ are 0.82 and 0.89, respectively.

Next we look at the distribution function of stopping times.
In Fig.~\ref{figsm7a} we show the distribution function with $\gamma=0.1$ and using $N=500$ with $\theta^*=0.13$ 
and $\alpha=0.05$. Fig.~\ref{figsm7b} shows the distribution function with $\gamma=0.95$. The means are 244 and 384 for $\gamma=0.1$ and $\gamma=0.95$, respectively.

For the theory for this example, if the alternative is $\theta^*$, and $\tau$ is defined to be the stopping time, i.e. $\tau=\min_n\{Q_n\geq\gamma\}$, then
$P(\tau\leq m)\geq P(Q_m\geq\gamma)$ and
$$Q_m=1-\Phi\left(\frac{\sqrt{N}\Phi^{-1}(1-\alpha\gamma)-T_m-(N-m)\theta^*}{\sqrt{N-m}}\right),$$
so
$$P(\tau\leq m)\geq 1-\Phi\left(\frac{\sqrt{N}\Phi^{-1}(1-\alpha\gamma)-\sqrt{N-m}\Phi^{-1}(1-\gamma)-N\theta^*}{\sqrt{m}}\right)$$
which becomes large very quickly due to the presence of $N\theta^*$.



\begin{center}
\begin{figure}[!htbp]
\begin{center}
\includegraphics[width=12cm,height=4cm]{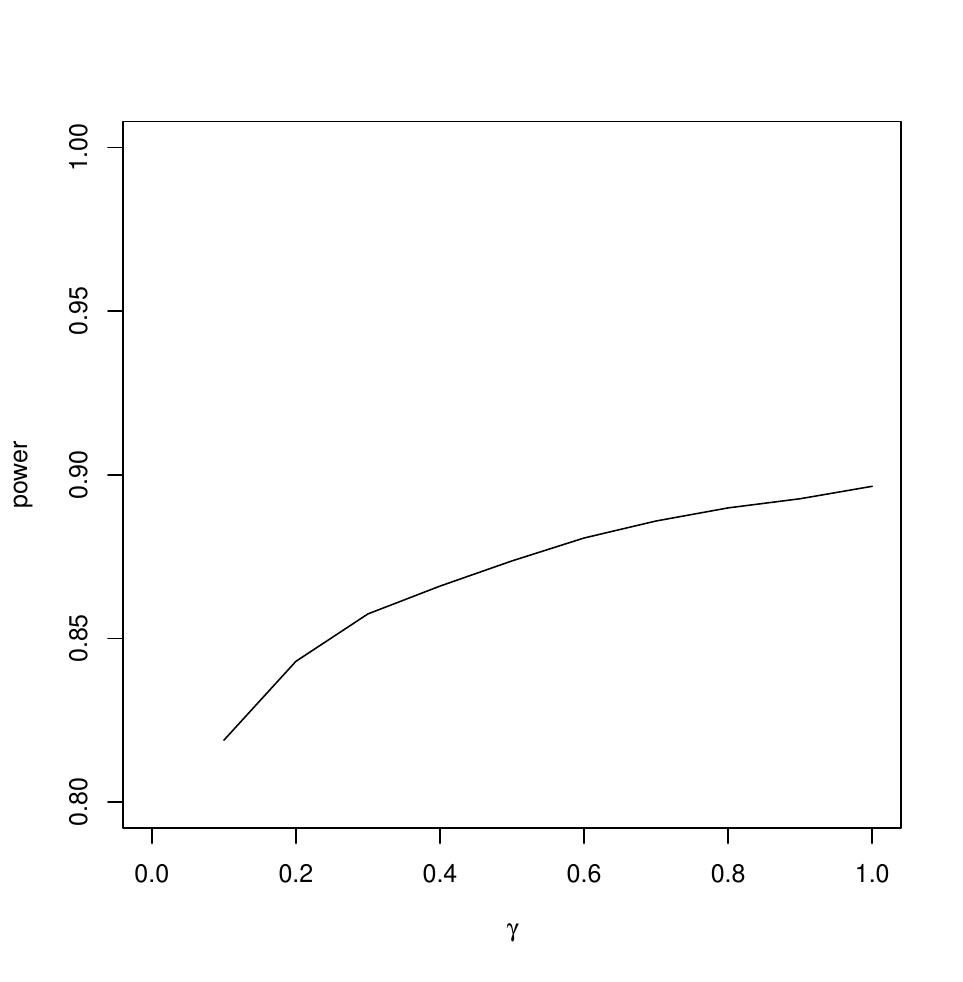}
\caption{ Power function of sequential test plotted against value of $\gamma$, using $\theta^*=0.13$ with $N=500$ and a Monte Carlo sample size of 10,000.}
\label{figsm6}
\end{center}
\end{figure}
\end{center} 


\begin{center}
\begin{figure}[!htbp]
\begin{center}
\includegraphics[width=12cm,height=4cm]{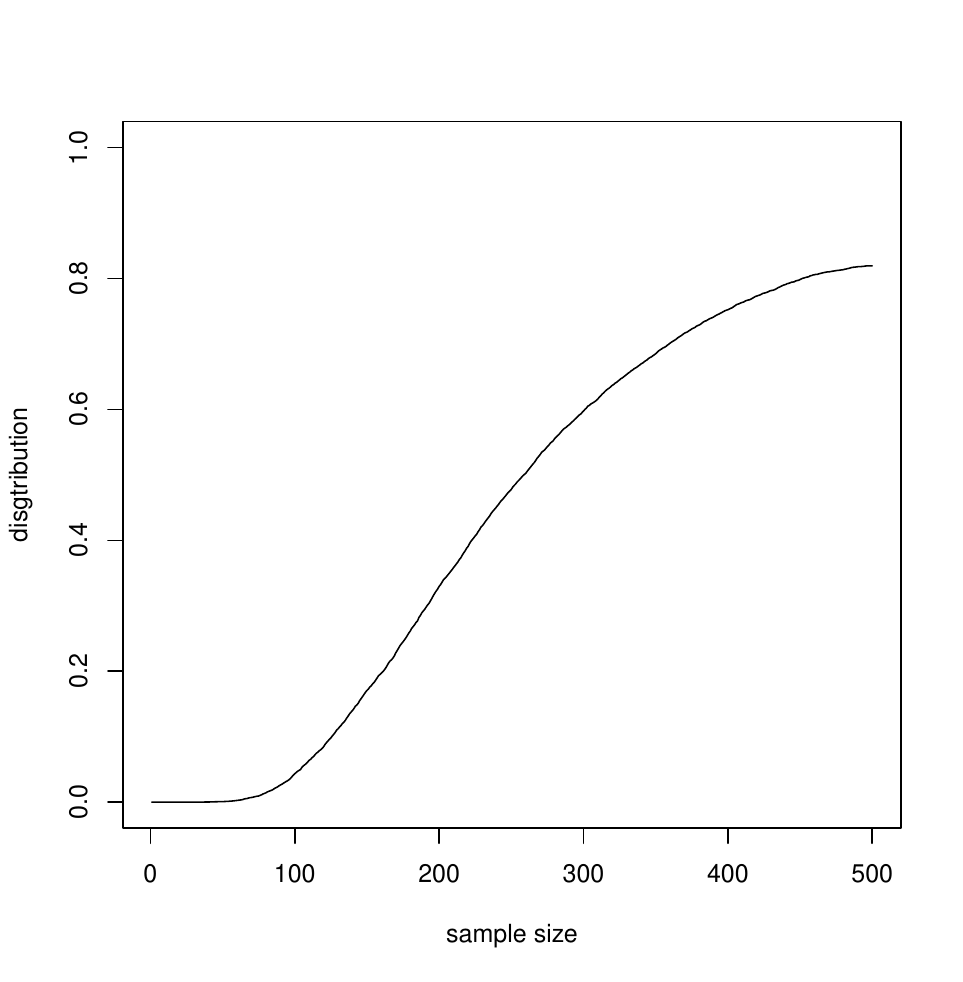}
\caption{Distribution of stopping times with $\gamma=0.1$ and $\theta^*=0.13$  and $N=500$ }
\label{figsm7a}
\end{center}
\end{figure}
\end{center} 

\begin{center}
\begin{figure}[!htbp]
\begin{center}
\includegraphics[width=12cm,height=4cm]{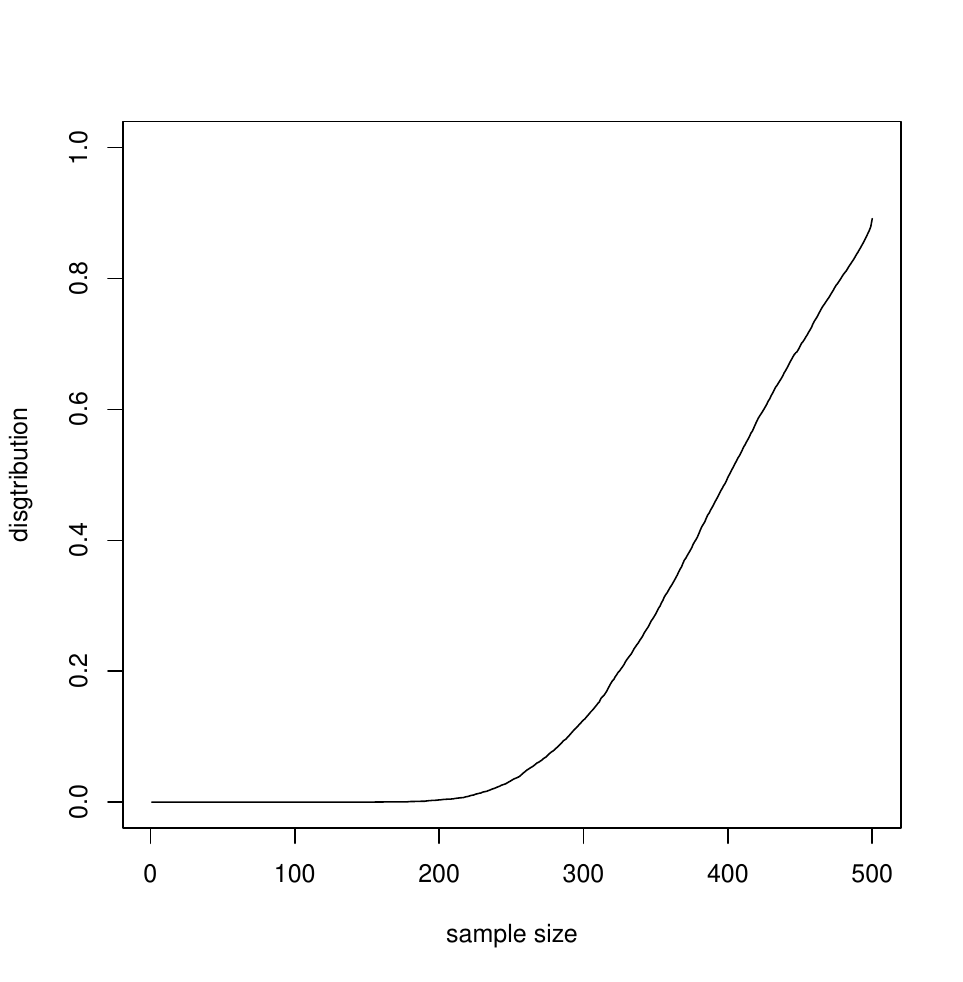}
\caption{Distribution of stopping times with $\gamma=0.95$, and $\theta^*=0.13$  and $N=500$.}
\label{figsm7b}
\end{center}
\end{figure}
\end{center} 




\noindent
Here we show that the power for the anytime valid test is upper bounded by the power for the fixed sample size $N$ test. The Type I error for the fixed sample size $N$ test is given by $\alpha$ and the power for true alternative $\theta^*>0$ is given by
$p_N(\alpha)=1-\Phi(z_\alpha-\sqrt{N}\theta^*).$
 We can obtain an upper bound for the anytime valid test, using the martingale inequality, which is given by
 $p_\gamma=p_N(\wt)/\gamma$ where $\wt=\alpha\gamma$. 
 
 \begin{lemma}
 It is that for all $\alpha$ and $\psi>0$,
 $$1-\Phi(z_\alpha-\psi)\geq \gamma^{-1}\bigg(1-\Phi(z_{\alpha\gamma}-\psi)\bigg)$$
 with equality only when $\gamma=1$.
 Here $\psi>0$ stands in for $\sqrt{N}\theta^*$.
 \end{lemma}

\vspace{0.1in}
\noindent
This is an easy proof using calculus. Hence, we can only recover the fixed sample size $N$ test when we take $\gamma=1$ which means the anytime valid test becomes the fixed sample size test. There can be no early rejection of the null hypothesis.  


\section{Relation to e-processes and betting odds}
\label{sec:S8}

E-processes can be interpreted as betting strategies whose capital is a
supermartingale under $H_0$.
While this guarantees anytime validity, it ties construction to likelihood ratios
and often leads to substantial power loss for fixed maximal sample size.
Predictive anytime-valid testing differs fundamentally by tracking the fixed-$N$
decision event directly.

There has been a substantial amount of research on sequential testing in recent years and the key device for ensuring the protection of Type I errors is martingales. Further, the crucial property of the martingale is the maximal inequality due to Doob: \cite{Doob53}, and the inequality due to Ville: \cite{Ville39}. 
The former results, which we rely on, are that if $X_1,X_2,\ldots$ is a submartingale, then for all $a>0$ it is that
$$P\left(\max_{1\leq i\leq n} X_i\geq a\right)\leq E\left(\max(X_n,0)\right)/a,$$
whereas if $X_1,X_2,\ldots$ is a non-negative supermartingale then 
$$P\left(\max_{1\leq i\leq n} X_i\geq a\right)\leq E(X_1)/a.$$
The latter result is that if $X_0,X_1,X_2,\ldots$ are a non-negative supermartingale, then for all $a>0$ it is that
$$P\left(\sup_n X_n\geq a\right)\leq E(X_0)/a.$$
These results can be used for controlling a Type I error when tests are performed sequentially. See for example, \cite{Shafer11}, \cite{Vovk21}, \cite{Grunwald24}, \cite{ramdas2023game}.

Prior to these papers, the essence of the idea behind e-processes and test martingales had been introduced by \cite{Robbins70} and demonstrated for clinical trials in \cite{Jennison89}. In \cite{Robbins70} the author introduced what would be recognized as a Bayes factor, define
$z_n=\tilde{g}_n(X_{1:n})/g_n(X_{1:n}),$
where $g_n$ and $\tilde{g}_n$ are density functions on $X_{1:n}$,
and was interested in the \cite{Ville39} result, namely
\begin{equation}\label{robbins}
P_{g_n}(z_n\geq\epsilon\,\,\mbox{for some}\,\,n\geq 1)\leq 1/\epsilon.
\end{equation}
In particular \cite{Robbins70} considered $g_n(X_{1:n})=\prod_{i=1}^n f(X_i\mid\theta_0)$ and
$$\tilde{g}_n(X_{1:n})=\int \prod_{i=1}^n f(X_i\mid\theta)\,\pi(\theta)\,d\theta$$
for some density function $\pi(\theta)$. Of particular interest was a sequence of confidence intervals $I_n(\epsilon)$,
derived from (\ref{robbins}), for which
$$P_{\theta}\bigg(I_n(\epsilon)\ni \theta\,\,\mbox{for all}\,\,n\geq 1\bigg)\geq 1-\epsilon.$$
A number of illustrations are presented in \cite{Robbins70}. Given the well known connection between confidence intervals and tests, \cite{Robbins70} also discusses sequential tests.
Further, \cite{Jennison89} adopt the settings of \cite{Robbins70} to a finite framework for which there is a maximal sample size  and focus on clinical trials. Specifically they use the sequence of confidence intervals of the type introduced by \cite{Robbins70} to perform sequential tests. However, only a few interim analyses can be performed due to the difficulty in setting appropriate critical regions.

We differ from \cite{Jennison89} in the respect that our predictive anytime-valid test allows for testing at every $n \le N$. This by-passes the need to pre-specify a rigid interim testing schedule and, for example, we can incorporate streaming data which is becoming more prevalent\footnote{https://engineering.atspotify.com/2023/07/bringing-sequential-testing-to-experiments-with-longitudinal-data-part-1-the-peeking-problem-2-0}


Here we describe the ideas behind what has become known as ``safe testing" which does allow for testing for streaming data. Suppose $(S_n)_{n\geq 0}$ is a sequence of observable test statistics which under the simple null hypothesis have expectations bounded by 1, i.e. $E_{p_0}(S_n)\leq 1$ for all $n$, where $p_0$ is the posited distribution of observables under the null hypothesis. Also assume the $S_n$ form a supermartingale. Then
Ville's inequality provides
$P\left(S_n\geq 1/\alpha\mid H_0\right)\leq \alpha$
for all $\alpha>0$. A sequential test with optional stopping will have the Type I error bounded by $\alpha$. For example, suppose the simple null hypothesis is that a sample is i.i.d. from $p_0$ and the simple alternative is that observations come from density $p_1$. Define
$$S_n=\prod_{i=1}^n \frac{p_1(X_i)}{p_0(X_i)}$$
so that $E_{p_0}(S_n)=1$ for all $n$. It is now possible to reject the null hypothesis at any sample size $n$ whenever $S_n\geq 1/\alpha$. So, if a typical $\alpha=0.05$ is chosen, the null hypothesis is rejected whenever $S_n\geq 20$.  

For a composite alternative hypothesis, say $p\in \Omega$, the recommendation is to use a Bayes factor of the form
$$B_n=\frac{\int_\Omega \prod_{i=1}^n p(X_i)\,\pi(dp) }{\prod_{i=1}^n p_0(X_i)}$$
which remains a martingale under the null hypothesis. Hence, safe testing is again possible due to Ville's inequality.

Without linking to a formal test with the notion of Type I errors, an alternative concept is required for deciding on the weight of evidence in favor or otherwise of a hypothesis. A notion of ``betting" has therefore been introduced into the e-process framework as a means by which to discuss evidence.
See, for example, \cite{Shafer21}.
We have no need for such a concept since we ground our decisions using a classical test framework and probability.

For a comparison we return to the normal mean illustration though now we consider the alternative hypothesis $H_1:\theta\ne 0.$ 
For the existing anytime-valid approach the test statistic is given by
$$E_n=\exp\left\{\half T_n^2/(n+1)\right\}/\sqrt{n+1}$$
where $T_n=\sum_{i=1}^n X_i$. This is based on a prior for $\theta$ which is standard normal. The null hypothesis is rejected if ever $E_n\geq 1/\alpha$.
For the predictive procedure we have
$$Q_n=1-\Phi\left(\frac{\sqrt{N}z_{\wt/2}-T_n}{\sqrt{N-n}}\right)+\Phi\left(\frac{-\sqrt{N}z_{\wt/2}-T_n}{\sqrt{N-n}}\right)$$
and the null hypothesis is rejected if ever $Q_n\geq\gamma$ where $\alpha=\wt\gamma$. 

\section{Further tests}
\label{sec:S10}

\subsection{Student $t$ test}\label{sec:student}

Here we consider an example not included in the paper of \cite{Jennison89} and neither is it part of the stochastic curtailment literature.
We consider the well known Student $t$ test for a normal mean with unknown variance. The statistic of interest is
$T_N=\sqrt{N}\bar{X}_N/S_N$ with the null hypothesis being the normal mean with unknown variance is 0. It is not possible to include this in the stochastic curtailment setting due to the lack of knowledge of the distribution of $T_N$ given $X_{1:n}$ for $n<N$.

We define the predictive test statistic given a sample of size $n$ as $T_N^{(n)}=\sqrt{N}\bar{X}_N/S_n$
where $S_n$ is the sample standard deviation from the sample of size $n$. Then, for a one sided test,
$$Q_n=P\bigg(T_N^{(n)}\geq c\bigg),$$
for some critical value $c$,
where the $X_{n+1:N}$ are sampled from the normal distribution with zero mean and variance $S_n^2$.
Then, some straightforward algebra gives
$$Q_n=1-\Phi\left(\frac{c\sqrt{N}-\sqrt{n}T_n}{\sqrt{N-n}}\right),$$
for $n<N$ and 
$Q_N=1(T_N>c)$. While $(Q_n)$ is not a martingale, the joint distribution of $(Q_{n_0},\ldots,Q_n)$ does not depend on the unknown $\sigma$. We introduce $n_0$ here as the minimal interim analysis time so as to ensure suitable estimation of the unknown parameter. Theoretically $n_0$ could be 2 but in practice it would be taken to be larger.
Now $Q_n\geq\gamma$ is equivalent to
$T_n\geq c\sqrt{N/n}-\sqrt{N/n-1}\Phi^{-1}(1-\gamma)$.
We can easily find the connection between $\alpha$, the Type I error, and $\gamma$ using Monte Carlo sampling.

\begin{figure}
\center
\includegraphics[width=12cm,height=6cm]{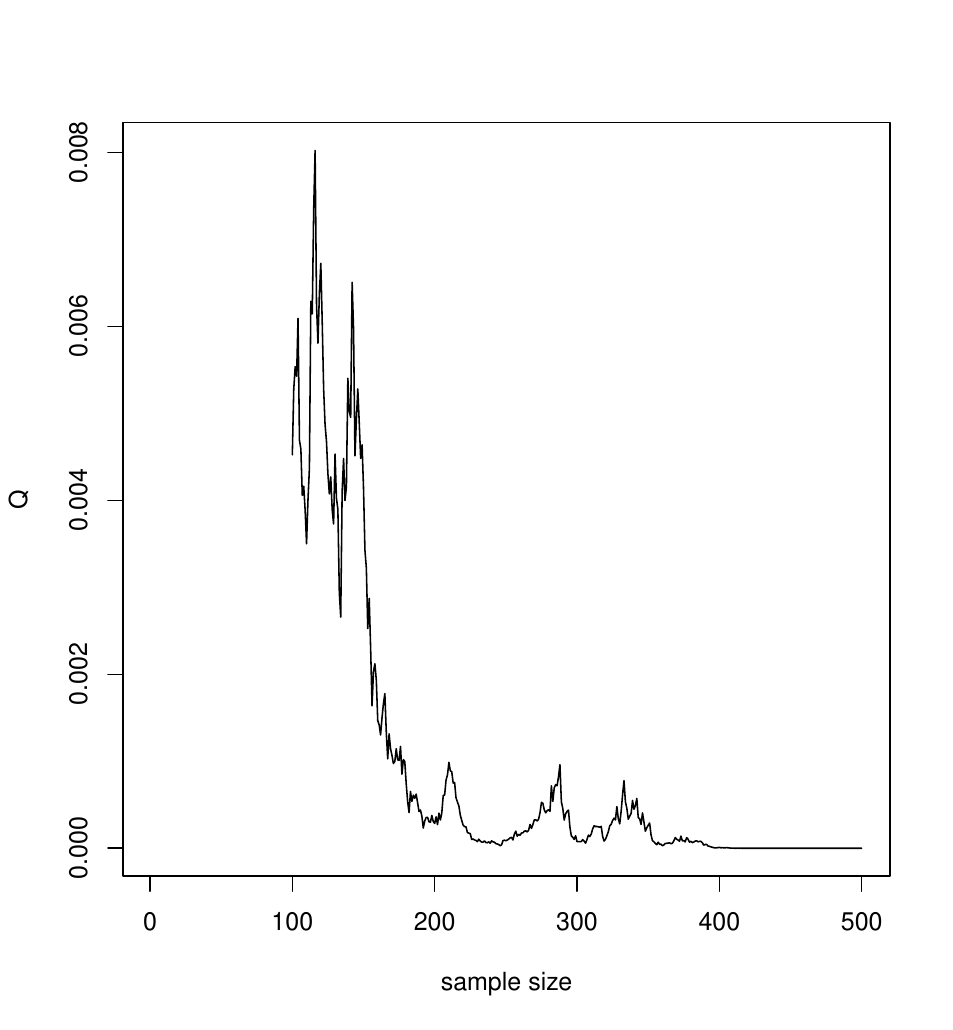}
\caption{Sample path of $(Q_n)_{n\geq 100}$ for Student $t$ test under null hypothesis}
\label{fig20}
\end{figure}

\begin{figure}
\center
\includegraphics[width=12cm,height=6cm]{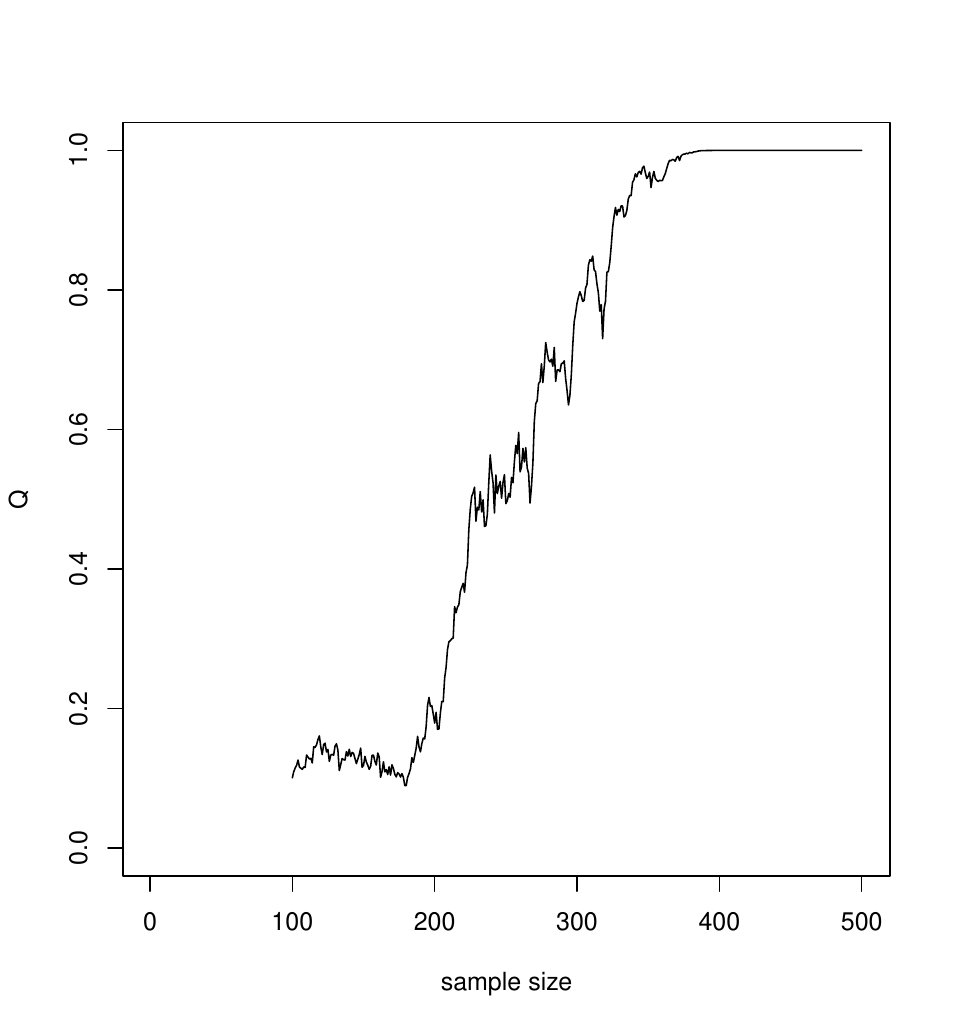}
\caption{Sample path of $(Q_n)_{n\geq 100}$ for Student $t$ test under alternative hypothesis with true mean 0.2}
\label{fig21}
\end{figure}

For example, taking $n_0=100$ and $N=500$ we can show that $\alpha=0.047$, i.e.
$P (\max_{n_0\leq n\leq N}Q_n\geq\gamma)=0.047$, using a Monte Carlo sample size of 10,0000,
when taking $c=\Phi^{-1}(1-0.05\gamma)$ and $\gamma=0.95$. This does not change with different true $\sigma$.
Fig.~\ref{fig20} and \ref{fig21} present sample paths of $(Q_n)$ for $n\geq 100$ with $N=500$, under the null hypothesis and alternative hypothesis (true mean being 0.2). The variance was taken to be 1 though the results will not depend on this value. 

To elaborate on the Monte Carlo sampling, for each Monte Carlo sample we generate $X_{1:N}$ to be i.i.d. from the normal distribution with mean 0 and variance 1. We then compute $T_{n_0:N}$ and then $Q_{n_0:N}$. These have the same joint distributions regardless of the choice of the variance as 1. We then check for $\max_{n_0\leq n\leq N}\{Q_n\}\geq\gamma$. Repetition allows for the Monte Carlo approximation to
$P(\max_{n_0\leq n\leq N}\{Q_n\}\geq\gamma)$. We could either fix $\alpha$ and find the corresponding $\gamma$ or alternatively we could fix $\gamma$ and check the probability is upper bounded by $\alpha$. The former would make for the more powerful test.

\subsection{Normal means with known variance}
Here we consider the test of two normal means, so the  $(X_n)$ are i.i.d. $N(\cdot\mid\theta_X,\sigma^2)$ and $(Y_n)$ are i.i.d. $N(\cdot\mid \theta_Y,\sigma^2)$ and the goal is to test $H_0:\theta_X=\theta_Y$ versus $H_1:\theta_X\ne \theta_Y$. We assume $\sigma$ is known and, without loss of generality, is equal to 1. The relevant statistic from a maximal sample size $N$ is $T_N=|\bar{X}_N-\bar{Y}_N|$ and the null hypothesis is rejected if $T_N\geq c$ for some $c>0$. Specifically, for a Type I error of $\wt$ it is that $c=z_{\wt/2}\sqrt{2/N}$.

We now define $Q_n=\mathbb{P}_0(T_N^{(n)}>c)$ where
$$T_N^{(n)}=\left|\frac{n(\bar{X}_n-\bar{Y}_n)+\sum_{i=n+1}^N Z_i}{N}\right|$$
and the $Z_{n+1:N}$ are independent normal variables with zero mean and variance 2.
Therefore,
$$Q_n=1-\Phi\left(\frac{cN-n(\bar{X}_n-\bar{Y}_n)}{\sqrt{2(N-n)}}\right)+\Phi\left(\frac{-cN-n(\bar{X}_n-\bar{Y}_n)}{\sqrt{2(N-n)}}\right).$$
It is straightforward to confirm that this is a martingale and hence we provide an anytime valid test with Type I error $\alpha$ for the two-sample problem based on rejecting the null hypothesis if ever $Q_n\geq\gamma$ with $\wt=\alpha\gamma$. 

\subsection{Normal means with unknown common variance}\label{sec:Normalmean}

In this subsection we consider a two sample normal mean test  assuming the common variance $\sigma^2$ is unknown. The maximal sample size test comes from the test statistic
$$T_N=\frac{\bar{X}_N-\bar{Y}_N}{S_N\sqrt{2/N}}$$
which under the null hypothesis $H_0:\theta_X=\theta_Y$
has a standard Student $t$ distribution with $2N-2$ degrees of freedom. Here $S_N$ is the pooled sample variance which for $2\leq n\leq N$ is defined by
$$S_n^2=(n-1)(S_{n,X}^2+S_{n,Y}^2)/(2(n-2))$$
where $S_{n,X}^2$ and $S_{n,Y}^2$ are the sample variance for the $X$ sample and $Y$ samples, respectively. 
The final maximal sample size test based on $(X_{1:N},Y_{1:N})$ is to reject the null hypothesis the means are identical if $|T_N|>c$ for some $c$ which will determine the Type I error.


We define the predicted statistic $T_N^{(n)}$ similar to how it was done in the one sample Student $t$ test, i.e.
$$T_N^{(n)}=\frac{\sqrt{N}(\bar{X}_N-\bar{Y}_N)}{S_n\sqrt{2}}$$
which can be written as
$$T_N^{(n)}= \frac{\sqrt{1/N}\{n\bar{X}_n-n\bar{Y}_n+\sum_{i=n+1}^N(X'_i-Y'_i)\} }{S_n\sqrt{2}}$$
where the $(X'_{n+1:N})$ and $(Y'_{n+1:N})$ are i.i.d. normal with mean 0 and variance $S_n^2$.  Hence,
$$T_N^{(n)}=\sqrt{n/N}\,T_n+\sqrt{1-n/N}\,Z$$
for some independent standard normal random variable.
We now define 
$Q_{n}=P\left(T_{N}^{(n)}\geq c\right)$ 
for $n\geq 2$, so that the $T_N^{(n)}$ is properly defined.
Then
$$Q_n=1-\Phi\left(\frac{c\sqrt{N}-\sqrt{n}T_n}{\sqrt{N-n}}\right),$$
for $n<N$ and $Q_N=1(T_N\geq c)$.
The evaluation of settings for the appropriate Type I error are as in section~\ref{sec:student}.


\begin{center}
\begin{figure}[!htbp]
\begin{center}
\includegraphics[width=12cm,height=5cm]{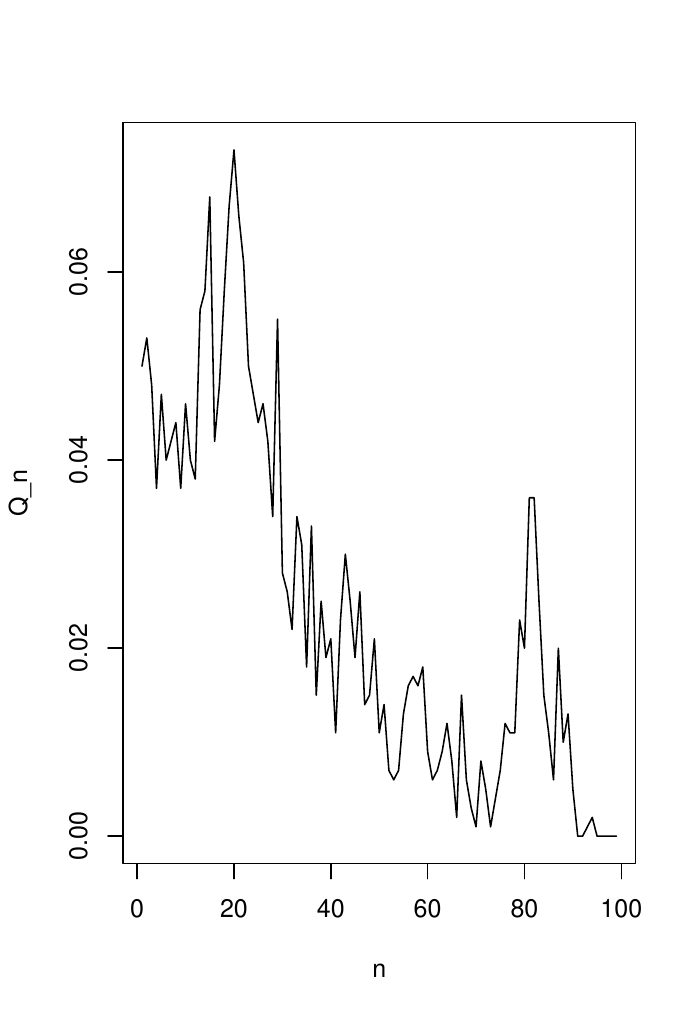}
\caption{Plot of $Q_n$ for normal means with unknown common variance under null hypothesis with the difference of means being 0.}
\label{fig24}
\end{center}
\end{figure}
\end{center} 

\begin{center}
\begin{figure}[!htbp]
\begin{center}
\includegraphics[width=12cm,height=5cm]{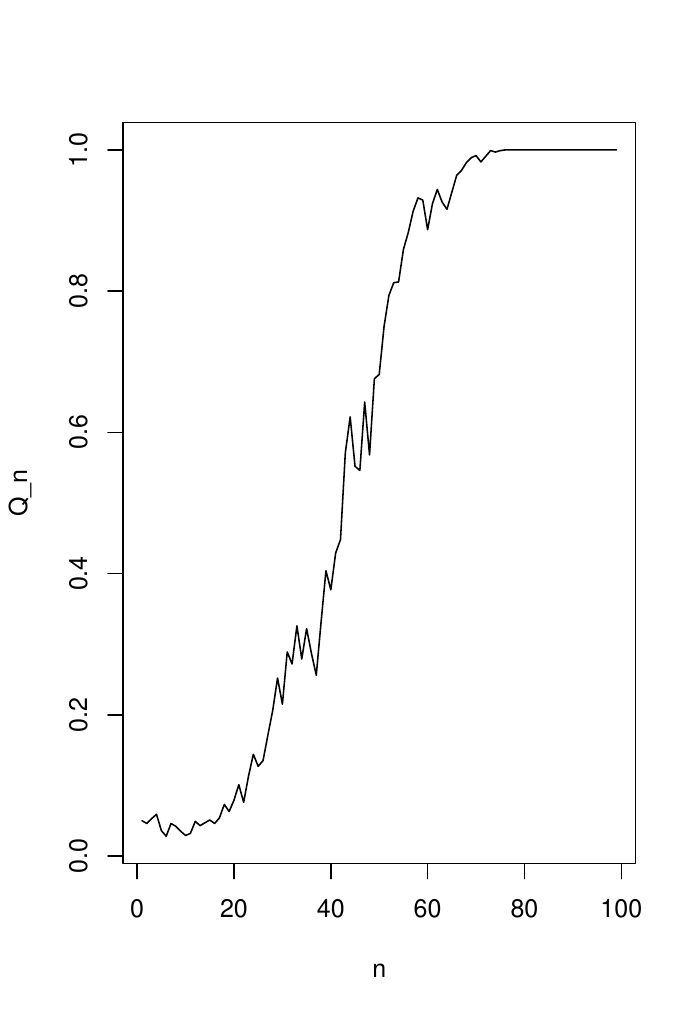}
\caption{Plot of $Q_n$ for normal means with unknown common variance under null hypothesis with difference of means being 1}
\label{fig25}
\end{center}
\end{figure}
\end{center} 

\noindent By means of illustration we take $N=100$ and set the Type I error to be $\wt=0.05$ which requires $c=2.05$. With both groups having standard normal distributions we compute $Q_n$ for $n=2,\ldots,N-1$. A sample path is provided in Fig.~\ref{fig24}.
For the alternative hypothesis we now take one of the groups to have a mean of 1, all other setting remaining as before. A sample path of $Q_n$ is presented in Fig.~\ref{fig25}.

\subsection{Asymptotic anytime valid test}

May fixed sample size tests rely on asymptotic distributions of statistics. In this section we look at the MLE. As is well known, under mild regularity conditions it is that $\widehat\theta_n\sim N(\theta,\sigma^2_n)$ where $\sigma^2_n=I(\theta)^{-1}/n$ and $I(\theta)$ is the Fisher information. If the test is $H_0:\theta=\theta_0$ we need to be able to predict the $\widehat\theta_N$ conditional on $\widehat\theta_n$.
Asymptotically, $\mbox{Cov}(\widehat\theta_n,\widehat\theta_N)=\sigma^2_N$ and so asymptotically, based on standard bivariate normal theory, it is that
$$\widehat\theta_N=\theta_0(1-n/N)+(n/N)\widehat\theta_n+N^{-1}\sqrt{N-n}Z/\sqrt{I(\theta_0)},$$
where $Z$ is an independent standard normal variable. So we would take $T_n=\widehat\theta_n$ and $T_N^{(n)}=\widehat\theta_N$ as given above.
Then we define
$$Q_n=P\left(T_N^{(n)}\geq c\right)=1-\Phi\left(\frac{(cN-\theta_0(N-n)-n\,T_n)\sqrt{I(\theta_0)}}{\sqrt{N-n}}\right).$$
Here the sequence $(Q_n)$ will not be a martingale but as has been previously highlighted we can easily find the $\gamma$ for which
$P(\max_{n_0\leq n\leq N}Q_n\geq\gamma)=\alpha$ using Monte Carlo methods.
The practical application to determine $\gamma$ involves running multiple Monte Carlo runs, one of which would be to sample $X_{1:N}$ to be i.i.d. from $f(\cdot\mid\theta_0)$ and from this to construct the sequence $(T_n)$ from which sequence $(Q_n)$ follows. It is then easy to check on the event  $\{\max_{n\leq N} Q_n\geq\gamma\}$.

For illustration, we use the Poisson distribution and test the hypothesis the mean value is 1.
Hence, taking $N=500$, for a single Monte Carlo run, we sample $X_{1:N}$ and compute $Q_n$. Note the $I(\theta_0)=1/\theta_0=1$ in this case. Over the full Monte Carlo runs for the  $(Q_n)$  we approximate $p(\gamma)=P(\max_{n\leq N} Q_n\geq\gamma)$ for $0<\gamma<1$.
This is a decreasing function and taking $c$ such that the Type I error for the fixed sample size $N$ test is $0.95\times 0.05$; i.e.
$c=\Phi^{-1}(1-0.95\times 0.05)/\sqrt{N}$, we find the $\gamma$, which yields a Type I error of $\alpha=0.05$, is $\gamma=0.90$ and note that this is smaller than the value which we would have obtained (if we could) using the martingale inequality, which would be 0.95. We took the Monte Carlo sample size to be 10,000.

\subsection{Nonparametric model}

To demonstrate the generality of predictive sampling anytime valid tests, we show how  to perform a nonparametric test. 
Here we consider the test of the type $H_0:F=F_0$ where $F_0$ is a specified distribution function.
Many tests rely on comparing the empirical distribution function $F_n$ with $F_0$, such as the Kolmogorov-Smirnov test which rejects the null hypothesis if $T_N=\sup_x|F_N(x)-F_0(x)|\geq c$, for some $c$ which determines the Type I error.
We define
$$F_N^{(n)}(x)=\frac{nF_n(x)+(N-n)\,F'_{N-n}(x)}{N}$$
where the $F_n$ is observed and the $F'_{N-n}$ is generated from $X'_{n+1:N}$ which are coming i.i.d. from $F_0$. 
This defines the predictive statistic $T_N^{(n)}$. Then
$$Q_n=P\left(\sup_x|F_N^{(n)}(x)-F_0(x)|\geq c\right).$$
It is straightforward to show that $(Q_n)$ is a martingale under the null hypothesis. 
It is also straightforward to sample $(Q_n)$ to allow for a Monte Carlo evaluation of the $\gamma$ for which
$P(\max_{n\leq N} Q_n\geq\gamma)=\alpha$ which will lead to a more powerful anytime valid test compared with the $\gamma$ determine via the martingale inequality.

\subsection{Bernoulli outcomes}\label{sec:Bernoulli}

For this test, to be observed are $X_{1:N}$ and $Y_{1:N}$ which are Bernoulli random variables and the null hypothesis is that the mean is common to both groups. The $N$ sample test is based on the statistic
$$T_N=\frac{|S_{NX}-S_{NY}|}{\sqrt{2N\widehat{p}_N(1-\widehat{p}_N)}},$$
where $P_N$ is the common estimator; i.e.
$\widehat{p}_N=(S_{NX}+S_{NY})/(2N)$ and $S_{NX}=\sum_{1=1:N}X_i$ and $S_{NY}=\sum_{i=1:N}Y_i$. The null hypothesis is rejected if $T_N>c$ which will be chosen to ensure a suitable Type I error, e.g. $\alpha=0.05$ 

Similar to subsection \ref{sec:Normalmean} we define the predicted test statistic at sample size $N$ given current sample size $n$ by
$$T_N^{(n)}=\frac{S_{NX}-S_{NY}}{\sqrt{2N\widehat{p}_n(1-\widehat{p}_n})}$$ 
where
$$S_{NX}-S_{NY}=S_{nX}-S_{nY}+\sum_{i=n+1}^N (X'_i-Y'_i)$$
and the $X'_{n+1:N}$ and $Y'_{n+1:N}$ are i.i.d. Bernoulli with parameter $\widehat{p}_n$.
Hence, using the normal approximation motivating the test statistic and critical values, we have
$$T_N^{(n)}=\sqrt{n/N}\,T_n+\sqrt{1-n/N}\,Z$$
for some independent standard normal random variable. 
As usual we define
$Q_n=P(T_N^{(n)}\geq c)$
which is given by
$$Q_n=1-\Phi\left(\frac{c\sqrt{N}-\sqrt{n}\,T_n}{\sqrt{N-n}}\right),$$
with
$Q_N=1(T_N\geq c)$.  While $(Q_n)$ is not a martingale it is possible to determine a critical value for rejecting the null hypothesis with Type I error $\alpha$ using Monte Carlo methods; see section~\ref{sec:student}.


\subsection{Group means for binary data}

In this section we consider two groups where each individual provides a binary outcome. The test statistic is given by
$$T_{N,M}=\frac{\widehat{p}_X-\widehat{p}_Y}{\sqrt{\widehat{p}(1-\widehat{p})(1/N+1/M)}}$$
where $\widehat{p}_X=n_X/N$, $\widehat{p}_Y=n_Y/M$ and 
$\widehat{p}=(n_X+n_Y)/(N+M)$ and  $n_X$ and $n_Y$ are the number of successes from group $X$ and group $Y$, respectively. Under the null hypothesis that $p_X=p_Y$ the $T_{N,M}$ statistic is approximately standard normal.
After time $t_j$ assume there are $N_j$ observations from group $X$ 
and $M_j$ observations from group $Y$ with the number of successes as $n_{Xj}$ and $n_{Yj}$, respectively.
We then estimate the pooled mean as
$$\widehat{p}_j=\frac{n_{Xj}+n_{Yj}}{N_j+M_j}.$$
We now define 
$T_{N,M}^{(j)}$ to be the predicted $T_{N,M}$ at stage $j$ by taking
$X'_{N_j+1:N}$ to be i.i.d. Bernoulli with probability $\widehat{p}_j$ and likewise for $Y'_{M_j+1:M}$.  This gives us a full data set for which we can obtain the prediction for $T_{N,M}$. Repeated simulation using Monte Carlo methods gives
an estimator for $Q_j=P(T_{N,M}^{(j)}\geq c)$.

\begin{figure}
\center
\includegraphics[width=12cm,height=6cm]{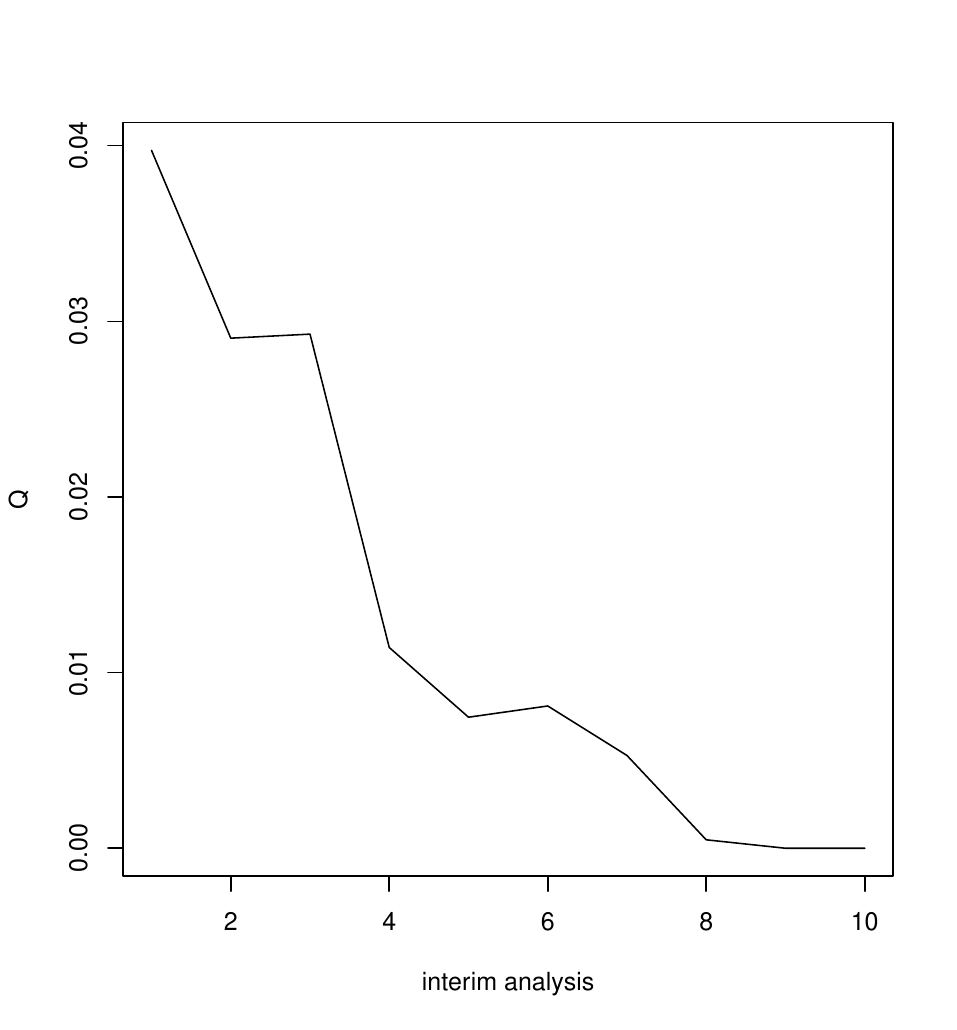}
\caption{Sample path of $Q_j$ with $p_X=p_Y=0.3$}
\label{fig18}
\end{figure}

\begin{figure}
\center
\includegraphics[width=12cm,height=6cm]{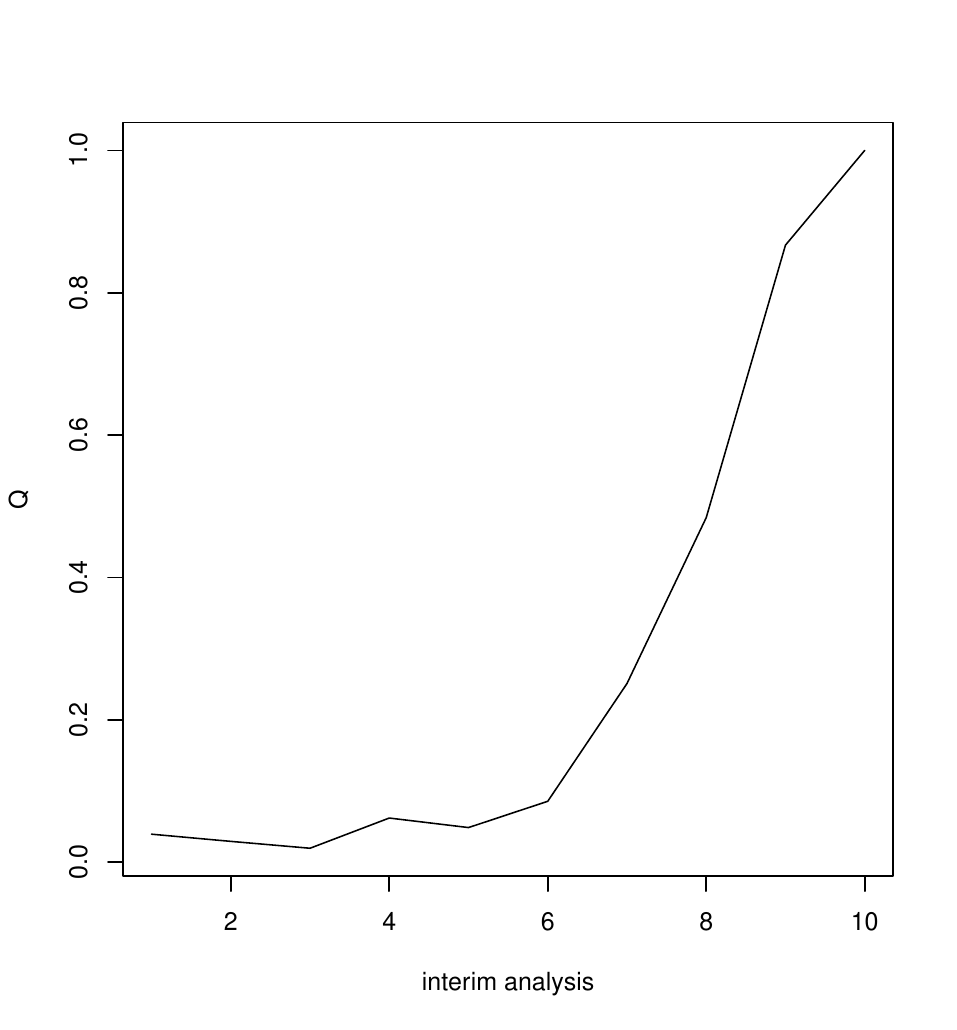}
\caption{Sample path of $Q_j$ with $p_X=0.3$ and $p_Y=0.4$}
\label{fig19}
\end{figure}

For an illustration using simulated data we took $N=M=250$ and the times for interim analysis were $t_j=j/10$ for $j=1,\ldots,10$ while the observation times for the 500 individuals were generated uniformly on $(0,1)$. At each $t_j$ the missing observations are sampled using the current pooled estimator $\widehat{p}_j$. From the completed data we then compute the predicted $T_{N,M}$ statistics. Repeated simulation of this leads to $Q_j$ where we take $c=\Phi^{-1}(1-\alpha/2)$ with $\alpha=0.05$. 
A sample path for $(Q_j)$ with $p_X=p_Y=0.3$ is presented in Fig.~\ref{fig18} while a path with $p_X=0.3$ and $p_Y=0.4$ is presented in Fig.~\ref{fig19}.

\subsection{Test based on Wilks statistics}

Here we define the predictive statistic $T_N^{(n)}=T(\theta(X_{1:n},X'_{n+1:N}),\widehat\phi_n)$ where $X'_{n+1:N}$ 
come from the distribution defined by the null hypothesis with nuisance parameter estimated by $\widehat\phi_n$ 
and $\widehat\phi_n=\phi(X_{1:n})$. We then define
$Q_n=P(T_N^{(n)}\in C)$. So now $Q_n$ is not a martingale but can be evaluated for all $n$ and Monte Carlo methods can be used to find the $\gamma$ for which $P(\max_{n\leq N} Q_n\geq\gamma)=\alpha$.

A general approach to constructing anytime valid tests involves Wilks' statistics (\cite{Wilks1938}). The statistic is given by
$$T_N=2\sum_{i=1}^N\log \frac{f(X_i\mid\widehat\theta_N,\widehat\phi_N)}{f(X_i\mid \theta_0,\widehat\phi_N)}$$
for hypothesis $H_0:\theta=\theta_0$ in the presence of nuisance parameter $\phi$ which is estimated via the MLE. The distribution of $T_N$ is approximately $\chi^2_p$ where $p$ is the dimension of $\theta$.
We define $T_N^{(n)}$ using the current estimator for $\phi$, i.e. $\widehat\phi_n$. Hence, we take $X'_{n+1:N}$ to be i.i.d. from $f(\cdot\mid\theta_0,\widehat\phi_n)$ and define
$$T_N^{(n)}=2\sum_{i=1}^N\log \frac{f(X_i\mid\widehat\theta_N,\widehat\phi_n)}{f(X_i\mid \theta_0,\widehat\phi_n)}$$
where $\widehat{\theta}_N$ is evaluated using $(X_{1:n},X'_{n+1:N})$. Then we define
$Q_n=P(T_N^{(n)}\geq c)$
where $c=F_{\chi^2_p}^{-1}(1-\alpha\gamma)$
and the null hypothesis is rejected if ever $Q_n\geq\gamma$. Since the $(Q_n)$ does not form a martingale we evaluate $\gamma$ using Monte Carlo methods.


\begin{figure}
\center
\includegraphics[width=12cm,height=6cm]{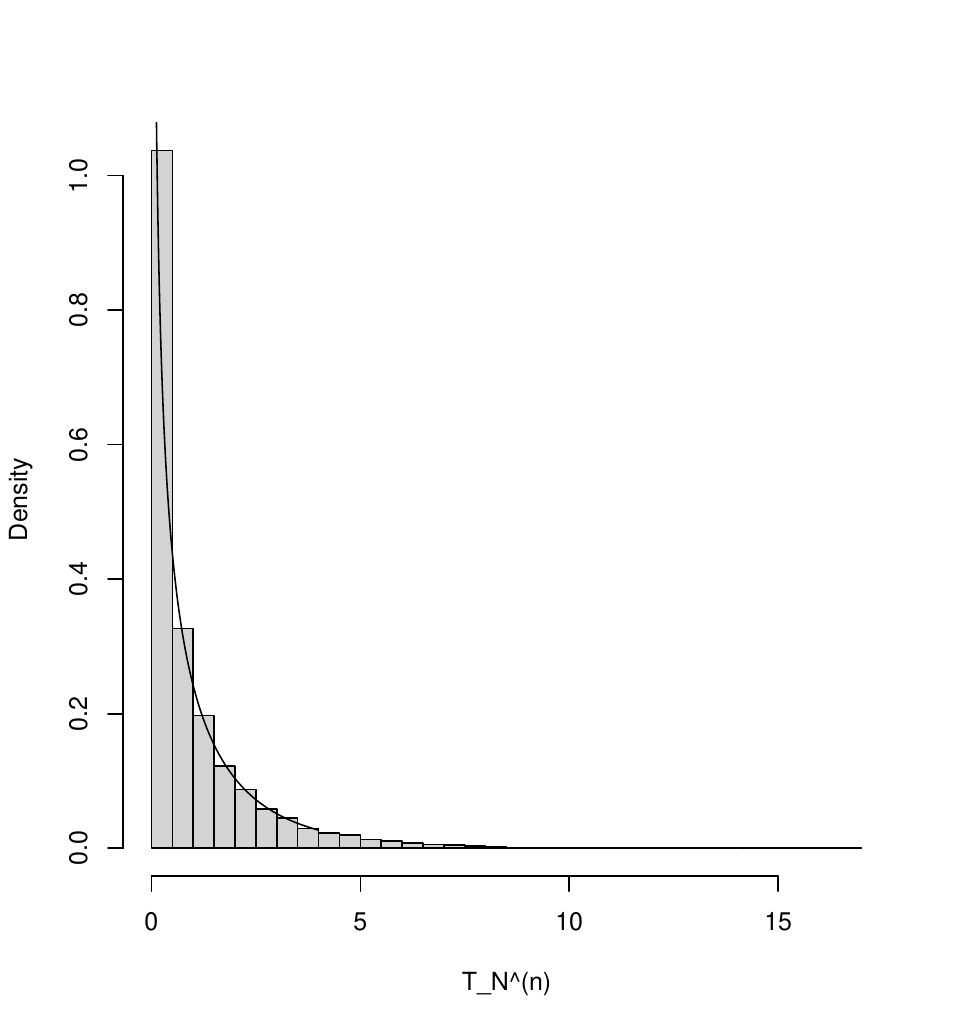}
\caption{Comparison of samples from $T_N^{(n)}$ for predictive Wilks statistic with the $\chi_1^2$ density function}
\label{fig17}
\end{figure}

To look at how good the predictive statistic is, we compare the distribution of the $T_N^{(n)}$ statistic with the $\chi^2_1$ distribution. The comparison is provided in Fig.~\ref{fig17}. The relevant summaries from 10,000 $T_N^{(n)}$ samples, with $N=500$ and $n=250,$ are that the sample mean is 1.007, the true value should be 1, and the sample variance is 2.05, the true value being 2.

\subsection{Composite null hypotheses}\label{sec:Composite} 


\subsubsection{Normal mean}  We start with a normal mean example where the variance is known (and without loss of generality is taken to be 1) and the hull hypothesis is $H_0:\theta\leq \theta_0$ versus $H_1:\theta>\theta_0$. Also  without loss of generality we take $\theta_0=0$. 
The statistic at sample size $N$ is given by $T_N=\sqrt{N}\bar{X}_N$ and the critical value is given by $c=\Phi^{-1}(1-\alpha)$ since under the hull hypothesis it is that
$P(T_N\geq c)\leq\alpha$.
We define
$$T_N^{(n)}=\left(\sum_{i=1}^n X_i+\sum_{i=n+1}^N X'_i\right)/\sqrt{N},$$
where the $X'_{n+1:N}$ are i.i.d. normal with mean $0$ and variance 1. So
$Q_n=P(T_N^{(n)}\geq c)$ for $c=\Phi^{-1}(1-\wt)$ for some $\wt<\alpha$.
 
The aim now is to show that $(Q_n)$ is a supermartingale. To this end, consider
$$Q_{n+1}=P\left(\sum_{i=1:n}X_i+X_{n+1}+\sum_{i=n+1:N}X'_i\geq c \Big| X_{1:n+1}\right),$$
which we now take the expectation of, conditional on $X_{1:n}$,
with respect to the null hypothesis and the correct distribution within the null constraint. Regardless of which one this is, the expectation is upper bounded by
$$Q_n=P\left(\sum_{i=1:n} X_i+X'_{n+1}+\sum_{i=n+1:N}X'_i\geq c \Big|  X_{1:n}\right).$$ 
Therefore, $E(Q_{n+1}\mid X_{1:n})\leq Q_n$ and so we recover the same result as with the martingale, i.e. if we reject the null hypothesis whenever $Q_n\geq \gamma$, we do so with a Type I error bounded by $\wt/\gamma$.

If we wish to avoid the martingale approach and the approximations that come with it, we can define the predicted statistic
$$T_N^{(n)}=\left(n\bar{X}_n+\sum_{i=n+1}^N X'_i\right)/\sqrt{N}$$
with the $X'_{n+1:N}$ taken to be i.i.d. normal with mean $\widetilde{X}_n=\min\{0,\bar{X}_n\}$, which is the estimator of $\theta$ restricted to $H_0$, and variance 1. Then
$$\sqrt{N}\,T_N^{(n)}=\sqrt{n}T_n+(N-n)\widetilde{X}_n+\sqrt{N-n}Z$$
for some independent standard normal random variable $Z$. So
$$Q_n=1-\Phi\left(\frac{c_{\wt}\sqrt{N}-\sqrt{n}T_n-(N-n)\widetilde{X}_n}{\sqrt{N-n}}\right),$$
where $c_{\wt}$ is so that $P(T_N\geq c_{\wt})=\wt$.
Here $(Q_n)$ is not a martingale.

\subsubsection{Bayes factors}\label{sec:BayesF}

For the model $f(x\mid\theta)$ we consider the null hypothesis $H_0:\theta\in\Theta_0$ versus $H_1:\theta\in\Theta_1$. For this type of hypothesis, a Bayesian would proceed by assigning prior distributions $\pi_0$ and $\pi_1$ restricted to $\Theta_0$ and $\Theta_1$, respectively. With the full maximal sample of size $N$ the Bayes factor is given by
$$B_N=\frac{\int \prod_{i=1}^N f(X_i\mid\theta)\,\pi_1(\theta)\,d\theta}{\int \prod_{i=1}^N f(X_i\mid\theta)\,\pi_0(\theta)\,d\theta}$$
and the null hypothesis is rejected if $B_N>c$. The idea here is that the sequence $(X_{1:N})$ is exchangeable with prior $\pi_0$ 
under the null hypothesis. Hence, we define, in the usual way,
the sample size $n$ Bayes factor
$$B_N^{(n)}=\frac{\int \prod_{i=1}^n f(X_i\mid\theta)\prod_{i=n+1}^N f(X'_i\mid\theta)\,\pi_1(\theta)\,d\theta}{\int \prod_{i=1}^n f(X_i\mid\theta)\prod_{i=n+1}^N f(X'_i\mid\theta)\,\,\pi_0(\theta)\,d\theta},$$
where the $X_{1:n}$ are observed and the $(X'_{n+1:N})$ are from the
predictive model $p_0(X_{n+1:N}\mid X_{1:n})$ given by
$$p_0(X_{n+1:N}\mid X_{1:n})=\int \prod_{i=n+1}^N f(X_i\mid\theta)\,\pi_0(\theta\mid X_{1:n})\,d\theta$$
where $\pi_0(\theta\mid X_{1:n})$ is the usual Bayesian posterior distribution.
Under the null hypothesis, the sequence $(B_N^{(n)})$ is a martingale and so also will be 
$Q_n=P(B_N^{(n)}\geq c)$. Hence, the anytime valid test procedure rejects $H_0$ whenever $Q_n\geq \gamma$. 

\subsubsection{Classical tests}\label{sec:Classic}

Here we consider the same test as in subsection \ref{sec:BayesF} 
but without the use of prior distributions. Instead, we rely on MLEs. For a sample of size $N$ the outcome of a test for $\theta\in\Theta_0$ will rely on the value of 
$$T_N=\prod_{i=1}^N f(X_i\mid\widetilde\theta_N)/f(X_i\mid\widehat\theta_N)$$
where $\widehat\theta_N$ is the MLE restricted to lie in $\Theta_0$ and $\widetilde\theta_N$ is the MLE constrained to lie in $\Theta_1$. 

The appropriate null hypothesis can be adapted to
\begin{equation}\label{MLEhyp}
H_0:X_{n+1}\sim f(\cdot\mid\widehat\theta_n)\quad n=1,\ldots,N-1,
\end{equation}
where the sequence $(\widehat\theta_n)$ are the MLE constrained to lie in $\Theta_0$ from the $X_{1:n}$.
We view this hypothesis as equivalent to the aims of the hypothesis $H_0:\theta\in\Theta_0$. 
Indeed, \cite{ramdas2023game} used such a framework of relying on predictive densities conditional on plug-in estimators for e-values, following original work in \cite{Wasserman2020}. The key is that martingales are preserved and allow for a composite style null hypothesis.

Under the null hypothesis the sequence $(Q_n)_{n>1}$, given by
$Q_n=P(T_N^{(n)}\in C)$, where 
$$T_N^{(n)}=\frac{\prod_{i=1}^n f(X_i\mid \widetilde\theta_{i-1})\prod_{i=n+1}^N f(X'_i\mid\widetilde\theta_{i-1})}{
\prod_{i=1}^n f(X_i\mid \widehat\theta_{i-1})\prod_{i=n+1}^N f(X'_i\mid\widehat\theta_{i-1})},$$
is a martingale. Here, for $m>n$, 
$$\widehat\theta_m=\widehat\theta(X_{1:m},X'_{m+1:N})\quad\mbox{and}\quad \widetilde\theta_m=\widetilde\theta(X_{1:m},X'_{m+1:N}),$$
where the $X'_{i}\sim f(\cdot\mid \widehat\theta_{i-1})$ for $i=n+1,\ldots,N$, and $\widetilde\theta_0=\widehat\theta_0$ is taken as a boundary value between $\Theta_0$ and $\Theta_1$.  

Here we present an illustration of this framework with a normal mean with known variance 1 and look at $H_0:\theta\leq 0$ versus $H_1:\theta>0$. We recast the hypothesis as in (\ref{MLEhyp}).
In this case we do not define $T_N^{(1)}$ and for $n=2,\ldots,N-1$ we define
$$T_N^{(n)}=\prod_{i=1}^n \frac{N(X_i\mid\widetilde\theta_{i-1},1)}{N(X_i\mid \widehat\theta_{i-1},1)}\,\prod_{i=n+1}^N\frac{N(X'_i\mid\widetilde\theta_{i-1},1)}{N(X'_i\mid\widehat\theta_{i-1},1)},$$
where 
$$\widetilde\theta_n=\max\{0,\bar{X}_n\}\quad\mbox{and}\quad \widehat\theta_n=\min\{0,\bar{X}_n\},$$
and $X'_i\sim N(\widehat\theta_{i-1},1)$ for $i>n$. Under the null hypothesis it is that $(T_N^{(n)})_{n>1}$ is a martingale with expectation 1 and so we can therefore reject the null hypothesis with Type I error bounded above by $\wt$ whenever $T_{N}^{(n)}>1/\wt$, $n=2,\ldots,N-1$. This would be the e-process approach. For our test we would compute
$Q_n=P(T_N^{(n)}\geq 1/\wt\mid X_{1:n}).$ This we can do using Monte Carlo methods, which for a given $n$  involves multiple sampling of $X'_{n+1:N}$ and for each one to see if the randomly generated $T_N^{(n)}$ is larger or smaller than $1/\wt$.

\begin{center}
\begin{figure}[!htbp]
\begin{center}
\includegraphics[width=12cm,height=5cm]{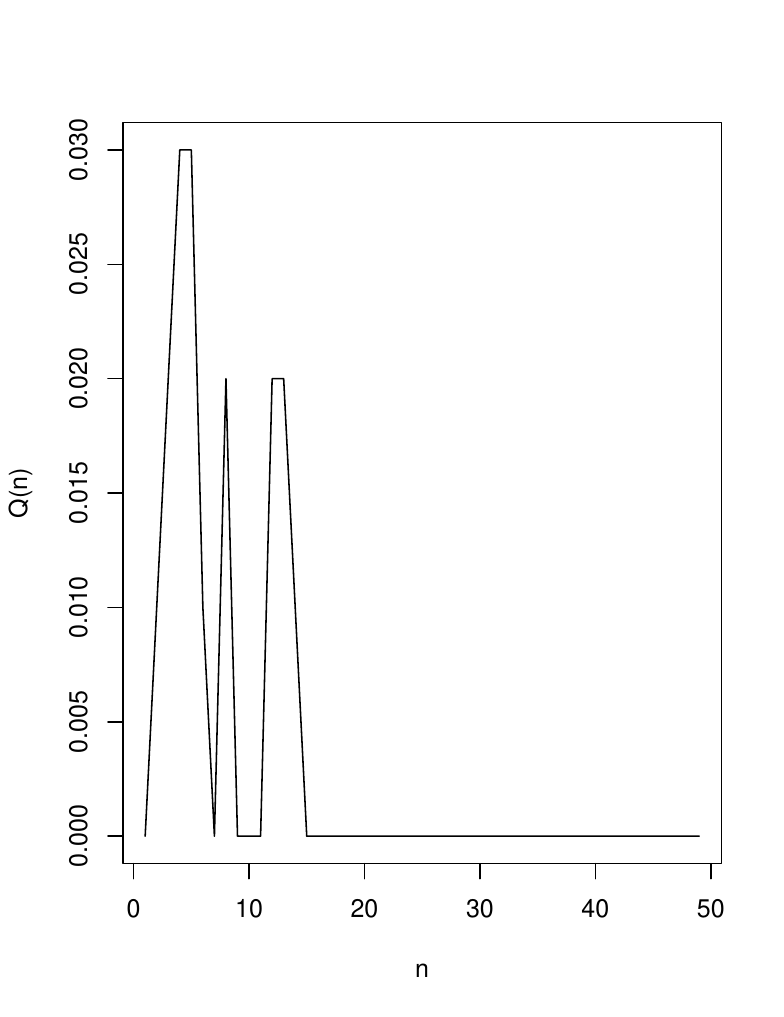}
\caption{$Q_n$ for MLE based sequential test using predictive sampling with true mean $-\half$}
\label{figa13}
\end{center}
\end{figure}
\end{center} 

\begin{center}
\begin{figure}[!htbp]
\begin{center}
\includegraphics[width=12cm,height=5cm]{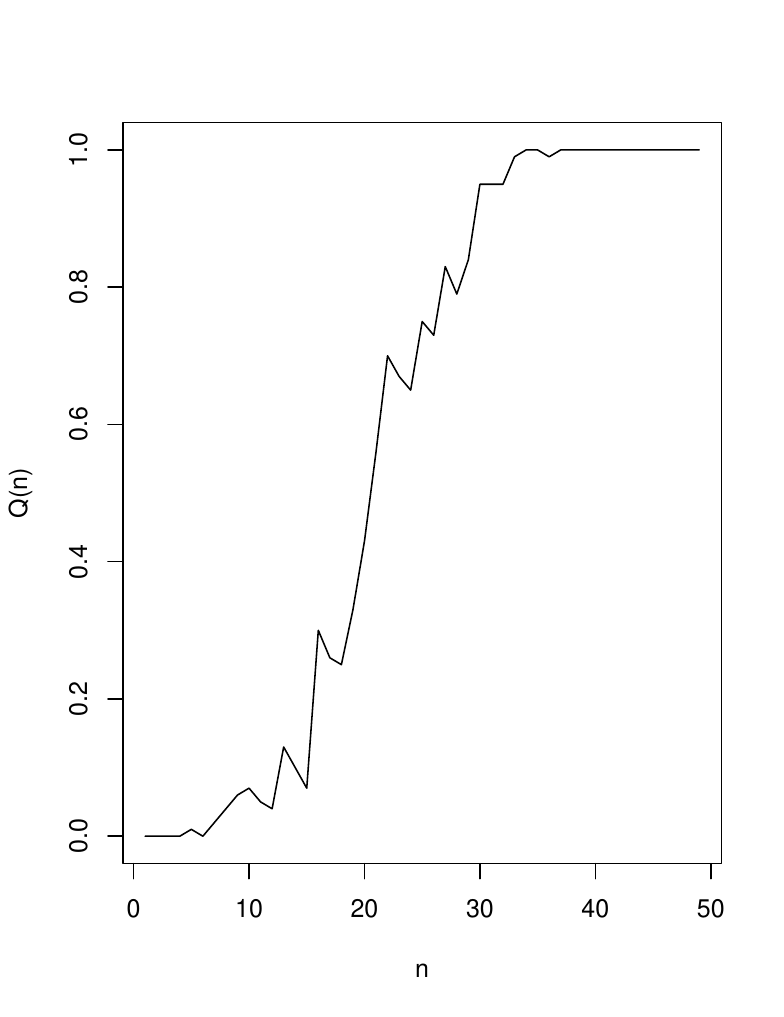}
\caption{$Q_n$ for MLE based sequential test using predictive sampling with true mean $\half$}
\label{figa14}
\end{center}
\end{figure}
\end{center} 

\noindent
We implemented this strategy with $N=50$ and with a true value of the normal mean as $-\half$ and chose $\wt=1/5$. A plot of $Q_n$ versus $n$ is presented in Fig.~\ref{figa13}. For comparison Fig.~\ref{figa14} presents the corresponding $Q_n$ with the true mean as $\half$.

\subsection{Time to event trial}\label{sec:Event}

In this section we consider a different theme based on a trial in which two groups have been allocated different 
treatments (e.g. placebo and active) and the observations are times to event. Each group has $N$ individuals at the start.

\subsubsection{Discrete time}

In this case time to event is recorded at discrete times represented for ease of notation as $t=1,2,\ldots,T$, where $T$ is a fixed horizon.
So each individual provides a time to event or no time, i.e. that individual's event time is beyond $T$. We represent the event times by $(n_{01},n_{02},\ldots,n_{0T})$ for group 0 and $(n_{11},n_{12},\ldots,n_{1T})$ for group 1, where $n_{jt}$ is the number of individuals in group $j\in\{0,1\}$ who experience event time at $t$. 
We then consider the sequences $(p_{0t},p_{1t})$ where, for all $t$ and $j\in\{0,1\}$,
$p_{jt}=n_{jt}/m_{jt}$ and $m_{jt}=N-\sum_{s<t}n_{js}$ for $t>1$ and $p_{1t}=n_{1t}/N$. The null hypothesis $H_0$ is 
$$H_0:\,\,\,n_{0t}\equiv \mbox{Bin}(q_t,m_{0t})\quad\mbox{and}\quad n_{1t}\equiv Bin(q_t,m_{1t})\quad\mbox{independently}$$
for some unknown set of probabilities $(q_t)$.

The test at time $T$ would be based on some measure of distance between the two vectors $p_0=(p_{01},\ldots,p_{0T})$ and $p_1=(p_{11},\ldots,p_{1T})$, such
as reject $H_0$ if
$$S_T=\sum_{t=1}^T |p_{0t}-p_{1T}|\geq c.$$
To find the value of $c$ for which $P(S_T\geq c)\leq \wt$ we can take $q=\half$, since this value provides the largest stochastic value for $S_T$. The value of $c$ can be found from simulation by finding the appropriate quantile of a large sample of $S_T$ with $q=\half$. For example, with $N=50$, $T=10$ and $\wt=0.05$, it is that $c=3.52$.

We take, for $t=1,\ldots,T-1$,
$Q_t=P(S_T\geq c\mid {\cal F}_t)$
where ${\cal F}_t$ is the information up to and including time $t$. Explicitly,
$$Q_t=P\left(\sum_{s=1}^t |p_{0s}-p_{1s}|+\sum_{s=t+1}^T|p'_{0s}-p'_{1s}|\geq c\Big| {\cal F}_t\right),$$
where the $p'_{0\,t+1:T}$ and $p'_{1\,t+1:T}$ are derived independently from binomial distributions with $q=\half$. 
Further, $Q_T=1(S_T\geq c)$.

It can be shown that under the null hypothesis $(Q_t)$ is a sub-martingale and hence from the Doob sub-martingale inequality,
$$P(\max\{Q_1,\ldots,Q_{T}\}\geq \gamma)\leq E(Q_{T})/\gamma\leq \wt/\gamma$$
since $P(S_T\geq c)\leq\wt$.

\begin{center}
\begin{figure}[!htbp]
\begin{center}
\includegraphics[width=12cm,height=5cm]{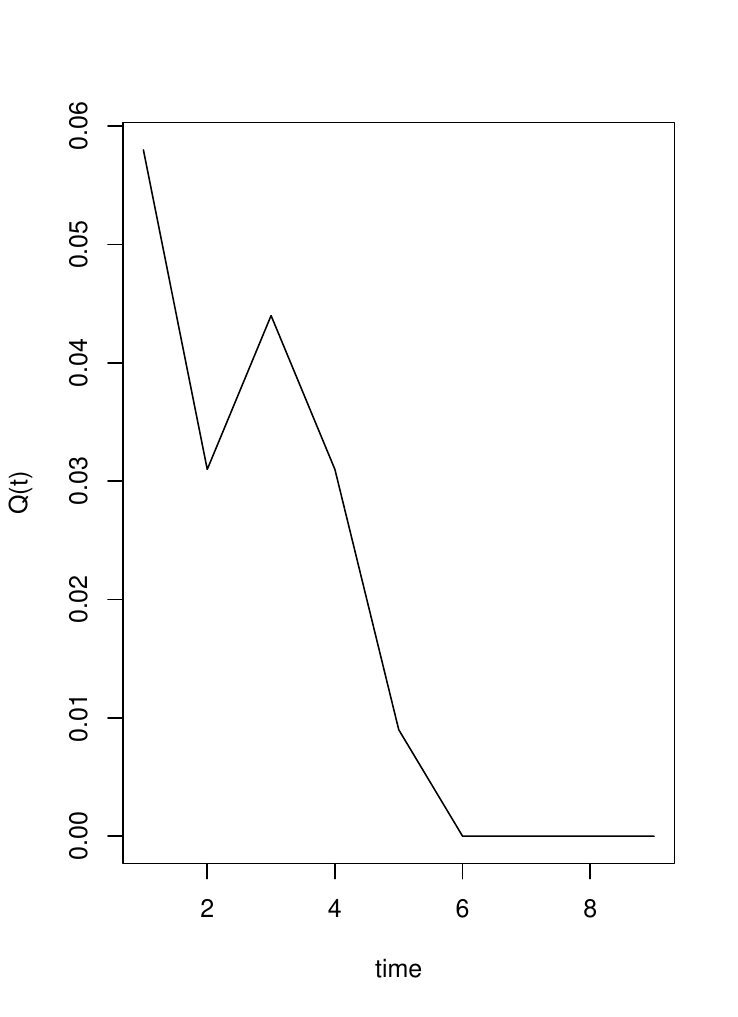}
\caption{Sample path of $Q_t$ under null hypothesis with $q_t=0.2$ for all $t$}
\label{figa16}
\end{center}
\end{figure}
\end{center} 

\begin{center}
\begin{figure}[!htbp]
\begin{center}
\includegraphics[width=12cm,height=5cm]{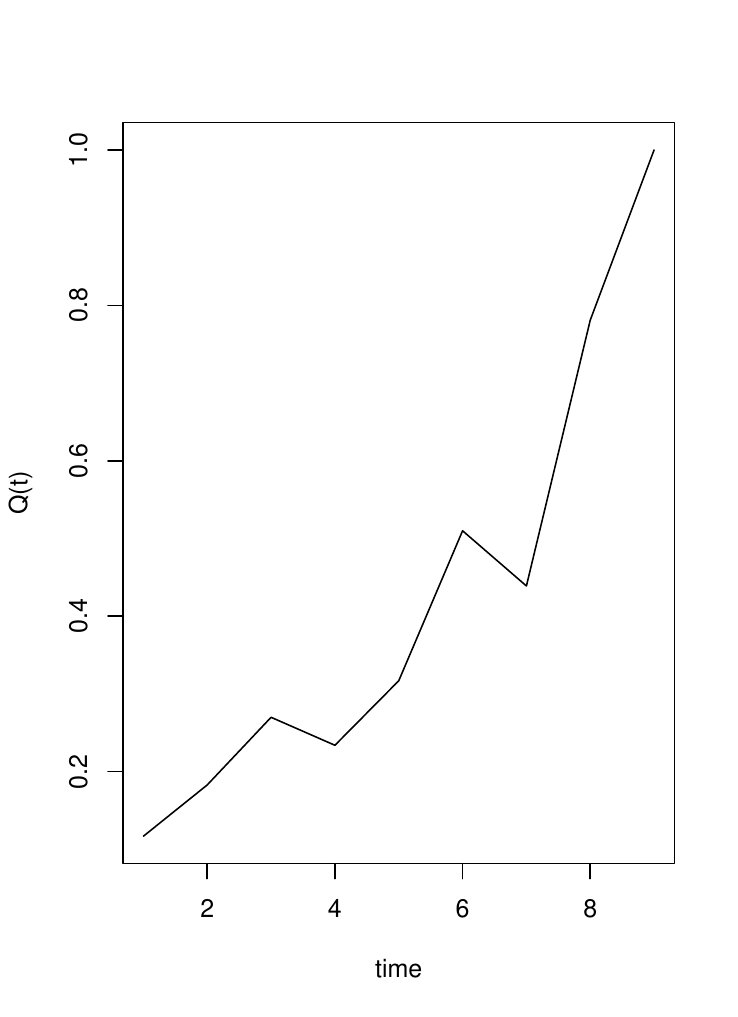}
\caption{Sample path of $Q_t$ with $q_{0t}=0.2$ and $q_{1t}=0.5$ for all $t$}
\label{figa17}
\end{center}
\end{figure}
\end{center} 

As an illustration we took $N=50$ and with the data under the null hypothesis being generated with $q_t=0.2$ for all $t$ and with $c=3.52$. A sample path of $Q_t$ for $t=1,2,\ldots,9$ is provided in Fig.~\ref{figa16}. In Fig.~\ref{figa17} a sample path of $Q_t$ is presented for which the $q_{0t}$, i.e. the binomial probability for group 0 is 0.2 for all $t$ and the binomial probability for group 1 is $q_{1t}=0.5$ for all $t$.

The power for this procedure can be improved when $q$ is likely to be small by including in the hypothesis the upper bound $q_t\leq q_0$ for all $t$. This $q_0$ can then replace the $\half$ for establishing the critical value and the sampling of unobserved data.

\subsubsection{Continuous time}

In this case the observations are continuous on the interval $(0,T]$. This makes the following easier with respect to notation and it is straightforward to extend to the case when observations can lie beyond $T$. Also, without loss of generality, we take $T=1$. The test at time $T$ is based on the statistic
$$S=\sup_{u}|F_N(u)-G_N(u)|$$
where $F_N$ and $G_N$ are the empirical distribution functions from the two groups. The test will be the two sample Kolmogorov-Smirnov (KS) test. An important property of the test is that the distribution of $S$ under the null hypothesis, i.e. $H_0:\,F\equiv G$, is independent of $F$. Hence, for establishing the critical value of $S$, of the form $S>c$, can be determined by taking $F$ to be uniform on $(0,1)$.  

Our first result is that we can write, for any $0<t<1$,
\begin{equation}\label{split}
S=\max\left\{\sup_{u\leq t} |F_N(u)-G_N(u)|,\,\sup_{t<u}|F_N(u)-G_N(u)| \right\}
\end{equation}
and we can write
$F_N(u)=F_N^{(t)}(u)+F_N(t)$ where $F_N^{(t)}$ is the empirical distribution function restricted to the sample larger than $t$. So $F_N^{(t)}(t)=0$.
If we consider the distribution
$$F_t(u)=\left\{\begin{array}{ll}
F(u) & u\leq t \\ 
F^{(t)}(u)+F(t) & u>t
\end{array}\right.
$$
then (\ref{split}) would be the form of the Kolmogorov-Smirnov statistic for testing $F_t\equiv G_t$. 

This then allows us to define $Q_t$ as
$$Q_t=P\left(\max\left\{\sup_{u\leq t} |F_N(u)-G_N(u)|,\,\sup_{t<u}|F'_N(u)-G'_N(u)| \right\}\geq c\Big| {\cal F}_t\right)$$
where ${\cal F}_t$ is the information collected up to time $t$.
Here, $F_N(u)$ is the empirical distribution distribution of the $(X_i\leq t)$  and, for $u>t$, $F'_N(u)=F_N^{(t)}(u)+F_N(t)$ with $F^{(t)}_N(u)$ being the empirical distribution function of a uniform sample of size $N-n_t$ from $(t,1)$, where $n_t=\sum_{i=1}^n 1(X_i\leq t)$. Similarly for $G_N$, $G_N'$ and the $(Y_i)$. 

It can be shown that $(Q_t)_{t\geq 0}$, with $Q_0$ defined with independent uniform $(0,1)$ samples of size $N$ from each group, is a martingale. That is, for any $s<t$,
$E(Q_t\mid {\cal F}_s)=Q_s$. 

\begin{center}
\begin{figure}[!htbp]
\begin{center}
\includegraphics[width=12cm,height=5cm]{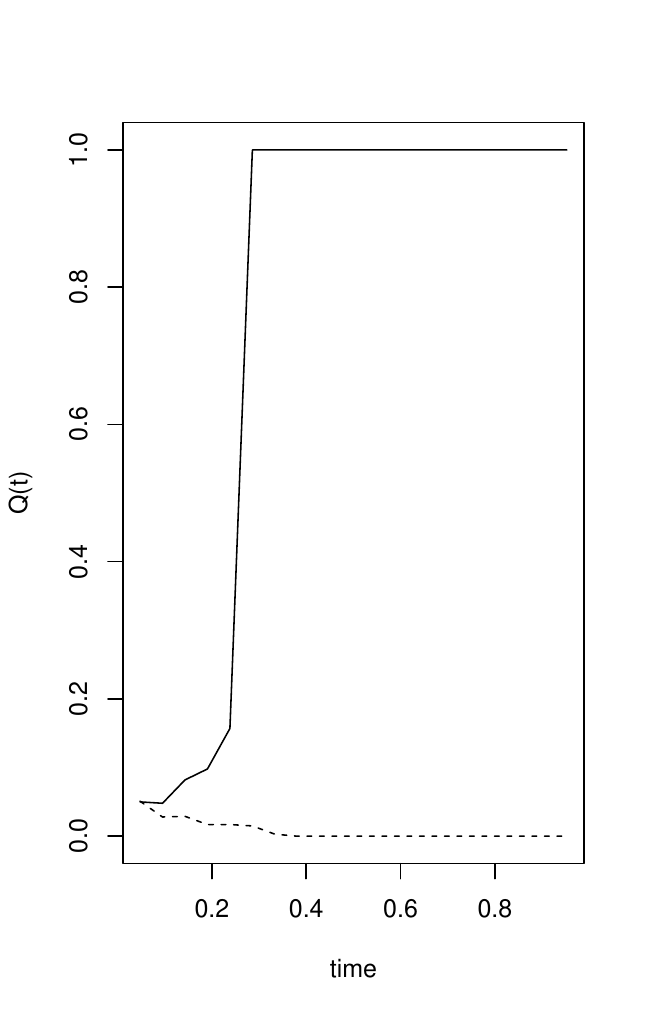}
\caption{$Q_t$ under null hypothesis (dashed line) and $Q_t$ under alternative hypothesis with one group as beta(2,3) and the other group as beta(3,4) }
\label{figa18}
\end{center}
\end{figure}
\end{center} 

As an illustration we use $N=100$ for each group and aim to evaluate $Q_t$ at 20 time points $t_j=j/21$, for $j=1,\ldots,20$. The value of $c$ is 0.19, based on a KS two sample test with $\wt=0.05$. Fig.~\ref{figa18} contains a sample $Q_t$ curve for the null hypothesis that the two distributions from the groups are the same (dashed line) and also contains a sample $Q_t$ curve where the distribution from one group is beta(2,3) while that of the other group is beta(3,4).

The more general setting of observations coming from $(0,\infty)$ can be easily covered, by taking the $X_i>t$ to be, for example, standard exponential random variables, which become effectively censored if they are beyond $T$. 

\subsection{Longitudinal trial}\label{sec:Long} 

In this section we consider a longitudinal phase 3 clinical trial as described in \cite{Hopper23}. Two randomized groups, one of size 163 who received the active treatment {\sl sotatercept}, and the other of size 160 who received placebo, were studied for 24 weeks. At the start, each individual provided a distance walked in a 6 minute period, and repeated the exercise after the 24 weeks with the change in distance walked between the two occasions being recorded. For the sake of the present article we assume and will base our simulated data on the recording of the changes in distance walked in 6 minutes for each individual every 3 weeks. So, in total, each individual provides 9 distances, transformed to 1 baseline distance at the start and 8 changes in distance.    

For a description of the testing procedure, and for ease of notation, we will assume the size of the two groups are the same at $N$. The data are represented by $(X_{i0},Y_{i0})$ and
$(\Delta^X_{it},\Delta^Y_{it})$ for $t=1:T$, and for all $i=1,\ldots, N$. Here $\Delta^X_{it}=X_{it}-X_{i\,t-1}$ and the $(X_{it})$ are distances walked in 6 minutes for the $X$ group, and similarly for the $Y$ group.
The test will be formulated around a hypothesis related to $\Delta^X\equiv_d \Delta^Y,$
that is, all the $\Delta^X_{it}$ and $\Delta^Y_{it}$ are all independent and identically distributed from some common (unknown) distribution. The null hypothesis can then be re-formulated sequentially, as
$$H_0: (\Delta^X_{t},\Delta^Y_{t})\quad\mbox{i.i.d.}\quad F_{t-1},\quad t=2,\ldots,T,$$
where $F_{t-1}$ is the empirical distribution function of all the changes from both groups up to time $t-1$ and $\Delta^X_{t}=\{\Delta^X_{it},\,\,i=1,\ldots,N\}$ and similarly for the $Y$ group.

The null hypothesis is rejected at time $T$ if 
$\sup_u\left|F^X_T(u)-F^Y_T(u)\right|>c$
where $F_T^X$ is the empirical distribution function for the sum of changes over the $X$ group recorded at time $T$. Similarly for $F^Y_T$. We now define the empirical distribution function $F^X_{t,T}$ which is constructed from the observed changes for group $X$ up to and including time $t$, and the sampled values $\Delta^X_{t+1:T}$ which come from the sampling plan in $H_0$. Similarly for $F^Y_{t,T}$ and $\Delta^Y_{t+1:T}$. Then
$$Q_t=P\left( \sup_u\left|F^X_{t,T}(u)-F^Y_{t,T}(u)\right|>c\mid {\cal F}_t\right),$$
where ${\cal F}_t$ is all the observed information up to and including time $t$. It is easy to see that under the null hypothesis $H_0$ it is that $Q_t$ is a martingale.

\begin{center}
\begin{figure}[!htbp]
\begin{center}
\includegraphics[width=12cm,height=5cm]{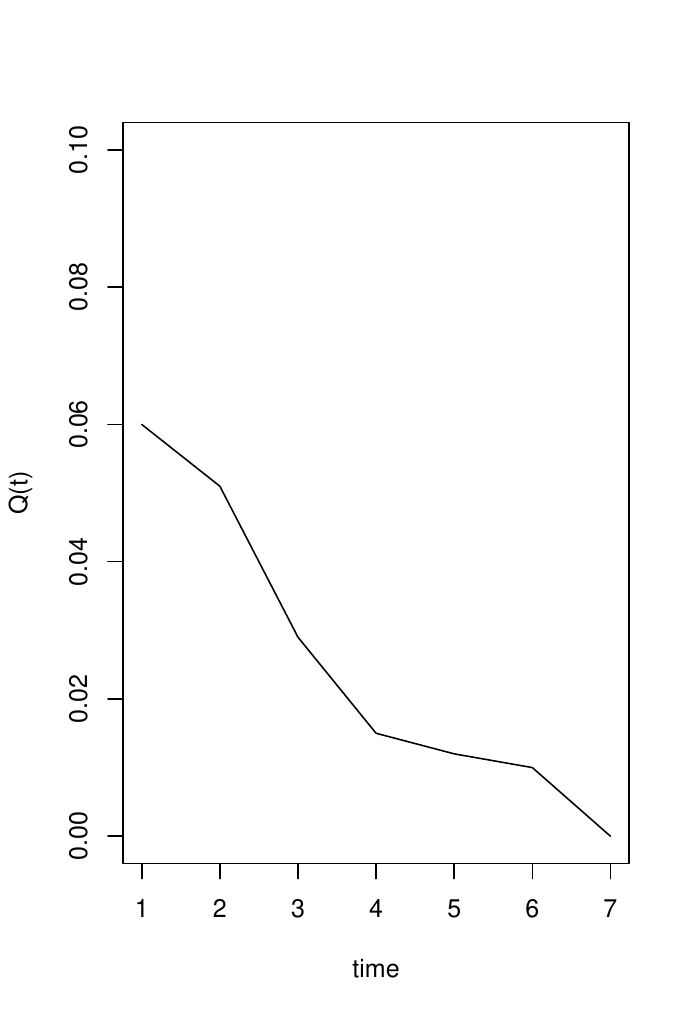}
\caption{$Q_t$ under null hypothesis}
\label{figa19}
\end{center}
\end{figure}
\end{center} 

\begin{center}
\begin{figure}[!htbp]
\begin{center}
\includegraphics[width=12cm,height=5cm]{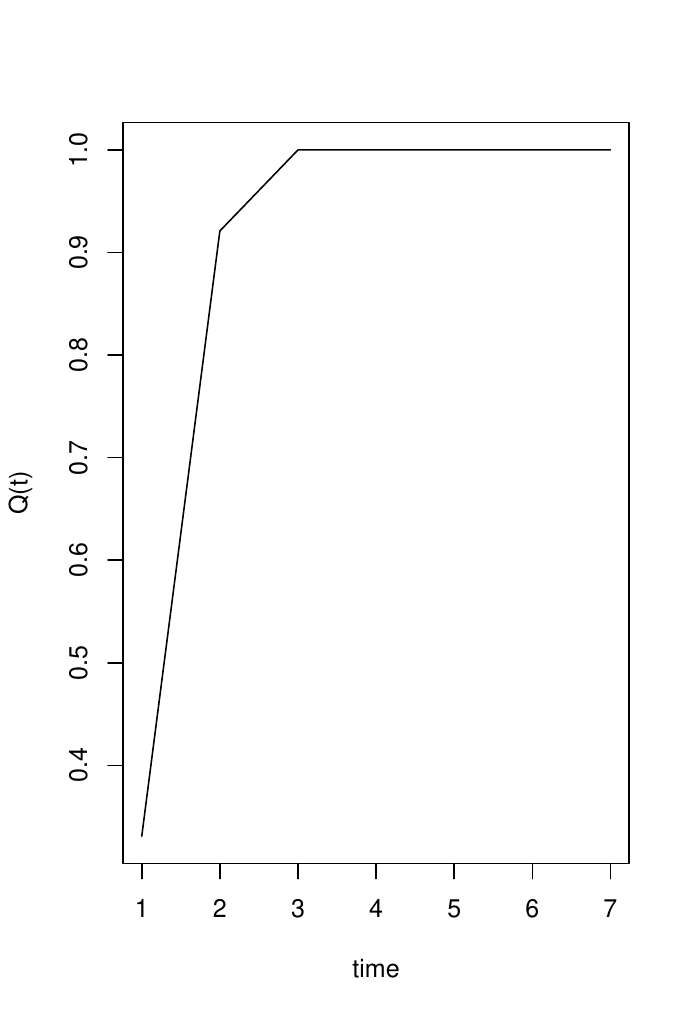}
\caption{$Q_t$ under alternative hypothesis}
\label{figa20}
\end{center}
\end{figure}
\end{center} 

\noindent
As an illustration we took $N=150$ individuals in each group. 
We took $T=8$ and so we compute $(Q_t)$ for $t=1,\ldots,7$. First we illustrate with a true null hypothesis and sampled all the changes to be independent normal with mean 0 and variance $0.1^2$.
A sample path of $(Q_t)$ is presented in Fig.~\ref{figa19}. On the other hand, a sample path of $Q_t$ under the null hypothesis is presented in Fig.~\ref{figa20}. The distributions of the changes for one group are independent normal with zero mean and variance $0.1^2$ while the changes for the other group are independent normal with mean 0.05 and variance $0.1^2$.


\section{Real-data illustration}
\label{sec:S9}
\noindent
{\sc Nonparametric two-sample test for the International Stroke Trial}.
Real data does not often fit into a parametric model and for testing $H_0:F_X=F_Y$ for two samples, $X_{1:N}$ and $Y_{1:M}$, a resource is the distribution free two sample Kolmogorov-Smirnov (KS) test.
The statistic for sample sizes $N$  and $M$ is given by
$$T_{N,M}=\sup_{0<u<T}|F_N(u)-G_M(u)|,$$
where $F_N$ is the empirical distribution of the $X$ sample and $G_M$ is the empirical distribution for the $Y$ sample. All samples are assumed to lie in $(0,T)$, though the following can easily be adapted if $T=\infty$.

Predictive sampling from the true common distribution under the null hypothesis is not possible since it is unknown as the null hypothesis only says the distributions are the same. We assume there are $N$ individuals from one group and $M$ individuals from the other group and each individual provides an event time (in the real data analysis to come it is a death time) within an interval $(0,T)$ for some fixed $T$. Interim analyses will be conducted at times $0<t_1<\cdots<t_L<T$ which can be as fine as required.

We first define $T_N^{(t)}$ for any $0<t<T$. To this end we define
$$F_N^{(t)}(u)=\frac{\sum_{i=1}^N 1(X_i\leq t,X_i\leq u)+\sum_{i=1}^N 1(t<X_i'\leq u)}{N},$$
where the $(X'_{n_t+1:N})$ are independent and identically distributed from the uniform distribution on $(t,T)$, and $n_t$ is the number of the $(X_{1:N})$ which are less than or equal to $t$. The point here is that the $X_i>t$ are unobserved but known to lie in $(t,T)$ and hence from a predictive point of view we sample them as uniform. This is without loss of generality due to the invariant property of the KS two sample test.

Likewise, we define $G_M^{(t)}(u)$ for the $Y$ group. The interesting aspect regarding the two sample KS statistic is that it is invariant to the common distribution, even if it is based on a mixture as in $F_N^{(t)}$ and $G_N^{(t)}$. That is, under the null hypothesis, that $F_X\equiv F_Y$,
$$T_{N,M}^{(t)}=\sup_{0<u<T} \bigg|F_N^{(t)}(u)-G_M^{(t)}(u)\bigg|$$
is distribution invariant under the given construction. Hence, if we define
$$Q_{j}=P\left(T_{N,M}^{(t_j)}\geq c\right),\quad j=1,\ldots,L,$$
and $Q_0=P(T_{N,M}^{(0)}\geq c)$,  where in this case all samples are uniform on $(0,T)$, and further define $Q_{M+1}=1(T_N^{(T)}\geq c)$, then $(Q_j)$ will be a martingale. 

The critical value $c$ will be based on the two sample KS test, which for a Type I error of $\alpha=0.05$ 
we take as $c=1.36\sqrt{(N+M)/NM}$.
For $N=M=500$ this is $c=0.086$.
Taking $N=M=500$, we illustrate $(Q_j)$ by taking all samples independently from a beta$(2,4)$ distribution and plot the sequence of probabilities $(Q_t)$. See Fig.~\ref{fig28} for a plot $t$ and $Q_t$ taking $t_j=j/11$ for $j=1,\ldots,10$. 

With random right censoring, as we have with the real data analysis, we exchange the two sample KS test with the two sample nonparametric log rank test. 
The statistic in this case is given by
$$T_{N,M}=\frac{\sum_{j=1}^J (n_X(j)-e_X(j))}{V},$$
where $n_X(j)$ is the number of events in category $j$ (deaths in  month $j$) from the $X$ group and
$e_X(j)=n_X(j)\,m_X(j)/m(j)$, where $m_X(j)$ is the number of individuals at risk of dying in month $j$ and $m(j)=m_X(j)+m_Y(j)$.
Also, $V=\sum_{j=1}^J V_j$ and
$$V_j=e_X(j)\left(1-\frac{n_X(j)+n_Y(j)}{m(j)}\right)\,\frac{m_Y(j)}{m(j)-1}.$$
Under the null hypothesis the $T_{N,M}$ is approximately a standard normal random variable.

\begin{center}
\begin{figure}[!htbp]
\begin{center}
\includegraphics[width=12cm,height=5cm]{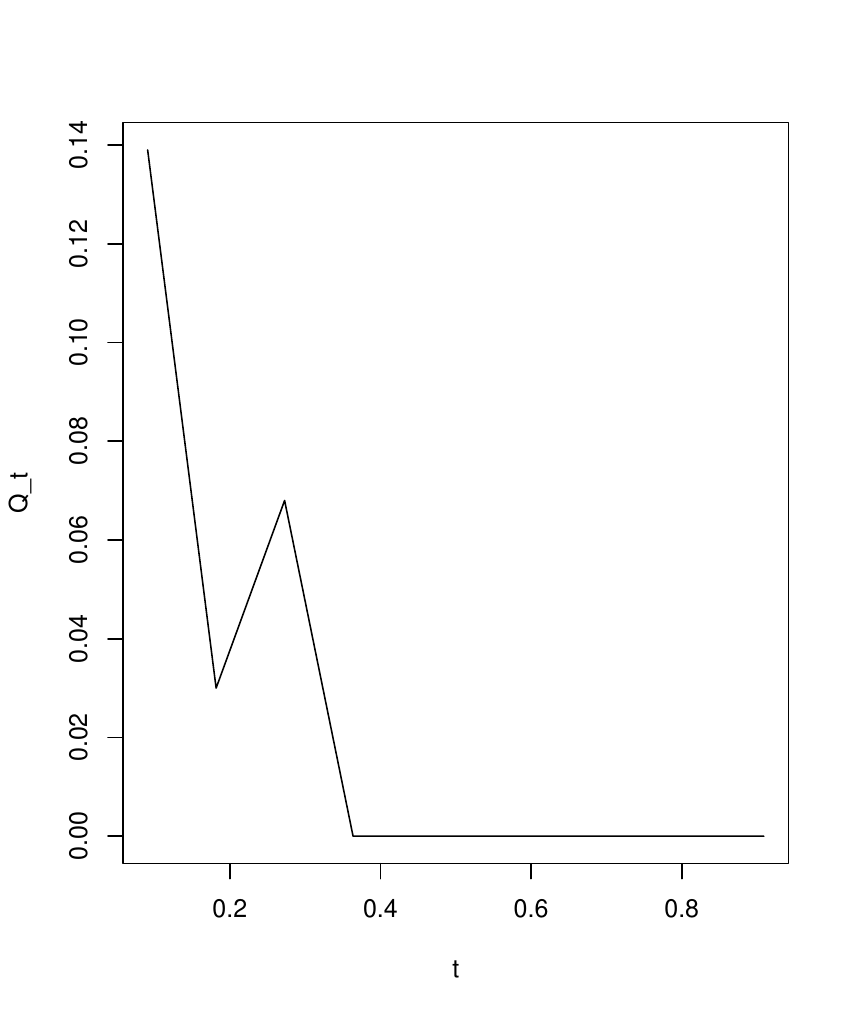}
\caption{Plot of $Q_t$ for a true null hypothesis }
\label{fig28}
\end{center}
\end{figure}
\end{center}

Here we provide details for the International Stroke Trial (IST) analysis reported in the main text, see \cite{Sander11}. One of the aims of the trial is whether the administration of aspirin influences recovery following a stroke.
There were a total of 19,345 patients in the study from 467 hospitals and 36 countries. The total numbers recorded for the assignment of aspirin for the first 14 days of the trial was 9071 with 10326 not being assigned aspirin over the 14 days. There is no record of assignment or otherwise for the remaining patients. 
From the former group there were 1897 death times while for the latter group there were 2462 death times, with all other times given being censoring times. The number of categories is $J=15$ which represent the number of months and so all events (death and censoring) are recorded as occurring in some month $j\in\{1,\ldots,15\}$. The final category 15 contains all observations greater than 14 months.


\begin{center}
\begin{figure}[!htbp]
\begin{center}
\includegraphics[width=12cm,height=5cm]{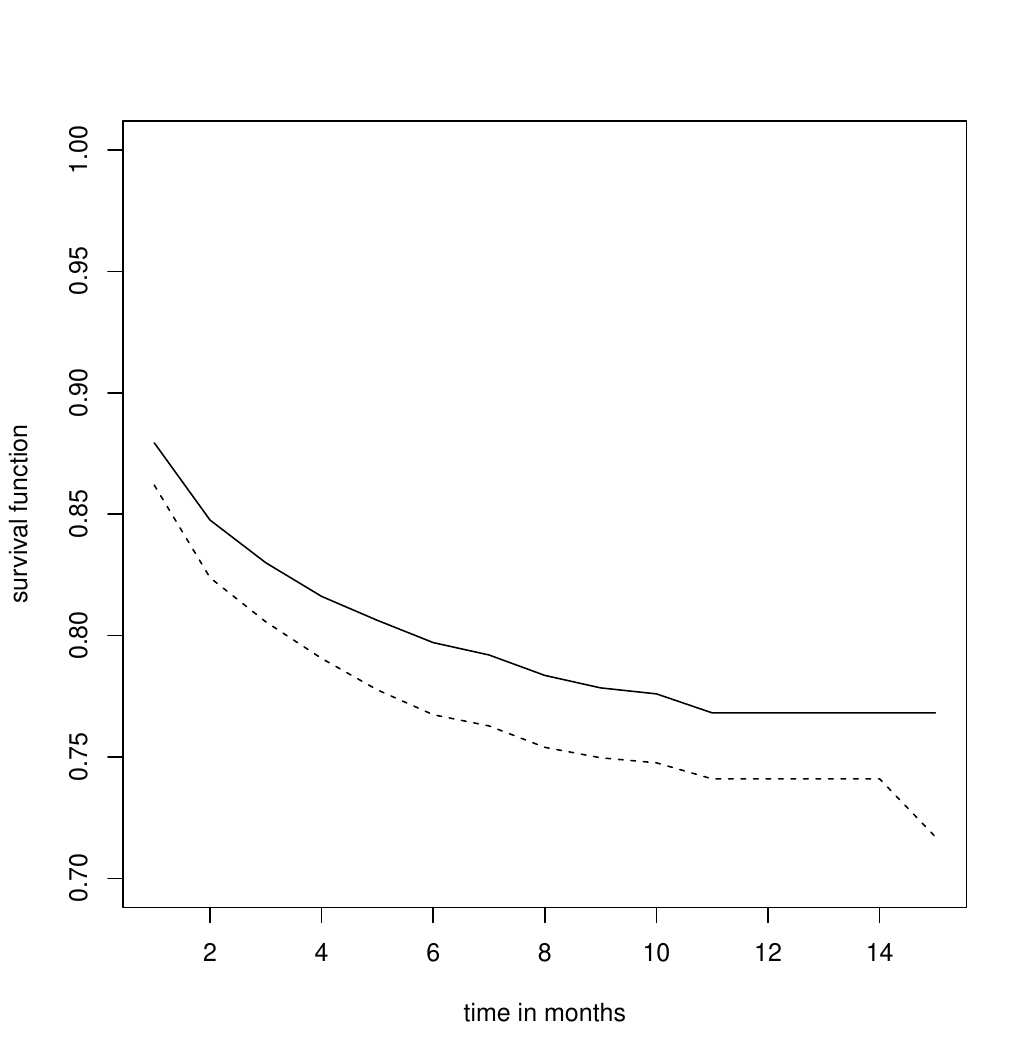}
\caption{ Kaplan-Meier estimators from the two groups for the first 15 months of the trial: assigned aspirin (bold line), not assigned aspirin (dashed line) }
\label{figsm8}
\end{center}
\end{figure}
\end{center}

The null hypothesis is that the conditional distributions of the death times for the two groups are equal.
The two Kaplan-Meier survival functions are shown in Figure~\ref{figsm8}, with the assigned aspirin group shown in bold and the not assigned aspirin group shown as the dashed line.
For the two sample log rank test, the relevant statistic is 8.69 while the critical value for a Type I error of 0.05 is 3.84. So the equality of the distributions is rejected.

\begin{center}
\begin{figure}[!htbp]
\begin{center}
\includegraphics[width=12cm,height=5cm]{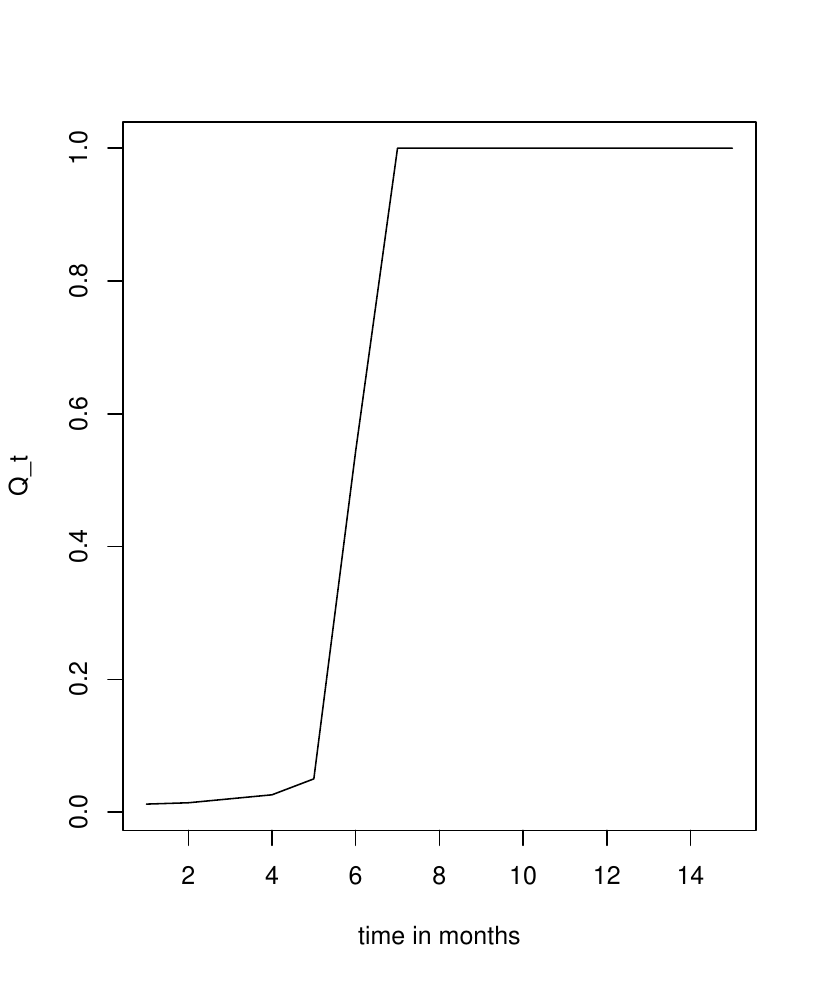}
\caption{ Plot of $Q_t$ over 14 months of the trial. The function exceeds 0.95 after 7 months.}
\label{figsm9}
\end{center}
\end{figure}
\end{center} 

For the plot of $Q_t$ with $t=1,\ldots,15$, representing time in months, is shown. The statistic is based on completing the full data set as described for the two sample log rank test. For sampling the outcomes of the future observations at each time point $t$, all individuals will be sampled to provide a death time whether it is known there future observation is a death or censoring time. With the sample statistic at time 15 months the probability of rejecting the null hypothesis is computed, using the critical value of 3.93  to obtain an overall rejection probability of 0.05 if the null hypothesis is true.
The hypothesis can be rejected safely after 7 months; see Fig.~\ref{figsm9}.




\section{Stopping early for futility}\label{sec:Futility}


The main manuscript focuses on early stopping for efficacy, while controlling for the
Type I error. 
In some applications it is also desirable to stop early for futility, that is
to terminate the experiment when continued sampling is unlikely to detect an effect
of practical or scientific importance. In other words, strong evidence points to the null hypothesis being true. More generally, we can extend the approach to safely predict the truth status of a test for any hypothesis at $N$, by simulating the missing data assuming the hypothesis to be true.

\paragraph{Design alternative and updated power.}
Let $\theta$ index the effect size of interest, with $\theta=0$ corresponding to the
null hypothesis.
We define $\theta^*>0$ to be the \emph{design alternative}, i.e.\ the effect size used
in the original power calculation that motivated the choice of maximal sample size $N$.

At interim stage $n$, define the predictive rejection probability under the design
alternative,
\[
Q_n^* =
\mathbb{P}_{\theta^*}\!\left(T_N^{(n)} \notin \mathcal{C} \mid \mathcal{F}_n\right) \, , 
\]
where $T_N^{(n)}=T(X_{1:n},X^*_{n+1:N})$ with the $X^*_{n+1:N}$ coming from the alternative hypothesis $\mathbb{P}_{\theta^*}$ with $\theta=\theta^*$.
If the alternative hypothesis is true then $(Q_n^*)$ is a martingale.
Therefore, we can safely stop and {\em{not reject}} the null hypothesis if ever $Q_n^*$ is too large.

This quantity can be interpreted as the updated Type II error at sample size $n$, namely, 
the conditional probability, given the data observed so far, that the fixed-$N$
test would not reject the null hypothesis if the experiment were continued to its planned sample size and
the true effect were $\theta^*$. Note that $Q_0^*$ specifies the power of the fixed-N test, i.e. $Q_0^*=P_{\theta^*}(T_N\notin \mathcal{C})$, and Power $= 100\% \times (1- Q_0^*)$.

\paragraph{Futility stopping rule.}
A natural decision rule for futility is to stop sampling when the updated Type II error passes a threshold $\gamma^{\mathrm{F}}$,
\[
\tau_F = \inf\{n \le N : Q_n^* \ge \gamma^{\mathrm{F}}\},
\]
with a default choice such as $\gamma^{\mathrm{F}}=0.99$. 
Hence, we stop the experiment and don't reject the null hypothesis if ever $Q_n^* \geq\gamma^F$. As $Q_n^*$ is a martingale, so
$$P\bigg(\max_n Q_n^*\geq\gamma\bigg)\leq Q_0^*/\gamma^F$$
and with $Q_0^*$ typically given by the design of the fixed sample test, we would then choose $\gamma$ so that $Q_0^*/\gamma^F$ is small, say 0.01. 

Intuitively, the decision rule stops for futility whenever the current experiment falls below a predicted power threshold to detect the effect size that the study was designed to detect; the
probability of the eventual null hypothesis being rejected is small and future measurement of units is wasteful, as continuing with the study is unlikely
to change the scientific conclusion.

\paragraph{Screening via the null-calibrated process $Q_n$.} 

In many common settings (e.g.\ one-sided tests in monotone likelihood ratio families), under the alternative hypothesis, 
the rejection event $\{T_N \in \mathcal{C}\}$ is increasing in the data and the model
is stochastically ordered in $\theta$.
Under these conditions, when the null hypothesis is false, 
$
1- Q_n^* \ge Q_n,
$
where $Q_n=\mathbb{P}_0(T_N^{(n)} \in \mathcal{C} \mid \mathcal{F}_n)$ is the
null-calibrated predictive rejection probability.
Consequently, if $1 - Q_n < \gamma^{\mathrm{F}}$, then $Q_n^* <  \gamma^{\mathrm{F}}$
and the decision to stop for futility won't be made. 
This allows $Q_n$ to be used as a conservative screening statistic, so that
$Q_n^*$ need only be computed when $1- Q_n$ falls below $\gamma^{\mathrm{F}}$. 

\paragraph{Type I error protection.}
Stopping early for futility cannot inflate the Type I error rate.
By construction, rejection of the null hypothesis only occurs via the efficacy
stopping rule based on the null-calibrated process $Q_n$.
Any futility rule can only terminate sampling without rejection and therefore can
only reduce (or leave unchanged) the probability of rejecting a true null.
Errors associated with futility stopping correspond to stopping and non-rejection of the null when the true effect is at least $\theta^*$ and are interpreted as losses in design-stage power rather
than violations of Type I error.

\begin{center}
\begin{figure}[!htbp]
\begin{center}
\includegraphics[width=14cm,height=5cm]{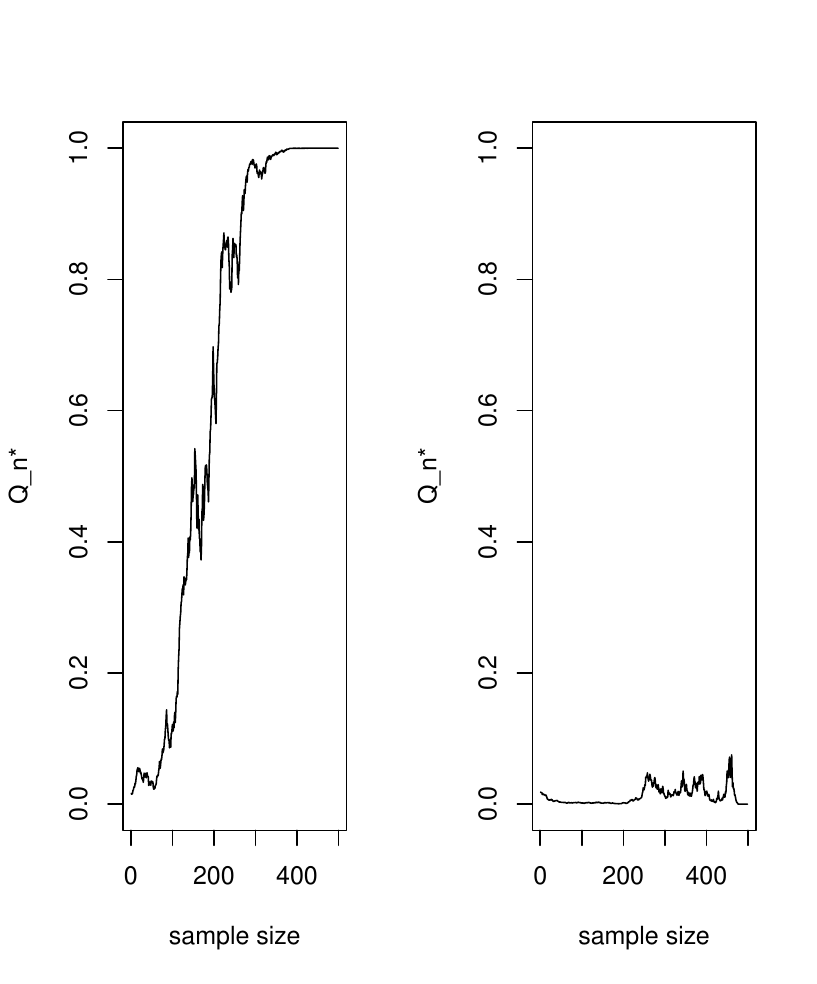}
\caption{Plots of $Q_n^*$. Stopping for futility. The left panel is when the null hypothesis $H_0$ is true and the right panel is when the alternative hypothesis is true with $\theta^*=0.13$ 
}
\label{figsm4}
\end{center}
\end{figure}
\end{center}

\paragraph{Relation to Bayesian adaptive futility stopping.}
Bayesian adaptive trials often include a decision to stop early for futility by monitoring the posterior
probability that a clinically meaningful effect will be detected at the end of the
study, and terminating the trial when this probability becomes sufficiently small.
The predictive futility criterion described here addresses an analogous
forward-looking question, but does so without reliance on prior distributions or
posterior inference.
Instead, futility is assessed relative to the original fixed-$N$ test and the
design alternative $\theta^*$ used in the study’s power calculation, ensuring that
early stopping decisions remain anchored to the frequentist calibration of the
original fixed-sample analysis.

\paragraph{A canonical example: one-sided testing of a normal mean}

For the normal mean illustration, 
$$Q_n^*=\Phi\left(\frac{\sqrt{N}z_{\wt}-T_n-(N-n)\theta^*}{\sqrt{N-n}}\right).$$
The experiment is stopped and the null hypothesis is not rejected if ever $Q_n^*\geq\gamma^F$.

\section{Predictive repeated confidence intervals}\label{sec:Confidence}



A confidence interval of size $1-\alpha$ and based on a sample of size $N$ is of the form $I_N=(L_N,U_N)$ where
$P(I_N\ni \theta)\geq 1-\alpha,$
and $\theta$ is the true parameter value.
A repeated confidence interval (RCI) looks for a sequence of confidence intervals which apply over all sample sizes $n\leq N$, that is 
$$P\big(I_n\ni\theta\,\,\,\mbox{for}\,\,\,n=1,\ldots,N\big)\geq 1-\alpha.$$ 
One of the pioneering papers on sequential confidence intervals is \cite{Robbins70} and \cite{Robbins71}. 
The math is provided by a theorem due to \cite{Ville39} concerning martingales. If $(L_n)$ is a non-negative martingale and $E(L_n)=1$ for all $n$ then
$$P(L_n\geq 1/\alpha\,\mbox{for some} \,\,n\geq 1)\leq\alpha$$
for any $0<\alpha<1$. One such martingale is
$$L_n=\frac{p'(x_1,\ldots,x_n)}{p_0(x_1,\ldots,x_n)},$$
where $p_0$ is a hypothesized density (for a test) or is the true density (for confidence intervals) and $p'$ an alternative density function. The test $H_0:p=p_0$ can be rejected if ever $L_n\geq 1/\alpha$. This has become known as anytime valid testing and $(L_n)$ is known as an $e$-process; see \cite{Shafer11} and \cite{Vovk21} for example.
Exploiting the well known connections between testing and confidence intervals, it is possible to construct a sequence of confidence intervals for $\theta$, on which $p_0$ depends, by equating the events $\{L_n\leq 1/\alpha\}$ with
$\{I_n\ni \theta\}$, the $(I_n)$ being a sequence of random intervals. For example, if $p_0(x_{1:n};\theta)=\prod_{i=1}^n \phi(x_i-\theta)$, the $(I_n)$ can be found using some elementary algebra.

The influential paper by \cite{Jennison89} focused on repeated confidence intervals and their connection with sequential tests and group sequential methods. The ideas presented in the paper remain the benchmark for today, see, for example, \cite{Lewis2023}. 

The key idea to group sequential methods and sequential tests is a finite number of statistics to be observed sequentially, say $(T_n)_{n=1:N}$ and the joint distribution of $(T_1,\ldots,T_n)$ is assumed known for all $n$ under the null hypothesis or true parameter model and for convenience is set or restricted to be multivariate normal with standard normal marginals for each $T_n$. The tests and repeated confidence intervals are constructed from a set of choice $(\pi_n)$ for which $\sum_{i=1}^n\pi_i=\alpha$ and the two sided critical values are set by deriving sequentially the $(c_n)$ for which
$$P(|T_n|\geq c_n,|T_1|<c_1,\ldots,|T_{n-1}|<c_{n-1})=\pi_n.$$
Tables of such $(c_n)$ appear in Table 1 in \cite{Jennison89}. We will look more at this later.

As argued in \cite{Jennison89} flexibility is required when stopping a medical trial prior to a pre-determined maximal sample  size. Often the decision is complex and a number of factors will be involved. For this reason, flexibility is essential and repeated or sequential confidence intervals for the parameter of interest allow for such.

Under certain circumstances we can compute a sequence of confidence intervals $(I_n)_{n=1:N}$ such that $P(I_n\ni\theta^*\,\,n=1,\ldots,N)\geq 1-\alpha$ or equivalently 
$P(I_n \,\,\mbox{does not contain}\,\, \theta^*\,\,\mbox{for some}\,\,n=1,\ldots,N)\leq\alpha$. Here $\theta^*$ is the true parameter value and here we  illustrate via an example.

\subsection{RCI for normal mean with known variance}

Consider a normal mean $\theta$ with unit variance and consider the hypothesis $H_0:\theta=\theta^*$ versus $H_1:\theta>\theta^*$. 
The fixed sample test at $N$ is given by reject $H_0$ if
$T_N\geq c$, where $T_N=\sum_{i=1}^N X_i$, and for a Type I error of $\wt$ it is that $c=\sqrt{N}\Phi^{-1}(1-\wt)+N\theta^*$.
We now define $Q_n=P(T_N^{(n)}\geq c\mid X_{1:n})$
where $T_N^{(n)}=T_n+\sum_{i=n+1}^N X'_i$ and the $X'_{n+1:N}$ are i.i.d. from the normal distribution with unit variance and mean $\theta^*$. Hence, 
$$Q_n=1-\Phi\left(\frac{\sqrt{N}\Phi^{-1}(1-\wt)+n\theta^*-T_n}{\sqrt{N-n}}\right).$$
As usual, $(Q_n)$ is a martingale and $P(\max_{n=1:N}Q_n\geq\gamma)\leq\wt/\gamma$ for all $0<\gamma<1$. Recall that we do not rely on the $(Q_n)$ being a martingale and can find $\gamma$ via Monte Carlo methods if we use an example for which $(Q_n)$ not a martingale.
The null hypothesis is rejected if ever 
$Q_n\geq\gamma$ which is equivalent to
$$T_n\geq \sqrt{N}\Phi^{-1}(1-\alpha\gamma)-\sqrt{N-n}\Phi^{-1}(1-\gamma)+n\theta^*.$$
Re-arranging, this is equivalent to
$$I_n=T_n/n-\sqrt{N}\Phi^{-1}(1-\alpha\gamma)/n+\sqrt{N-n}\Phi^{-1}(1-\gamma)/n\geq\theta^*.$$
From the properties of the test,
$P\left(I_n\geq\theta^*\,\,\mbox{for some}\,\,n=1,\ldots,N\right)\leq\alpha$
and so $(I_n,\infty)$ is a $(1-\alpha)$\% one-sided sequential confidence interval for $\theta^*$.
Note that $I_N=T_N/N-\Phi^{-1}(1-\alpha\gamma)/\sqrt{N}$ and so to match the fixed sample confidence interval at $N$ we will take $\gamma$ close to 1.

\begin{center}
\begin{figure}[!htbp]
\begin{center}
\includegraphics[width=12cm,height=5cm]{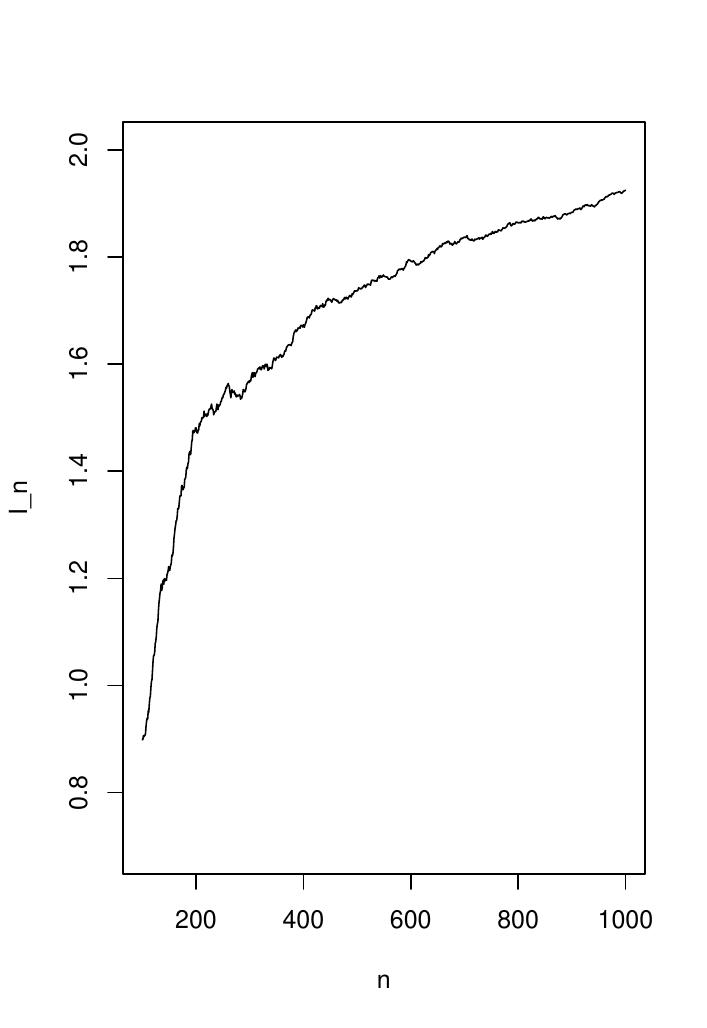}
\caption{A sample path of the lower bound for the one-sided  confidence interval with true mean value of $\theta^*=2$ and using $N=1000$, $\alpha=0.05$, and $\gamma=0.98$}
\label{figc1}
\end{center}
\end{figure}
\end{center} 

\begin{center}
\begin{figure}[!htbp]
\begin{center}
\includegraphics[width=12cm,height=5cm]{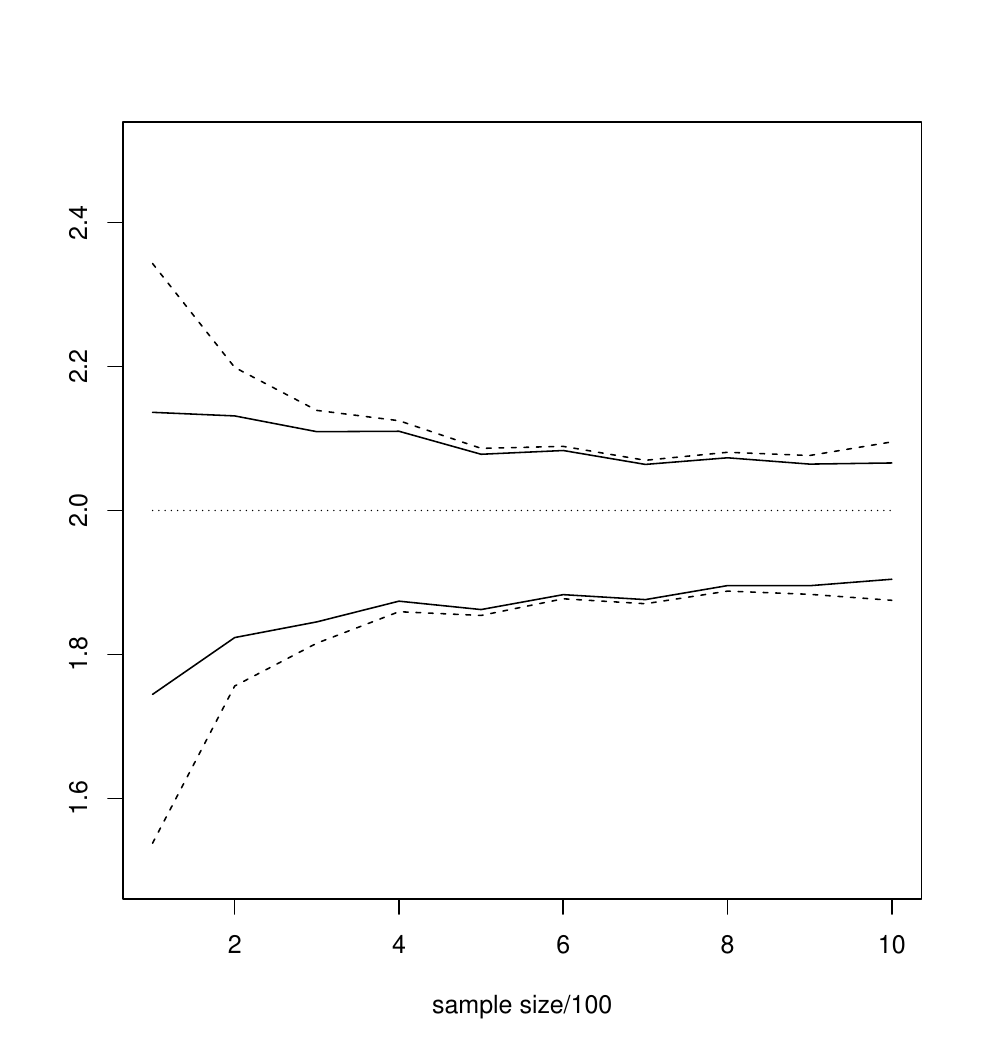}
\caption{A sample path of the bounds for the two-sided  confidence interval (dashed line) with true mean value of $\theta^*=2$, using $N=1000$, $\alpha=0.05$, and $\gamma=0.01$. The bold lines are the corresponding confidence intervals from equation (2.9) in \cite{Jennison89} using the Pocock values.
The intervals are evaluated in blocks of 100 samples.}
\label{figc2}
\end{center}
\end{figure}
\end{center} 

By means of an illustration, we take $N=1000$ and the true parameter value is $\theta^*=2$. A sample path of a confidence interval with lower bound $I_n$ is presented in Fig.~\ref{figc1}, using $\alpha=0.05$ and $\gamma=0.95$. 
A corresponding sequence of two-sided  confidence intervals with lower bound
$$I_{n1}=\bar{X}_n+(\sqrt{N-n}/n)\Phi^{-1}(1-\gamma)-(\sqrt{N}/n)\Phi^{-1}(1-\alpha\gamma/2)$$
and upper bound
$$I_{n2}=\bar{X}_n+(\sqrt{N-n}/n)\Phi^{-1}(\gamma)+(\sqrt{N}/n)\Phi^{-1}(1-\alpha\gamma/2),$$
is presented in Fig.~\ref{figc2}. Here we took $\gamma=0.01$ and plotted the intervals in blocks of 100 samples so as to compare with the bounds in equation (2.9) of \cite{Jennison89} using the Pocock values. 
  

The general structure is based on the hypothesis $H_0:\theta=\theta^*$. From the usual definition of $Q_n$, if $Q_n<\gamma$ implies $\theta^*\in I_n$ for some interval $I_n$ then $(I_n)$ is a $(1-\alpha)$\% sequence of confidence intervals in the sense that
$P(I_n\ni\theta^*\,\,\mbox{for all}\,\,n=1,\ldots,N)\geq 1-\alpha$.
To add some detail, suppose the hypothesis is to reject $H_0$ if
$T_N\in C_N(z_\alpha,\theta^*)$ and 
$$T_N^{(n)}=T_n+C_{N-n}(Z,\theta^*)$$
where $Z$ is a random variable independent of $T_n$ and has distribution $G$. Note that $z_\alpha=G^{-1}(1-\alpha)$.
The null hypothesis is rejected with Type I error $\alpha$ if ever
$$T_n\geq C_N(G^{-1}(1-\alpha\gamma),\theta^*)-C_{N-n}(G^{-1}(1-\gamma),\theta^*)$$
and so
$$I_n=\left\{\theta:\,\,C_N(G^{-1}(1-\alpha\gamma),\theta)-C_{N-n}(G^{-1}(1-\gamma),\theta)\geq T_n\right\}$$
is a sequence of $(1-\alpha)$\% confidence intervals.


\subsection{RCI for normal mean with unknown variance}

In this case the relevant test statistic is 
$T_N(\theta)=\sqrt{N}(\bar{X}_N-\theta)/S_N$
and we take the predicted test statistic as
$T_N^{(n)}(\theta)=\sqrt{N}(\bar{X}_N-\theta)/S_n$
where $X'_{n+1:N}$ are i.i.d. from the normal model with mean $\theta$ and variance $S_n^2$.
Then
$$Q_n=1-\Phi\left((c\sqrt{N}-\sqrt{m}T_n(\theta)/\sqrt{N-n}\right)$$
and we can find the relevant $\gamma$ for which $P(\max_{n_0\leq n\leq N} Q_{n}\leq\gamma)=1-\alpha$ for some starting value $n_0$; i.e. no interim test will be performed before sample $n_0$ is seen. 
This $\gamma$ will hold for all $\theta$. The null hypothesis is rejected  if ever
$Q_n\leq\gamma$ which arises if ever
$$T_n(\theta)\geq c\sqrt{N/n}-\sqrt{N/n-1}\Phi^{-1}(1-\gamma).$$
Hence, a one sided RCI for $\theta$ is presented by $I_n$ where
$$I_n=\{\theta:\,T_n(\theta)\leq c\sqrt{N/n}-\sqrt{N/n-1}\Phi^{-1}(1-\gamma)\}.$$
For example, by looking at $T_n(\theta)=\sqrt{n}(\bar{X}_n-\theta)/S_n$ we obtain the one sided interval of the form $(L_n,\infty)$ and 
$$L_n=\bar{X_n}-(S_n/\sqrt{n})\,\bigg(c\sqrt{N/n}-\sqrt{N/n-1}\Phi^{-1}(1-\gamma)\bigg).$$
It is also straightforward to find the two sided interval of size $1-\alpha$.

\subsection{Nonparametric test and RCI}

To highlight the extent of the predictive approach we provide an illustration of a nonparametric test and pointwise RCI. These are based on the Kolmogorov-Smirnov test. 

\subsubsection{Test} Here we aim to demonstrate by 
considering the test $H_0:F=F_0$ versus $H_0:F\ne F_0$ for some distribution function $F_0$.
The statistic $T_N$ can, for example, be the Kolmogorov-Smirnov test statistic, i.e.
$T_N=\sup_x|F_N(x)-F_0(x)|$
where $F_N$ is the empirical distribution function. The null hypothesis is rejected if $T_{N}\geq c$. 
To obtain $Q_n$ we first define 
$T_N^{(n)}=\sup_x|F_N^{(n)}(x)-F_0(x)|$
where
$$F_{N}^{(n)}(x)=\frac{nF_n(x)+\sum_{i=n+1}^N 1(X_i'\leq x)}{N}$$
and the $X'_{n+1:N}$ are i.i.d. from $F_0$.

We now define $Q_n=P(T_N\geq c\mid \mathcal{F}_n)=P(T_{N}^{(n)}\geq c)$ which can be computed via Monte Carlo simulation. That is, for $m=1,\ldots,M$, for a large $M$, we generate $(T_{N}^{(n)}(m))_{m=1:M}$  and estimate
$$Q_n=M^{-1}\sum_{m=1}^M 1\bigg(T_N^{(n)}(m)\geq c\bigg).$$
This describes the test and the null hypothesis will be rejected if ever $Q_n\geq \gamma$ once $c$ and $\gamma$ have been chosen appropriately to ensure a Type I error of $\alpha$. 

\subsubsection{RCI}

We restrict attention to pointwise confidence intervals for the unknown true distribution $F(x)$.
First, write
$$T_N^{(n)}(x)=|w_n\Delta_n(x)+(1-w_n)\Delta'_{N-n}(x)|$$
where $w_n=n/N$ and $\Delta_n(x)=F_n(x)-F(x)$ and
$$\Delta'_{N-n}(x)=(N-n)^{-1}\sum_{i=n+1}^N 1(X'_i\leq x)-F(x),$$
where the $(X'_i)$ should come from $F$. We obviously can not achieve this and the aim is to approximate
a random sample of a $\Delta'_{N-n}(x)$ by $D(x)=|(N-n)\sum_{i=n+1}^{N}(X_i'\leq x)-F_n(x)|$
where the $(X'_i)$ are i.i.d. from $F_n$, the empirical distribution function of the $X_{1:n}$.
Hence, we have samples of $T_{N}^{(n)}(x)$ given by
$$|w_n\Delta_n(x)+(1-w_n)D_j(x)|$$ for $j=1,\ldots,M$, for some large $M$. We now for each $n$ perform a numerical search to find the value of $c_n(x)=\Delta_n(x)$ such that
$$M^{-1}\sum_{j=1}^M 1\left(|w_nc_n(x)+(1-w_n)D_j(x)|\geq c\right)=\gamma.$$
The RCI for $F(x)$ is then given by $F_n(x)\pm c_n(x)$.

\begin{figure}
\center
\includegraphics[width=12cm,height=6cm]{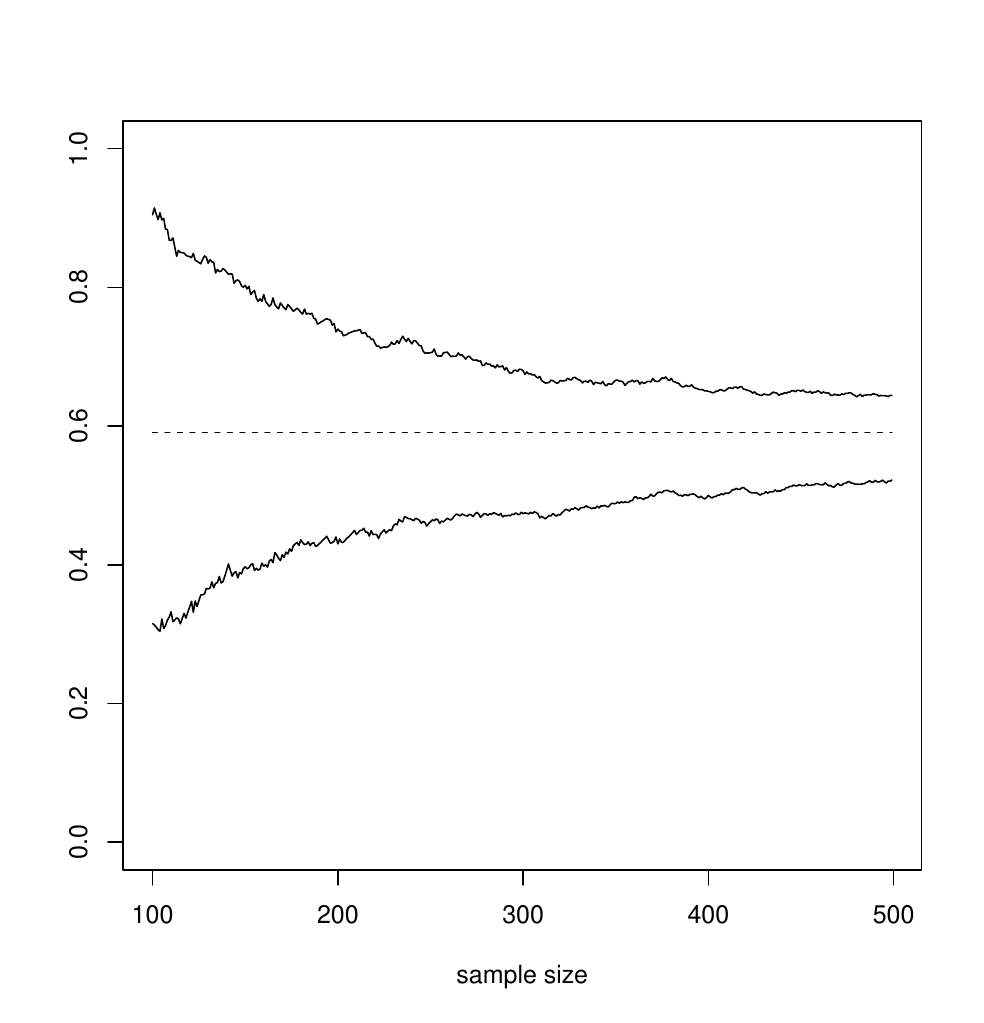}
\caption{Two sided RCI for $F(0.23)$ using $c=0.06$ and $\gamma=0.5$. The dotted line is the true value $\Phi(0.23)$}
\label{fig8}
\end{figure}

An an illustration, we find the RCI for $F(x)$ at $x=0.23$ when the true $F$ is the standard normal distribution function, and $F(0.23)=0.591.$
The RCI is presented in Fig.~\ref{fig8}, where we took the value of $c=0.06$ and $\gamma=0.5$.

\section{Brownian motion based tests and RCI}
\label{sec:Brownian}

A number of the illustrations appearing in \cite{Jennison89} can be attributed to the sequence of statistics $(S_k)_{k=1:K}$ being embedded within approximate Brownian motion. Here we use $S$ to represent a statistic and reserve $T$ and $t$ for times. So, for time points $(t_k)$ it is that $S_k=B(t_k)$ where $B$ denotes Brownian motion.
Hence, $B$ is Markov,  for any $t$, $B(t)$ is normal with mean $0$ and variance $t$, and for $t>s$ it is that $B(t)=B(s)+Z_{t-s}$
where $Z_{t-s}$ is a normal random variable with 0 mean and variance $t-s$ and is independent of $B(s)$. 
They then use multivariate normal theory to obtain critical values for setting tests and constructing RCI.
For some experiments there is the possibility that information arrives over time and appears at any time. If the sequence of statistics can be equivalent to a time change of Brownian motion, then we can define
$$Q_T(t)=P(S(T)\in C\mid \mathcal{F}_t)$$
which can be shown to be a martingale. 
Hence, we can construct tests and construct continuous time RCI via the $Q_T(t)\geq\gamma$ and $\{Q_T(t)\leq\gamma\}$, respectively. We also know, given the connection with Brownian motion how to predict $S(T)$ in terms of $S(t)$.

The prototype example is the accumulation of i.i.d. observations estimating a normal mean with known variance, say $\sigma=1$.
Then 
$S(n)=\sum_{i=1}^n X_i-n\theta$
can be embedded in Brownian motion since $E(S(n))=0$, $\mbox{Var}(S(n))=n$ and $\mbox{Cov}(S(n),S(m))=\min\{n,m\}$.
Here $\theta$ represents a hypothesized mean for a hypothesis test and represents the true mean for a RCI.

Hence we assume the sequence of statistics, suitably transformed to be standardized, approximate Brownian motion.
We now have a generic approach to construct tests and RCI. So define, for some $V(t)_{0<t\leq T}$,
$$Q_T(t)=P\bigg(B(V(T))\in C\mid B(V(t))\bigg)$$
which if $C=(c,\infty)$ becomes
$$Q_T(t)=1-\Phi\left(\frac{c-B(V(t))}{\sqrt{V(T)-V(t)}}\right).$$
The $V(T)$ is predicted using $V(t)$; i.e. $V(T)=(T/t)\,V(t)$ which turns out to be an excellent approximation, in that typically $V(t)=t\,H(\theta)$ for some function $H$ of the parameter which can be estimated from the data.

The corresponding test can be rejected if ever
$Q_T(t)\geq\gamma$ which is equivalent to 
$$B(V(t))\geq c-\sqrt{V(T)-V(t)}\,\,\Phi^{-1}(1-\gamma)$$
for any $0<t\leq T$. For the continuous time RCI, if
$B(V(t))=(S(t)-\theta/)W(t)$
then 
the one sided sequence of RCI is given by
$(I(t),\infty)$ where
$$I(t)=S(t)-W(t)\{c-\sqrt{V(T)-V(t)}\Phi^{-1}(1-\gamma)\}.$$
For the two sided test, so now $C=(-\infty,c)\cup (c,\infty)$, we have
$$Q_T(t)=\Phi\left(\frac{-c-B(V(t))}{\sqrt{T-t}}\right)+1-\Phi\left(\frac{c-B(V(t))}{\sqrt{V(T)-V(t)}}\right).$$
The two sided RCI, with $B(V(t))=(S(t)-\theta)/W(t)$, 
is given by
$I(t)=S(t)\pm W(t)\,D(t)$
where $D(t)$ is the positive solution of $x$ to
$$\Phi\left(\frac{-c-x}{\sqrt{V(T)-V(t)}}\right)+1-\Phi\left(\frac{c-x}{\sqrt{V(T)-V(t)}}\right)=\gamma.$$

\noindent
In subsections~\ref{sec:survival} and \ref{sec:binary} observations arise in continuous time and we exploit the connection with Brownian motion to develop continuous time RCI.

\subsection{Survival data}\label{sec:survival}

In this section we follow the set up of \cite{Jennison89} in the section of the same title: survival data. The underlying idea is that there are two groups $A$ and $B$ and for each interim time $t$ the pooled death times are $\tau_1<\tau_2<\cdots<\tau_{d(t)}$ which represent time from entry to death. Here $d(t)$ is the pooled number of deaths by time $t$. The assumption is that for those in group $A$ have unknown hazard function $h(t)$, while those in group $B$ have have hazard function $e^\theta h(t)$ for some $\theta$.
The aim is to provide a test and provide RCI for $\theta$. 

To this end, for each $\tau_i(t)$, define $r_{i1}(t)$ and $r_{i2}(t)$ to be those from group $A$ and $B$, respectively, to have survived a time $\tau_{i}(t)$. 
The log-rank statistic, originally considered by \cite{Peto1972}, is given by
$$L(t)=\sum_{i=1}^{d(t)}\left\{\frac{r_{i1}(t)}{r_{i1}(t)+r_{i2}(t)}-\delta_i(t)\right\}$$
where $\delta_i(t)=1$ if $i$ belongs to group $A$ and $\delta_i(t)=0$ if $i$ belongs to group $B$.

As explained in \cite{Jennison89}, under certain conditions,
$$S(t)=\frac{L(t)-\theta d(t)/4}{\sqrt{d(t)/4}}$$
is approximately marginally standard normal for each $t$.
When the conditions for the approximate normality for $L(t)$ fail, \cite{Jennison89} propose an alternative statistic which is
$S(t)=L(t)/\sqrt{V(t)}$  where now
$$L(t)=\sum_{i=1}^{d(t)}\left\{\frac{r_{i1}(t)}{r_{i1}(t)+e^\theta\,r_{i2}(t)}-\delta_i(t)\right\}$$
and 
$$V(t)=\sum_{i=1}^{d(t)}\frac{r_{i1}(t)\,r_{i2}(t)\,e^\theta}{(r_{i1}(t)+e^\theta r_{i2}(t))^2}.$$
The marginals for $S(t)$ are approximately standard normal. 
It is the process $L(t)$ which can be approximated by time change Brownian motion, i.e. $L(t)=B(V(t))$.

\subsection{Binary data}\label{sec:binary}

Here the data are based on the completion of $N(t)$ trials by time $t$ and the cumulative successes are $X(t)$, where each success is assumed to occur with probability $\theta$.
Then \cite{Jennison89} point out that
$$S(t)=\frac{X(t)-\theta N(t)}{\theta(1-\theta)}$$
is approximately Brownian motion with time change; i.e. $$S(t)=B\left(\frac{N(t)}{\theta(1-\theta)}\right).$$
We can proceed by defining $Q_T(t)=P(|S(T)|\geq c\mid\mathcal{F}_t)$ which is a martingale and therefore leads to both tests and the construction of continuous time RCI.

A specific application involves two groups in which $\theta$ is now the log odds ratio of two success probabilities $p_A$ and $p_B$, i.e. $\theta=\log\{p_A(1-p_B)/(p_B(1-p_A)\}$.
To complete the notation $X(t)$ is the cumulative successes for group $A$ from $N(t)$ trials and $Y(t)$ is the cumulative successes for group $B$  from $M(t)$ trials.
Then according to \cite{Jennison89} it is that
$S(t)=\widehat\theta(t)-\theta/\sqrt(V(t))$
is approximately Brownian motion, i.e. $S(t)=B(1/V(t))$ where
$$\exp\{\widehat\theta(t)\}=\frac{X(t)(M(t)-Y(t))}{Y(t)(N(t)-X(t))}$$
and
$$V(t)=\frac{1}{X(t)}+\frac{1}{N(t)-X(t)}+\frac{1}{Y(t)}+\frac{1}{M(t)-Y(t)}.$$




\subsection{Applications}

Here we consider two types of statistic, also considered by \cite{Jennison89}; the MLE statistic and the score test statistic. In the former the test statistic is
$T_n(\theta)=n(\widehat\theta_n-\theta)/\sqrt{V_n(\theta)}$
where $\widehat\theta_n$ is the MLE and $$V_n(\theta)=n^2\,\mbox{Var}(\widehat\theta_n)\approx n/I(\theta),$$ where $I(\theta)$ is the Fisher information. 
For the latter, 
$T_n(\theta)=U_n(\theta)/\sqrt{V_n(\theta)}$
where 
$U_n(\theta)=(\partial/\partial\theta)\,\log l_n(\theta)$
and $l_n(\theta)$ is the log-likelihood function and $V_n(\theta)=n\,I(\theta)$. 

The aim given $X_{1:n}$ is to predict $T_N(\theta)$. We do this universally using
\begin{equation}\label{eq:predictT}
T_{N}(\theta)=\sqrt{n/N}\,T_n(\theta)+\sqrt{(N-n)/N}\,Z
\end{equation}
where $Z$ is an independent standard normal random variable.
We find the positive solution $\xi_n$ to
$$1-\Phi\left(\frac{c\sqrt{N}-\sqrt{n}\xi_n}{\sqrt{N-n}}\right)+\Phi\left(\frac{-c\sqrt{N}-\sqrt{n}\xi_n}{\sqrt{N-n}}\right)=\gamma$$
where $c=\Phi^{-1}(1-\alpha\gamma/2)$
which provides the two sided RCI of size $1-\alpha$ given by
$$I_n=\{\theta:\,|T_n(\theta)|\leq \xi_n\}.$$
In both cases, as pointed out in \cite{Jennison89}, 
the covariance between $T_n$ and $T_{N}$ is $\sqrt{V_n/V_{N}}$ which is $\sqrt{n/N}$. Hence, we can motivate (\ref{eq:predictT}) as the prediction of $T_N$ given $T_n$. It also maintains the correct marginal distribution, that is, if $T_n$ is standard normal then so is $T_N$.

So, in the following two subsections we will be looking at two types of statistics. 
The first is based on an MLE and the process which can be embedded in Brownian motion is
$S^*_n(\theta)=(\widehat\theta_n-\theta)/V_n(\theta)$, where $V_n(\theta)=\mbox{Var}(\widehat\theta_n)$, and
$S^*_n=B(1/V_n)$.
The second is a score statistic, $S_n(\theta)=U_n(\theta)/\sqrt{V_n(\theta)}$
where $U_n(\theta)=(\partial/\partial\theta)\log l_n(\theta)$
and $V_n(\theta)=nI(\theta)$ with $I(\theta)$ being the Fisher information. It is $U_n(\theta)$ which can be embedded in Brownian motion with a time change, i.e. $U_t(\theta)=B(tI(\theta))$.

\subsubsection{MLE statistic}

The example we consider appears in section 5.2 in \cite{Jennison89}
and involves an odds ratio. Bernoulli trials from two arms $A$ and $B$ proceed over time and at interim analysis $k$ the number of completed trials is $n(k)$ with $X(k)$ successes for arm $A$ and $m(k)$ trials with $Y(k)$ successes for arm $B$. 
If $p_A$ and $p_B$ are the unknown Bernoulli parameters then the MLE for $\theta=\log(p_A(1-p_B)/(p_B(1-p_A))$ is
$$\widehat\theta_k=\log\left\{\frac{n(k)(m(k)-Y(k))}{m(k)(n(k)-X(k))}\right\}.$$

\begin{figure}
\center
\includegraphics[width=12cm,height=6cm]{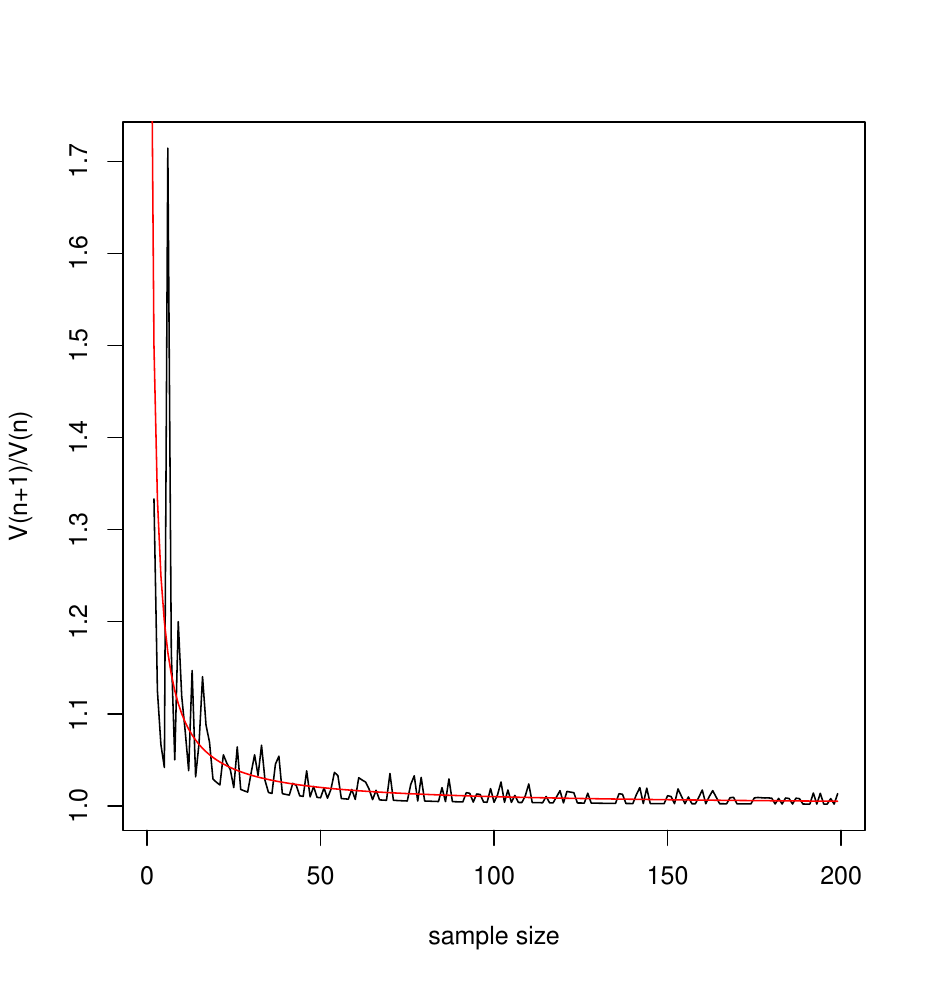}
\caption{Plot of $V_{n+1}/V_n$ (bold black line) alongside $(n+1)/n$ (red line))}
\label{fig9}
\end{figure}

\begin{figure}
\center
\includegraphics[width=12cm,height=6cm]{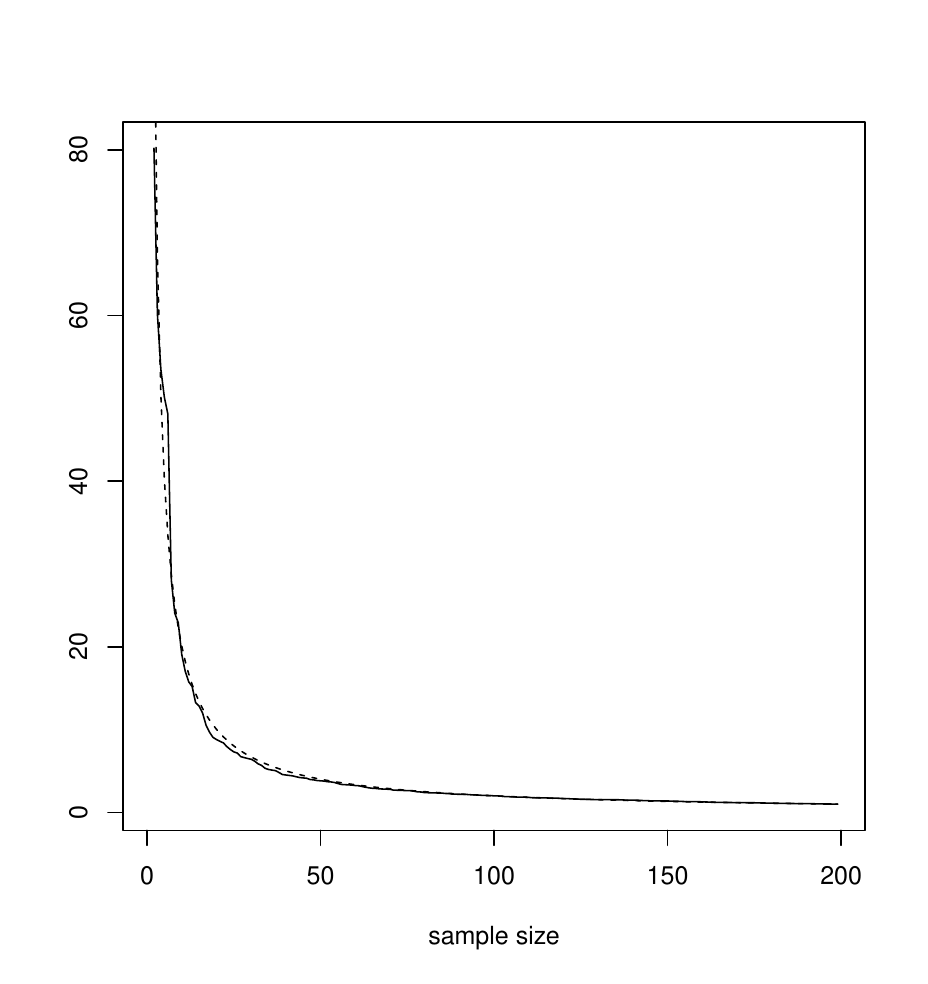}
\caption{P{lot of $V_N/V_n$ (bold line) alongside $N/n$ (dotted line)}}
\label{fig10}
\end{figure}

\begin{figure}
\center
\includegraphics[width=12cm,height=6cm]{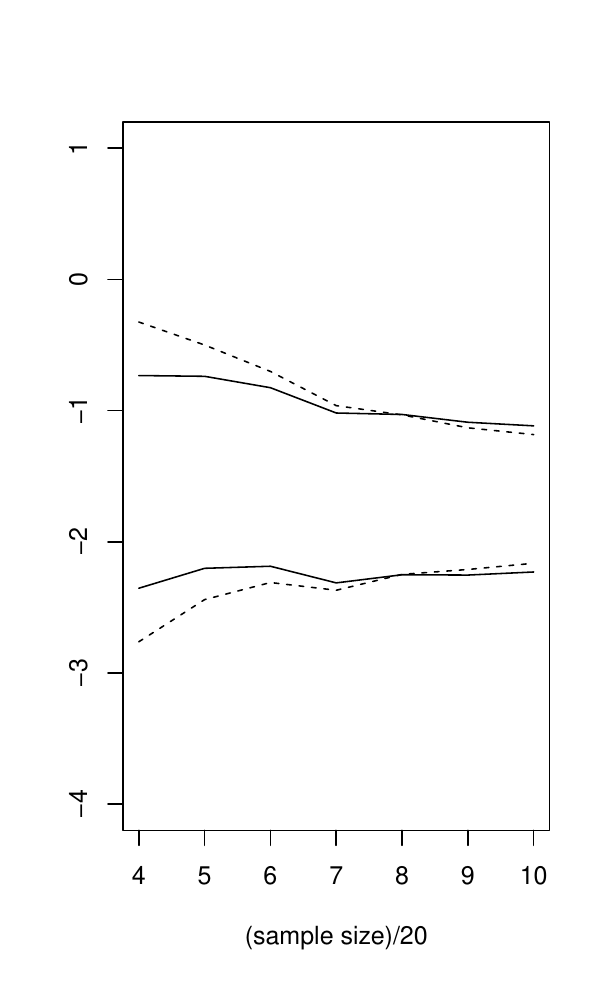}
\caption{Comparison of RCI for GSM approach (bold lines) with martingale RCI (dotted lines) for odds ratio model with sample size of 200 and an interim analysis every 20 observations. The true value of the parameter is -1.69}
\label{fig11}
\end{figure}

\begin{figure}
\center
\includegraphics[width=12cm,height=6cm]{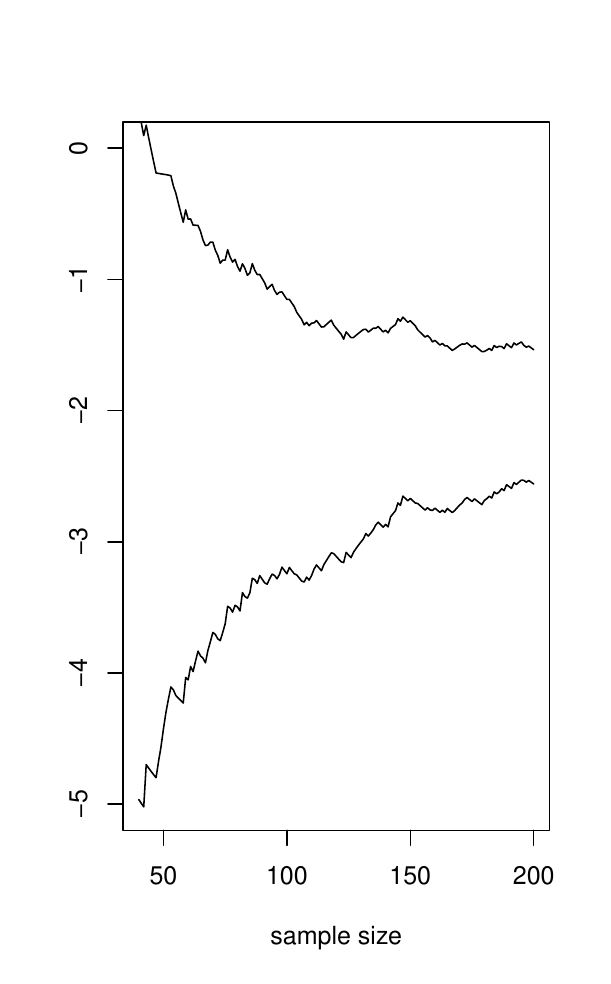}
\caption{Martingale RCI for all samples for the log odds parameter}
\label{fig12}
\end{figure}

\noindent
Further, the variance is estimated by
$$\widehat{V}_k=1/n(k)+1/m(k)+1/(X(k)-n(k))+1/(Y(k)-m(k)).$$
We first look at the prediction of $T_N$ given $T_n$. We took $N=200$ and $n=100$ and $p_A=0.3$ and $p_B=0.7$. We fix the statistic at $n=100$ and in one case we sample a further 100 outcomes and compute the $T_N$ from data and in the other case we predict $T_N$ using (\ref{eq:predictT}). We repeat this 1000 times and record the variances from the two cases. The fixed statistic at $n=100$ is 0.081. The variance of $T_N$ from further experiments is 0.513 and the variance of $T_N$ from the predictions is 0.512.

To look at the behavior of the variances, notably that $V_n/V_N=n/N$ we show a plot of $(V_{n+1}/V_{n})$ in Fig.~\ref{fig9},
with the ratio of variances the bold line and $((n+1)/n)$ as the red line. In Fig.~\ref{fig10} we show the sequence $(V_N/V_n)$ (bold line) along side the sequence $(N/n)$ (dotted line).

In Fig.~\ref{fig11} we present a comparison of two sided RCI of size 0.95. There are 10 interim analyses and a total of 200 samples. We show from the 4th interim time onwards. The bold lines are the \cite{Jennison89} RCI and the dotted lines are the martingale RCI. The true value of $\theta$ is -1.69 with $p_A=0.3$ and $p_B=0.7$. For the martingale RCI we took $\gamma=0.5$.
The martingale approach is not restricted to group interim analysis and Fig.~\ref{fig12} show the RCI for all samples from 40 onwards.

\subsubsection{Score statistic}

\begin{figure}
\center
\includegraphics[width=12cm,height=6cm]{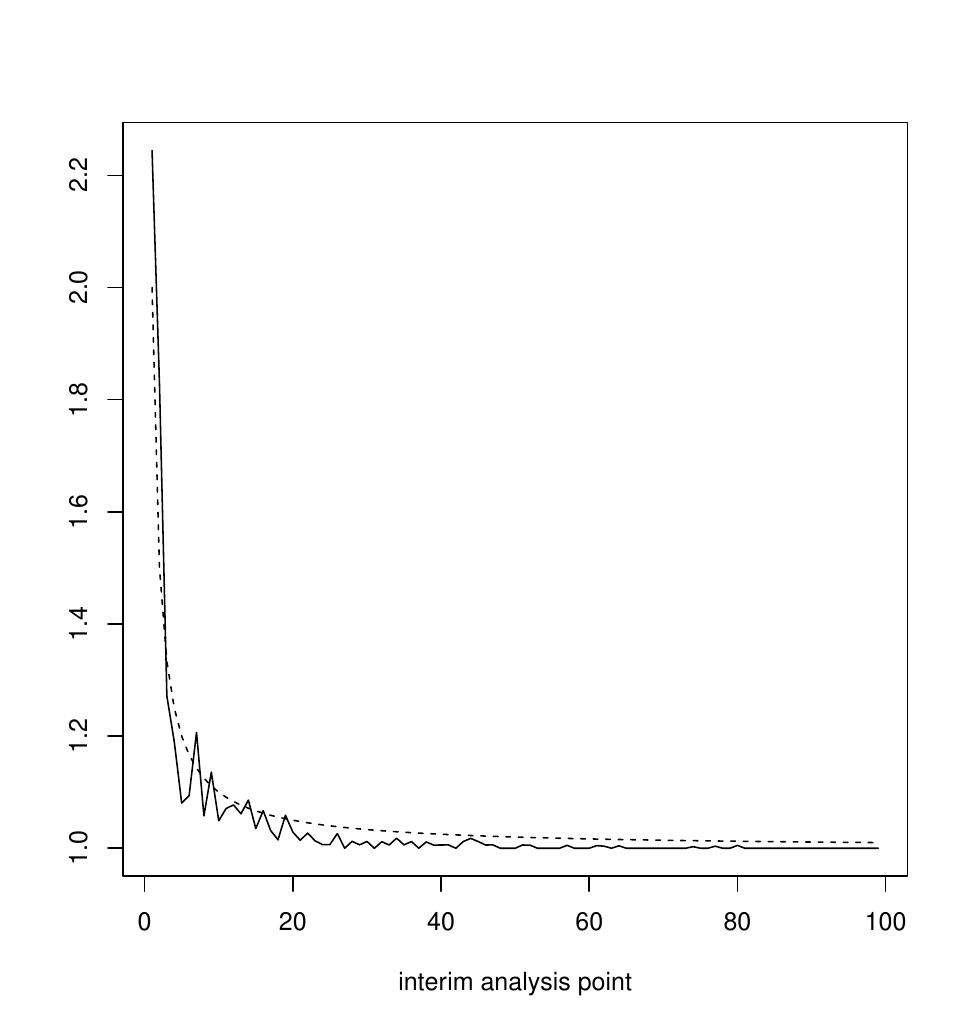}
\caption{Plot of $V_{k+1}/V_k$ and $(1+k)/k$ for the score statistic based on the partial likelihood model}
\label{fig14}
\end{figure}

\begin{figure}
\center
\includegraphics[width=12cm,height=6cm]{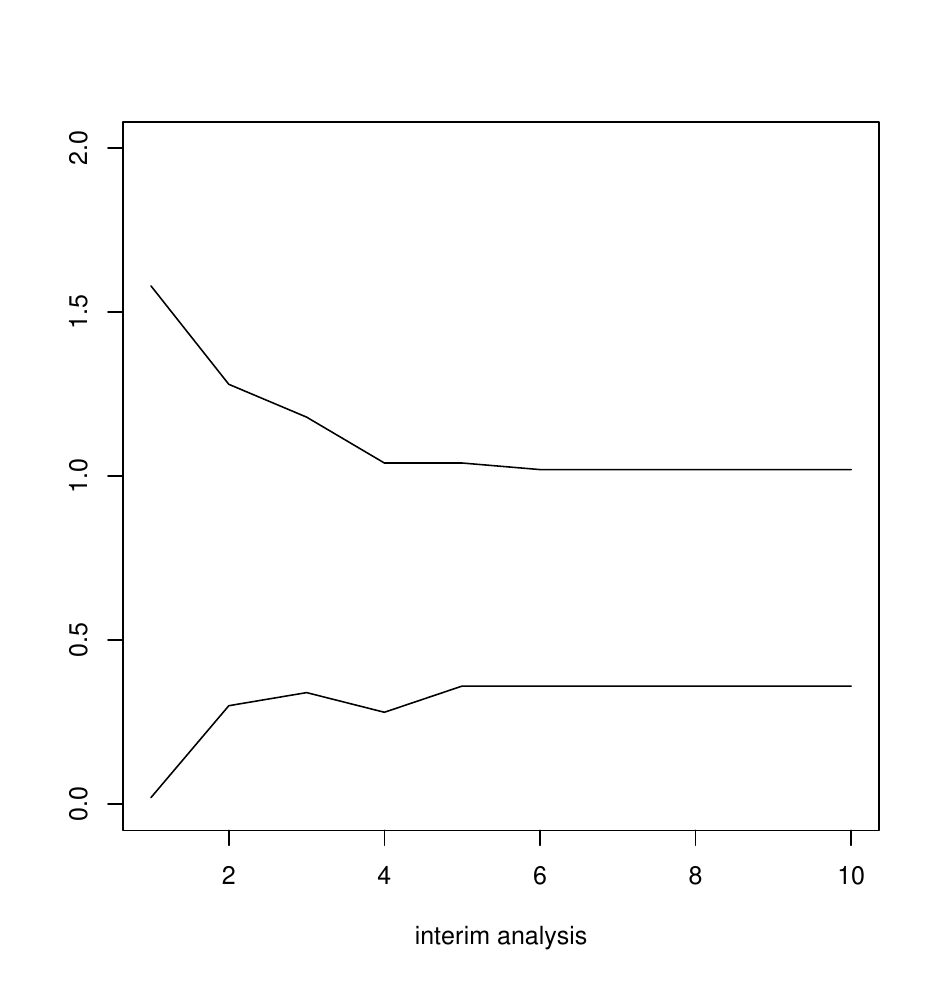}
\caption{Martingale RCI for all samples for the partial likelihood parameter}
\label{fig13}
\end{figure}

The example here appears in section 4.1.2 in \cite{Jennison89} and involves survival data with the partial likelihood of \cite{Cox1972}.
The sequence of score statistics are detailed in section 5.1.1 and
for each $k$, $T_k(\theta)=L_k(\theta)/\sqrt{V_k(\theta)}$ at time $t_k$. We first look at the sequence of variances $V_k$ behaving as $V_{k+1}/V_k=1+1/k$. See Fig.~\ref{fig14}.

We took 100 interim analysis points spaced equally from 0 to 5.8. We computed the variances at the true value of $\theta$ which is 0.69. A total of 200 observations were generated, from the exponential(1) distribution for arm A and exponential(2) for arm B.  Two sided RCI for the same setting with 10 interim points of analysis are shown in Fig.~\ref{fig13}.

\section{Boostrap tests}
\label{sec:bootstrap}

The bootstrap (\cite{Efron79}) is a resource for testing a hypothesis when the distribution of the sample is unknown. The illustration we use throughout is for a mean of a sample and the distribution of the sample is unknown.
The strategy for conducting a bootstrap test, $H_0:\mu=0$ versus $H_0:\mu>0$, with data $X_{1:m}$,  is as follows:

\begin{description}

\item 1. Center the data to $\widetilde{X}_{1:m}=X_{1:m}-\bar{X}_m$ and denote by $F_m$ the empirical distribution function of $\widetilde{X}_{1:m}$.

\item 2. Sample $X^*_{1:m}$ i.i.d. from $F_m$ and construct $T=\bar{X^*}_m$.

\item 3. Repeated 2. and the use of Monte Carlo methods, allows for an estimator of $c$, for which $P(T\geq c)=\alpha$ for some chosen Type I error $\alpha$.

\item 4. The null hypothesis is rejected if $\bar{X}_m\geq c$.

\end{description}

\noindent
This is a well known and established procedure when the distribution or the family of distribution of the sample is unknown.
We now show how to adapt this test to an anytime valid test. The key starting point remains the same, before any progress can be made, including the critical value, an observed sample of size $m$ is required.   


\subsection{Anytime valid bootstrap test}

As usual we make assumptions about the sample size which is that there is a maximal sample size $N$. 
We also assume that there is a sample size $m$, sufficiently large to allow for an adequate estimator of an empirical distribution function. Throughout we will use $m=100$ and $N=500$. Hence, early stopping by rejecting the null hypothesis is allowed for $m<n\leq N$.

In constructing the test we follow the principles behind the use of the Efron bootstrap. This is that statistics of interest are sampled by treating the empirical distribution function of an observed sample as a stand in for the true distribution function.
The first action, for a test of the hypothesis $\mu=0$, is to center the observed  data of size $m$; i.e. 
$$F_m(x)=m^{-1}\sum_{i=1}^m 1(\widetilde{X}_i\leq x)$$
where $\widetilde{X}_i=X_i-\bar{X}_m$, $i=1,\ldots,m$.
Hence, from this, a bootstrapped dataset of size $N$ can be generated by taking $X^*_{1:N}$ to be i.i.d. from $F_m$.
Multiple independent such samples are what can be used to construct critical values for the anytime valid test.

From one such $X^*_{1:N}$ we compute $\bar{X^*}_N$. We first find $c$ for which $P(\bar{X^*}_N\geq c)=\widetilde{\alpha}$ for some $\widetilde{\alpha}<\alpha$ where $\alpha$ is the designated overall type I error. For example, $\widetilde{\alpha}=0.95\alpha$. It is easy to find $c$ using the bootstrapped samples of $\bar{X^*}_N$.

We need to bootstrap to determine the $\gamma$ ensuring an overall type I error of $\alpha$. That is, we require
$P(Q\geq\gamma)=\alpha$ where $Q=\max_{m<n\leq N}\{Q_n\}$.
We estimate $Q_{m+1:N}$ using bootstrapped samples $X^*_{1:N}$. To elaborate let us fix an $n$ between $m$ and $N$.
We generate multiple copies of $T_N^{(n)}$ given by
$$\frac{\sum_{i=1}^n X_i^*+\sum_{i=n+1}^N X^{**}_i}{N}$$
where the $X^*_{1:n}$ are i.i.d. from $F_m$ and the $X^{**}_{n+1:N}$ are i.i.d. from $F^*_n$ which is the empirical distribution function of $\widetilde{X}^*_i=X_i^*-\bar{X^*}_n$,
for $i=1,\ldots,n$.  From these multiple copies we can obtain $Q_n$. This can clearly be extended to obtain $Q_{m+1:N}$ from which we can derive a single $Q$. Repeated, starting with new multiple bootstrapped $X^*_{1:N}$, gives us multiple  $Q$'s from which we can find $\gamma$.

In order to complete the test we use as the sequence of statistics of interest the probability, under the null hypothesis, that a predicted value for $\bar{X^*}_N$ given a sample size of $n>m$, exceeds $c$. The predicted $\bar{X^*}_N$ under the null hypothesis is given by
$$T_N^{(n)}=\frac{\sum_{i=1}^n X_i+\sum_{i=n+1}^N X_i^*}{N},$$
where the $X^*_{n+1:N}$ are i.i.d. from $F_n$, the empirical distribution function of the $\widetilde{X}_i=X_i-\bar{X}_n$, for $i=1,\ldots,n$.
Then 
$Q_n=P(T_N^{(n)}\geq c)$.
We also define $Q_N=1(T_N\geq c)$.
The null hypothesis will be rejected if ever $Q_n\geq\gamma$ for some $\gamma$.


\subsubsection{Algorithm}

A bootstrap approach to anytime valid testing of a mean without knowing the distribution of the sample. Assume the test is $H_0:\mu=0$ versus $H_1:\mu>0$. The algorithm starts with an observed sample of size $m$, just as it does for the fixed sample size bootstrap test. First we describe obtaining the critical value $c$ for a chosen $\widetilde\alpha$.

\begin{description}

\item 1. Construct $\widetilde{X}_{1:m}=X_{1:m}-\bar{X}_{m}$ and take $F_m$ to be the empirical distribution of $\widetilde{X}_{1:m}$.

\item 2. Take $X^*_{1:N}$ to be i.i.d.  from $F_m$.

\item 3. Repeated simulation of 2. gives $c$ such that $F_m^{-1}(c)=\widetilde{\alpha}$. 

\end{description}

\noindent
Now we describe how to obtain the critical value $\gamma$.

\begin{description}

\item 1. Given $X^*_{1:n}$ which are i.i.d. from $f_m$, let  $F^*_n$ be the empirical distribution of $\widetilde{X}^*_{1:n}=X^*_{1:n}-\bar{X^*}_n$

\item 2. Generate  $X^{**}_{n+1:N}$  i.i.d. from $F^*_n$ and construct 
$$T_N^{(n)}=\frac{\left(\sum_{i=1}^n X^*_i+\sum_{i=n+1}^N X^{**}_i\right)}{N}.$$

\item 3. Repeated simulation of $T^{(n)}_N$ from 2.  and extending to all $m<n\leq N$ yields 
$$\left(Q_n=P(T_N^{(n)}\geq c)\right)_{n=m+1:N},$$ via Monte Carlo methods, which provides the maximum value $Q$. 

\item  4. Repeated 1., 2. and 3. provide multiple copies of $Q$ from which $\gamma$ can be found.

\end{description}

\noindent
Finally we detail how the test is implemented given a to be  observed sample $X_{1:N}$.

\begin{description}

\item 1. For a given $m<n\leq N$ construct $F_n$ which is the empirical distribution of the $\widetilde{X}_i=X_i-\bar{X}_n$, for $i=1,\ldots,n$.
 
 \item 2. Sample $X^*_{n+1:N}$ to be i.i.d. from $F_n$ and construct
 $$T_N^{(n)}=\frac{\sum_{i=1}^n X_i+\sum_{i=n+1}^N X^*_i}{N}.$$
 
 \item 3. Repeated 2. yields $Q_n$. If $Q_n\geq\gamma$ the null hypothesis is rejected.

\end{description}

\subsection{Illustration}

We take $m=100$ and $N=500$ and test for a mean with the sample coming from a normal distribution with fixed variance 1. This is so we can compare the bootstrap anytime valid test with the anytime valid test which relies on martingales. 
If it is known the sample are normal with variance 1 then it is possible to determine $Q_n$ exactly, it is given by
$$Q_n=1-\Phi\left(\frac{d\sqrt{n}-\sum_{i=1}^n X_i)}{\sqrt{N-n}}\right)$$
with $Q_n=1(\bar{X}_N\geq d/\sqrt{N})$. Here $d=\Phi^{-1}(1-\alpha\gamma)$ and the null hypothesis is rejected if ever $Q_n\geq \gamma$.
The overall type I error is given by $\alpha$. The main point here is that the $(Q_n)$ forms a martingale sequence and hence the type I error can be established using the martingale inequality. See \cite{Lan1982}.

In the figures, sample paths of $(Q_n)$ are provided; the black lines are from the anytime valid bootstrap test and the red line is from the normal model anytime valid $z$ test. The match between the two approaches is that we took $\widetilde\alpha=\alpha\gamma$.
In Fig.~\ref{fig1} the data sample is normal with mean 0 and variance 1 ($H_0$ correct) and in Fig.~\ref{fig2} the data sample is normal with mean 0.2 and variance 1 ($H_0$ not correct). As can be seen the bootstrap sample paths are a good match with those from the model based normal test.

\begin{center}
\begin{figure}[!htbp]
\begin{center}
\includegraphics[width=14cm,height=5cm]{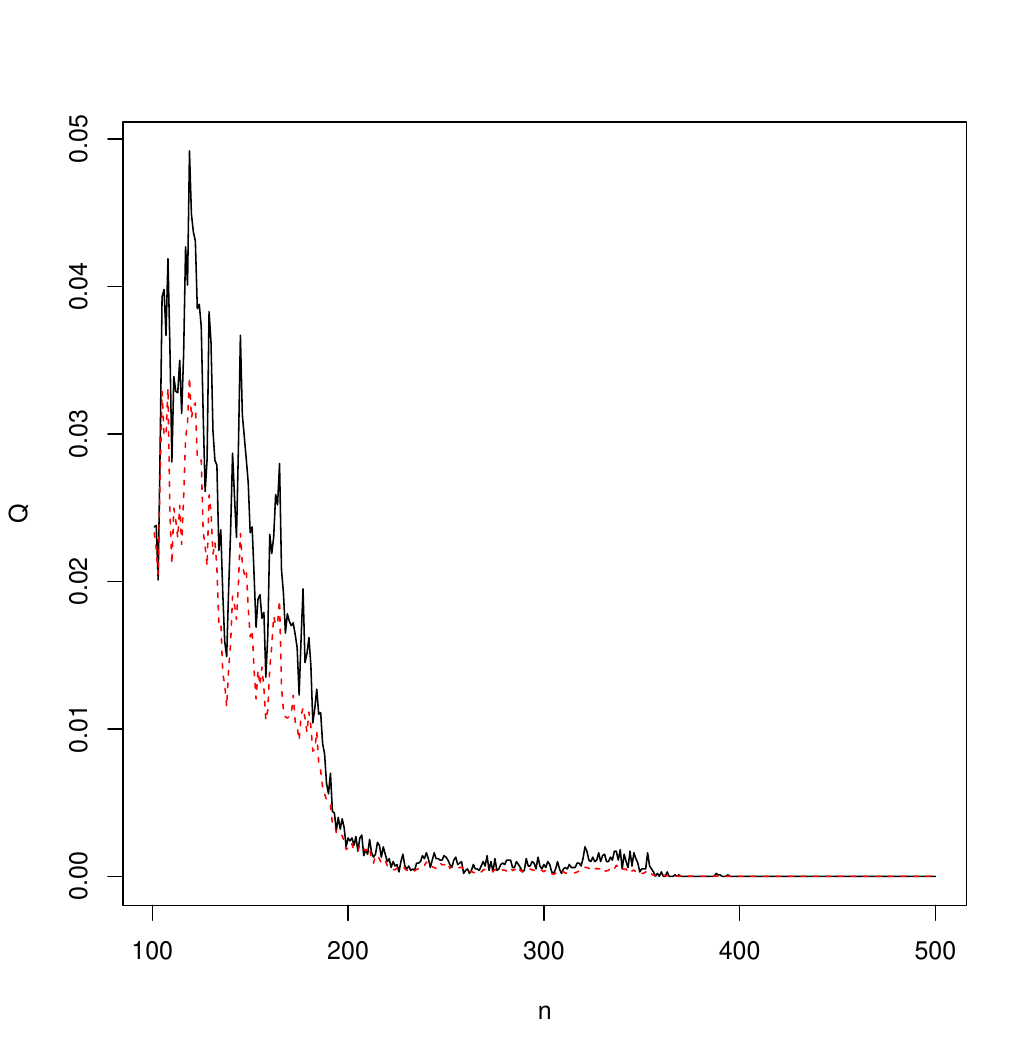}
\caption{$Q_n$ functions with data sample from normal with mean 0}
\label{fig1}
\end{center}
\end{figure}
\end{center} 

\begin{center}
\begin{figure}[!htbp]
\begin{center}
\includegraphics[width=14cm,height=5cm]{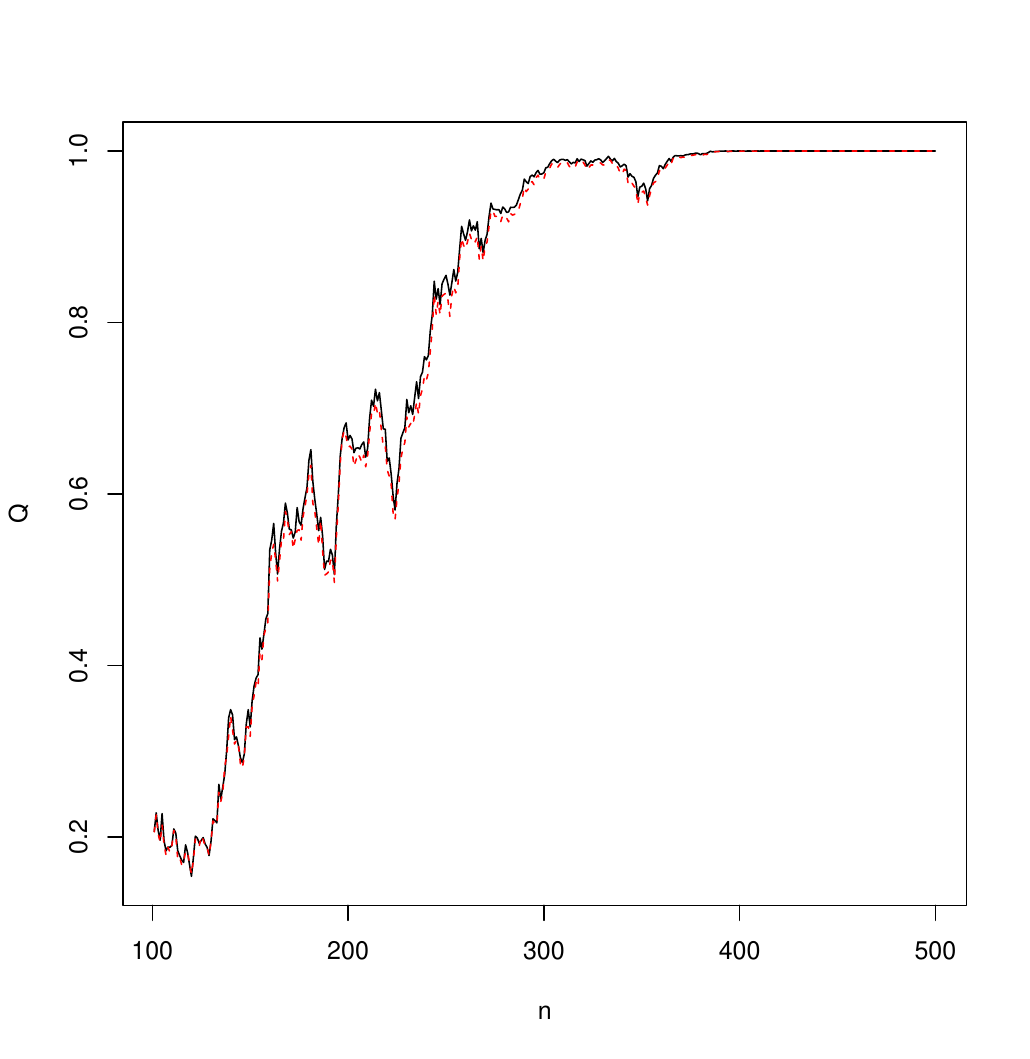}
\caption{$Q_n$ functions with data sample from normal with mean 0.2}
\label{fig2}
\end{center}
\end{figure}
\end{center}

\subsection{Real data analysis} In this section we use the dataset ``iris" which is in the R package ``datasets".
We use column 2 which consists of $N=150$ sepal width observations, comprising 50 observations each from 3 species of iris.
See \cite{Anderson35} and \cite{Fisher36} for further details. 
\begin{center}
\begin{figure}[!htbp]
\begin{center}
\includegraphics[width=14cm,height=5cm]{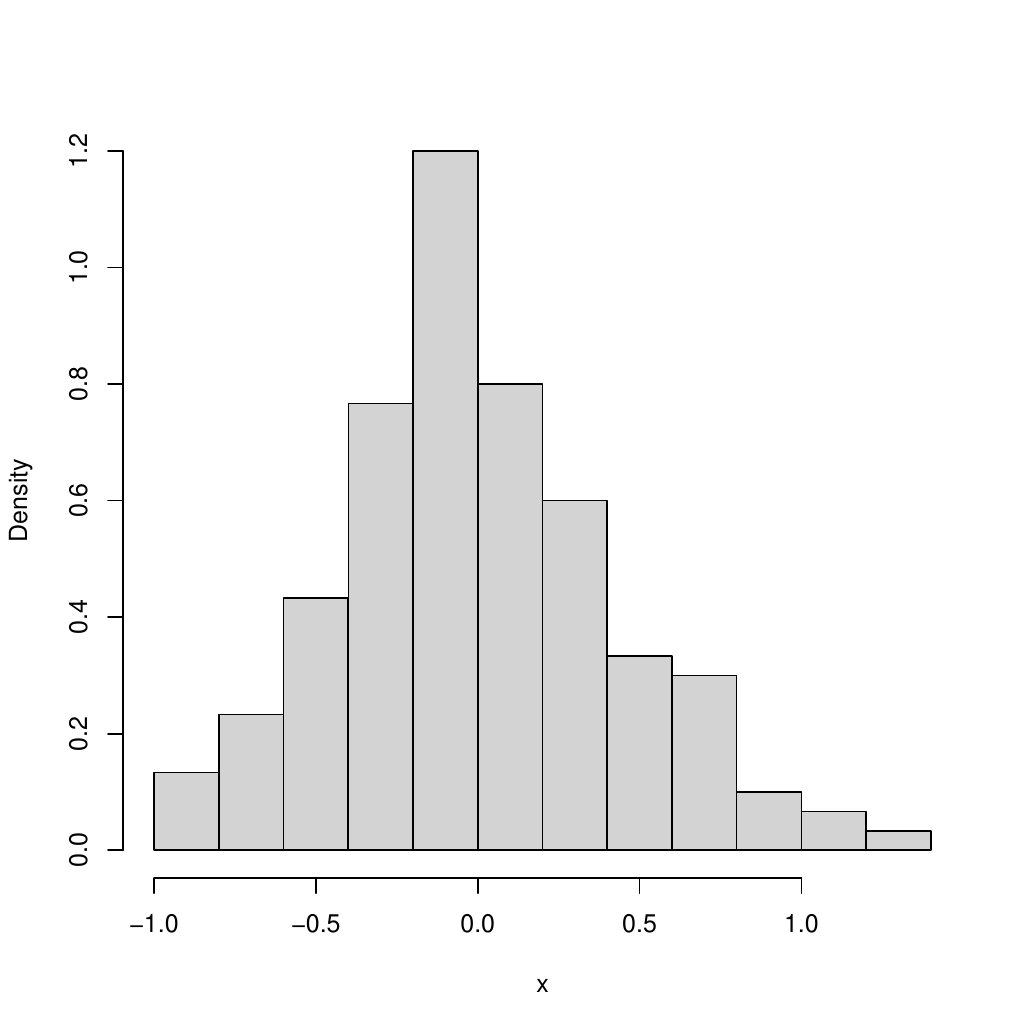}
\caption{Histogram of 150 iris sepal width data}
\label{fig3}
\end{center}
\end{figure}
\end{center} 

\begin{center}
\begin{figure}[!htbp]
\begin{center}
\includegraphics[width=14cm,height=5cm]{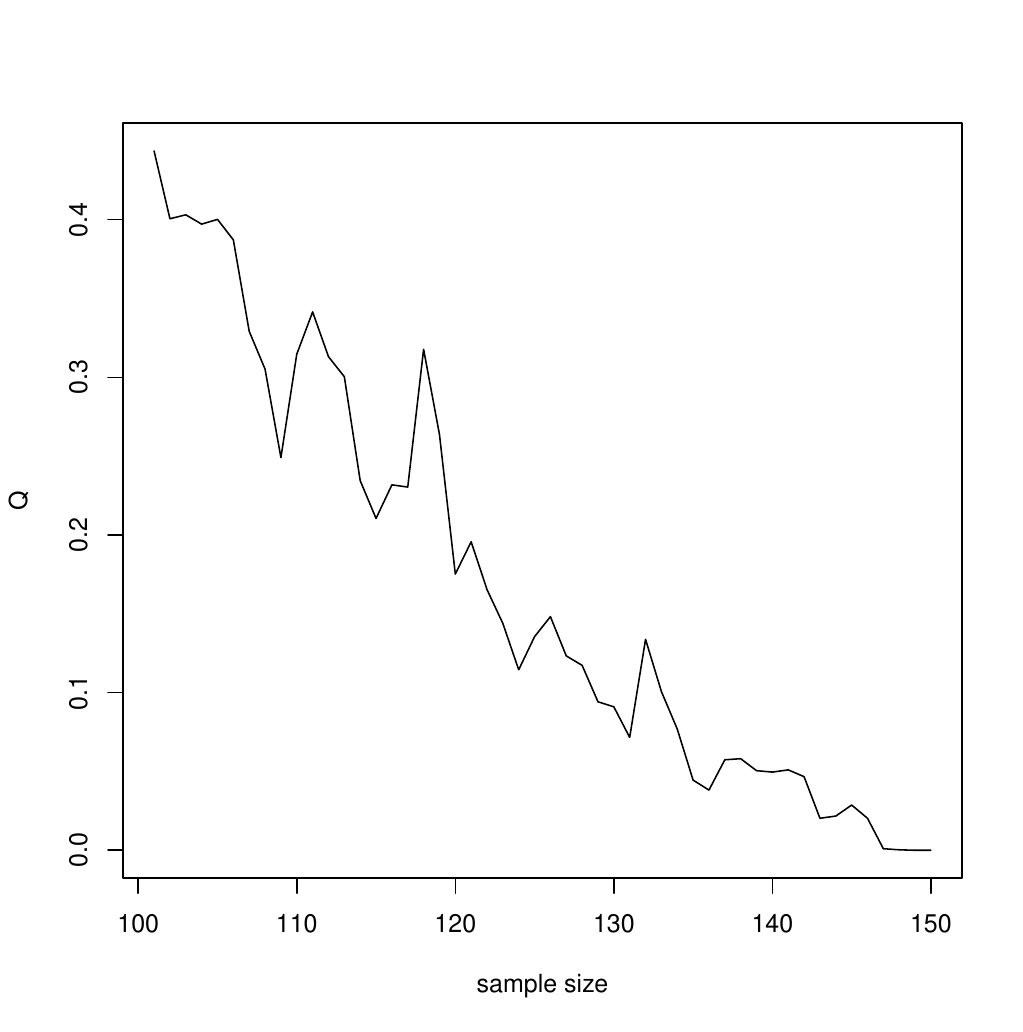}
\caption{Sample path of $Q_n$ for $m<n\leq N$ }
\label{fig4}
\end{center}
\end{figure}
\end{center}

\noindent
A histogram of the data, having subtracted 3 form each observation, is presented in Fig.~\ref{fig4}.   The hypothesis is $H_0:\mu=0$ versus $H_1:\mu>0$ and as can be seen the assumption of normality is suspect. Hence the appropriateness of the bootstrap test.
We take $m=100$ and $\alpha=0.05$. Using bootstrap sample sizes of 10,000 we obtain the value of $c$ to be 0.071. The corresponding value of $\gamma$ is given by 0.487. 

The sample path of the $(Q_n)$ for $m<n\leq N$ is presented in Fig.~\ref{fig4}. As can be seen the maximum value is smaller than $\gamma$, indeed the largest value is 0.444. Hence, the null hypothesis is not rejected. For reference, the mean of the data is 0.057.

\subsection{Parametric bootstrap test}

Here we come almost full circle.  The parametric version of the bootstrap, see \cite{Efron2012}, takes a bootstrap sample from the model with a plugged in MLE; i.e. instead of taking $X^*_{1:n}$ to be i.i.d. from $F_n$, the empirical distribution of the sample, it is taken from the parametric model $f(\cdot\mid \widehat\psi_n)$, where $\widehat\psi_n$ is the MLE. For a hypothesis test, the $F_n$ would be adapted to adhere to the null hypothesis constraint. In the parametric case this would look like a null hypothesis on a part of the parameter space, say $\psi=(\theta,\phi)$ in which case we would take $X^*_{1:n}$ i.i.d. from $f(\cdot\mid\theta_0,\widehat\phi_n)$ when the null hypothesis is $H_0:\theta=\theta_0$. This is exactly how we proceeded as in Section \ref{sec:student} with the Student $t$ test when applied to the anytime valid test.

It is worth providing the algorithm for performing an anytime valid test using the parametric bootstrap. For some suitable large $m$, estimate $\widehat\phi_m$. Further, the test statistic is represented by $T_n=T(X_{1:n})$ for a sample of size $n$.

\begin{description}

\item 1. Sample $X^*_{1:N}$ to be i.i.d. from $f(\cdot\mid\theta_0,\widehat\phi_m)$ and construct $T^*_N$ from $X^*_{1:N}$. Repeated sampling yields $\mathcal{C}_{\wt}$ for which $P(T^*_N\in \mathcal{C}_{\wt})=\wt$.

\item 2. For each $n>m$ sample $X^*_{1:n}$ to be i.i.d. from $f(\cdot\mid\theta_0,\widehat\phi_m)$ and then estimate
$\widehat\phi_n^*$ from $X^*_{1:n}$. Sample $X^{**}_{n+1:N}$ to be i.i.d. from $f(\cdot\mid\theta_0,\widehat\phi_n^*)$ and compute
$T_N^{(n)}=T(X^*_{1:n},X^{**}_{n+1:N})$. Repeated sampling gives $Q_n=P(T_N^{(n)}\in\mathcal{C}_{\wt})$ for each $n>m$ and hence $Q=\max_{1\leq n\leq N}\{Q_n\}$. This provides the $\gamma$ for which $P(Q\geq\gamma)=\alpha$.

\item 3. With observed data $X_{1:n}$ for $n>m$, sample $X^*_{n+1:N}$ to be i.i.d. from $f(\cdot\mid\theta_0,\widehat\phi_n)$. Repeated sampling
gives $Q_n=P(T(X_{1:n},X^*_{n+1:N})\in \mathcal{C}_{\wt})$. Reject the null hypothesis if $Q_n\geq\gamma$.

\end{description}

\section{Discussion}

Without any reference to martingales, which has been the foundational structure for anytime valid testing, we adopt the following sequence of strategies. 

\begin{description}

\item 1. Predict the missing data with a model $p(X_{n+1:N}\mid X_{1:n})$ for all $1\leq n\leq N$ which satisfies the constraints imposed by the null hypothesis.

\item 2. The model is used to define the predicted statistic $T_N^{(n)}$; i.e. $T_{N}^{(n)}=T(X_{1:n},X'_{n+1:N})$.
Note that $T_N=T(X_{1:N})$ and $T_N^{(n)}$ is predicting $T_N$.

\item  3. The $T_N^{(n)}$ are used to compute $Q_n=P(T^{(n)}_N\in\mathcal{C}_{\wt})$ for a given $\wt$.

\item 4. Using Monte Carlo methods it is possible to find the $\gamma$ for which $P(Q\geq \gamma)=\alpha$
where $Q=\max_{1\leq n\leq N}\{Q_n\}$.

\item 5. The null hypothesis is rejected at sample size $n$ if  $Q_n\geq \gamma$.

\end{description}

\noindent
We believe this is the simplest structure to date for constructing anytime valid tests. To elaborate, it is typically always possible to obtain a Type I error critical region $\mathcal{C}_{\wt}$ for $T_N$. For point 4., we would, using Monte Carlo methods, sample multiple
$T_N^{(n)}(b)$ for $b=1,\ldots,B$ and $n=1,\ldots,N$, estimate
$$Q_n=B^{-1}\sum_{b=1}^B 1\left(T_N^{(n)}(b)\in\mathcal{C}_{\wt}\right), \quad n=1,\ldots,N.$$
Each set of $(Q_n)$ gives $Q=\max_{n\leq N}\{Q_n\}$ from which we can, again using multiple copies of $Q$, find $\gamma$ such that
$P(Q\geq\gamma)=\alpha$.

\bibliographystyle{plainnat}
\bibliography{adapt}




\nocite{*}
\bibliographystyle{plainnat}
\bibliography{adapt}

\end{document}